\documentclass[11pt]{article}
\usepackage{amsmath}
\usepackage{graphicx}
\usepackage{enumerate}
\usepackage{natbib}
\usepackage{url} 
\newcommand{\blind}{1}
\usepackage{hyperref}
\usepackage[title]{appendix}
\usepackage{bm}
\usepackage{color}
\usepackage{array}
\usepackage{verbatim}
\usepackage{amsmath}
\usepackage{amssymb}
\usepackage{mathtools}
\usepackage{setspace}
\usepackage{esint}
\usepackage[ruled,vlined,noend]{algorithm2e}
\usepackage{enumitem}
\setlist[itemize]{leftmargin=*}
\DeclareMathOperator*{\argmin}{arg\,min}  
\makeatletter
\usepackage{amsthm}
\usepackage{mathabx}
\usepackage{url}
\usepackage{float}
\usepackage{enumerate}
\usepackage{tikz}
\usetikzlibrary{matrix}
\usepackage{listings}
\usepackage{multirow}
\usepackage{color} 
\usepackage{cancel}
\definecolor{mygreen}{RGB}{28,172,0} % color values Red, Green, Blue
\definecolor{mylilas}{RGB}{170,55,241}
\usepackage{subcaption}
\newcommand\given[1][]{\:#1\vert\:}

\newcommand{\independent}{\mathrel{\text{\scalebox{1.07}{$\perp\mkern-10mu\perp$}}}}

\newcommand{\notindependent}{\cancel{\mathrel{\text{\scalebox{1.07}{$\perp\mkern-10mu\perp$}}}}}

\newenvironment{talign*}
{\let\displaystyle\textstyle\csname align*\endcsname}
{\endalign}
\newenvironment{talign}
{\let\displaystyle\textstyle\align}
{\endalign}

\usepackage{titlesec}

\makeatletter
  \g@addto@macro \normalsize{%
   \setlength{\abovedisplayskip}{5pt}%
\setlength{\belowdisplayskip}{5pt}%
\setlength{\abovedisplayshortskip}{5pt}%
\setlength{\belowdisplayshortskip}{5pt}}%
\makeatother

\usepackage{natbib}
\setcitestyle{authoryear,open={(},close={)}, aysep{,}} 

\newtheorem{theorem}{Theorem}[section]
\newtheorem{lemma}{Lemma}[section]
\newtheorem{definition}{Definition}[theorem]

\newtheorem{proposition}{Proposition}[theorem]
\newtheorem{remark}{Remark}[section]
\newtheorem{assumption}{Assumption}[theorem]
\newtheorem{standing assumption}{Standing Assumption}[theorem]
\newtheorem{corollary}{Corollary}[theorem]

\usepackage{babel}

\newcommand{\X}{\mathcal{X}}
\newcommand{\Y}{\mathcal{Y}}
\newcommand{\Z}{\mathcal{Z}}
\let\P\relax
\DeclareMathOperator{\P}{\mathbb{P}}
\newcommand{\Pn}{\mathbb{P}_n}

\DeclareMathOperator{\R}{\mathbb{R}}

\newcommand{\Hk}{\mathcal{H}_k}
\DeclareMathOperator{\PTheta}{\mathcal{P}_{\Theta}}

\DeclareMathOperator{\simiid}{\overset{iid}{\sim}}

\DeclareMathOperator{\Q}{\mathbb{Q}}
\DeclareMathOperator{\E}{\mathbb{E}}
\DeclareMathOperator{\F}{\mathbb{F}}

\DeclareMathOperator{\GEM}{\operatorname{GEM}}
\DeclareMathOperator{\MMD}{\operatorname{MMD}}
\DeclareMathOperator{\Dir}{\operatorname{Dir}}
\DeclareMathOperator{\DP}{\operatorname{DP}}
\DeclareMathOperator{\HkX}{\mathcal{H}_{k_X}}
\DeclareMathOperator{\HkY}{\mathcal{H}_{k_Y}}
\DeclareMathOperator{\TLS}{\operatorname{TLS}}

\usepackage{array}
\newcommand\Tstrut{\rule{0pt}{2.6ex}}        
\newcommand\Bstrut{\rule[-0.9ex]{0pt}{0pt}}   

\begin{document}

\def\spacingset#1{\renewcommand{\baselinestretch}%
{#1}\small\normalsize} \spacingset{1}

%%%%%%%%%%%%%%%%%%%%%%%%%%%%%%%%%%%%%%%%%%%%%%%%%%%%%%%%%%%%%%%%%%%%%%%%%%%%%%

\if1\blind
{
  \title{\bf Robust Bayesian Inference for Measurement Error Misspecification: The Berkson and Classical Cases}
  \author{
  Charita Dellaporta \thanks{Department of Statistics, University of Warwick / correspondance to \texttt{c.dellaporta@warwick.ac.uk}},
    \hspace{.2cm} 
    Theodoros Damoulas \thanks{ Department of Statistics \& Department of Computer Science, University of Warwick}
    }
    \date{}
  \maketitle
} \fi

\if0\blind
{
  \bigskip
  \bigskip
  \bigskip
  \begin{center}
    {\LARGE\bf Robust Bayesian Inference for Measurement Error Models}
\end{center}
  \medskip
} \fi

\begin{abstract} 
Measurement error occurs when a covariate influencing a response variable is corrupted by noise. 
This can lead to misleading inference outcomes, particularly in problems where accurately estimating the relationship between covariates and response variables is crucial, such as causal effect estimation. 
Existing methods for dealing with measurement error often rely on strong assumptions such as knowledge of the error distribution or its variance and availability of replicated measurements of the covariates. 
We propose a Bayesian Nonparametric Learning framework that is robust to misspecification of these assumptions and does not require replicate measurements. 
This approach gives rise to a general framework that is suitable for both Classical and Berkson error models via the appropriate specification of the prior centering measure of a Dirichlet Process (DP). 
Moreover, it offers flexibility in the choice of loss function depending on the type of regression model. 
We provide bounds on the generalisation error based on the Maximum Mean Discrepancy (MMD) loss which allows for generalisation to non-Gaussian distributed errors and nonlinear covariate-response relationships. 
We showcase the effectiveness of the proposed framework versus prior art in real-world problems containing either Berkson or Classical measurement errors.
\end{abstract}

\noindent%

\spacingset{1.5} 
\section{Introduction}
\label{sec:intro}

An often encountered problem in statistical inference is when one or more covariate variables are subject to measurement error (ME). This is known as the ME, errors-in-variables, or input uncertainty problem and often arises in medical, economic, and natural sciences when it is hard or impossible to measure quantities in the real world exactly.
If ignored, ME in the covariates can result in biased parameter estimates \citep[see][]{carroll1995measurement}. If the goal is solely to make predictions, ME may not be a significant issue as long as observations continue to be realizations of the same random variable (RV). However, in many cases, we are interested in estimating the true parameter that explains the relationship between two RVs. For example, suppose we want to estimate a parameter that represents patients' condition based on the prevalence of a specific protein in their blood. This information can help clinicians determine the appropriate drug dosage that adjusts the protein levels. However, the available blood test may only provide an approximate measurement of the actual protein percentages. If this ME is ignored, the parameter estimates might give a misleading representation of the effect of the protein levels on the condition, leading to administering incorrect drug doses which could potentially be harmful to patients.

The ME problem has been extensively researched in both frequentist and Bayesian contexts. However, many proposed methods rely on knowing the ME distribution or variance, strong distributional assumptions, or the availability of replicate measurements and instrumental variables. Our aim is to develop a broadly applicable Bayesian approach, robust to ME on covariates when the model doesn't explicitly address ME and replicate measurements are unavailable. The focus is to provide informative, well-calibrated posterior densities for the parameters of interest while incorporating any available expert prior beliefs on ME, without assuming that these are well-specified. 
% If strong prior beliefs are absent, our method can still be employed by using non-informative priors or reducing the weight on priors. 
We introduce a Bayesian Nonparametric Learning (NPL) framework suitable for both Classical and Berkson error settings to achieve this.

\subsection{Paper outline}
The paper is outlined as follows. Section \ref{sec-related} introduces the Berkson and Classical measurement error and related work. Section \ref{sec-npl} provides the necessary background on the Bayesian NPL framework. In Section \ref{sec-method} we present the methodology for robust parametric inference in the ME setting (Robust-MEM) with two loss functions based on the TLS and the MMD. In Section \ref{sec-theory} we provide theoretical results of the generalisation error using the MMD-based loss function.  Finally, in Section \ref{sec:sim} we demonstrate strong empirical results on simulations and mental health and dietary datasets.

\subsection{Notation}
Consider two topological spaces $\mathcal{X}$ and $\mathcal{Y}$ equipped with Borel $\sigma$-algebras $\mathfrak{S}_{\X}$ and $\mathfrak{S}_{\mathcal{Y}}$ respectively and let $\mathcal{P}_{\X}$ and $\mathcal{P}_{\Y}$ be the space of Borel distributions on $\X$ and $\Y$ respectively. Let also $\mathcal{Z} := \mathcal{X} \times \mathcal{Y}$ equipped with $\sigma$-algebra $\mathfrak{S}_{\mathcal{Z}}$ and $\mathcal{P}_{\Z}$ the space of Borel distributions on $\Z$.
Throughout the paper we use the notation $\times$ for products of measures such that \(\P = \P_X \times \P_{Y\mid X}\) denotes the joint law on \(\Z \)
built from the marginal \(\P_X\) on \(\X\) and the conditional law \(\P_{Y\mid X}\)
on \(\Y\), given \(X\). Concretely, for measurable sets \(A\in \mathfrak{S}_{\X}\) and
\(B\in \mathfrak{S}_{\Y}\),
$$
\P(A\times B)\;=\;\int_A \P_{Y\mid X}(x,B)\,\P_X(dx).
$$
The notation $\delta_{x}$ is used to denote the Dirac measure at some point $x$. A Reproducible Kernel Hilbert space with associated kernel $k$ is denoted as $\Hk$ and its norm and inner product are denoted as $\| \cdot \|_{\Hk}$ and $\left<\cdot, \cdot \right>_{\Hk}$ respectively. Moreover, we denote by $\mu_{\P} := \E_{Z \sim \P}[k(Z,\cdot)] \in \Hk$ the kernel mean embedding of a probability measure $\P$. We further use $\independent$ and $\notindependent$ to denote independence and dependence between RVs, respectively. Lastly, the pushforward measure of a probability measure $\mu \in \mathcal{P}_{\X}$ by a function $f: \mathcal{P}_{\X} \rightarrow \Y$ is denoted as $f^{\#}(\mu) \in \mathcal{P}_{\Y}$.

\subsection{Berkson and Classical Measurement Error} \label{sec-related}
 Assume that we have observed data $\{(w_i, y_i)\}_{i=1}^n \in \mathcal{Z}$ and wish to perform parametric inference for a model family $\mathcal{P}_\Theta := \{\P_{g(\theta, x)} : x \in \X, \theta \in \Theta\} \subset \mathcal{P}_{\Y}$ with regression function $g: \Theta \times \X \rightarrow \mathbb{R}$ and parameter space $\Theta$ such that $x \rightarrow \P_{g(\theta,x)}(A)$ is $\mathfrak{S}_{\X}$-measurable for all $A \in \mathfrak{S}_{\Y}$ and all $\theta \in \Theta$. Further assume that we have RVs $W, X, N \in \X$ such that observations $\{w_i\}_{i=1}^n$ from $W$ are equal to the unobserved true values $\{x_i\}_{i=1}^n$ from $X$ corrupted with additive error $\{\nu_i\}_{i=1}^n$ from $N$. We assume that there exists unknown $\theta_0 \in \Theta$ such that: 
\begin{talign}
     Y &= g(\theta_0, X) + E, \quad E \sim \F^0_{E} \label{eq-eps-dist}
\end{talign}
for unknown probability measure $\F_{E}^0$  $ \in\mathcal{P}_{\Y}$ with $\mathbb{E}_{\F_E^0}[E] = 0$. We further assume that both $N$ and $E$ are homoscedastic, i.e. they have constant variance across observations and we denote their variances as $\sigma_\nu^2 := \text{Var}_{\mathbb{F}_N^0}[N]$ and $\sigma_\epsilon^2 := \text{Var}_{\F_E^0}[E]$ respectively. In this work, we consider the two most common types of ME, namely Berkson \citep{berkson1950there} and Classical. The main difference lies in the assumptions of dependence between the true and noisy covariates with the measurement error. Mathematically: 
\begin{itemize}
    \item {\textcolor{blue}{\tt{Berkson}}: $X = W + N, \quad W \independent N, \quad X \notindependent N, \quad N \sim \F_N^0$}
    \item \textcolor{red}{\tt{Classical}}: $W = X + N, \quad W \notindependent N, \quad X \independent N, \quad N \sim \F_N^0$
\end{itemize}
\begin{figure}[ht!]
    \centering
    \begin{subfigure}{0.4\textwidth}
    \centering
        \includegraphics[width=0.5\textwidth]{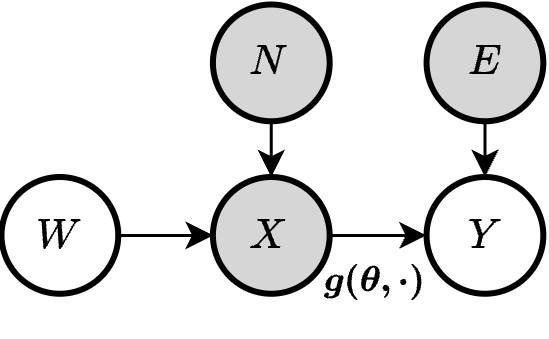}
        \caption{\textcolor{blue}{\texttt{Berkson}}}
        \label{fig:berkson}
    \end{subfigure}
    \hfill
    \begin{subfigure}{0.4\textwidth}
    \centering
        \vspace{0.1cm}
        \includegraphics[width=0.7\textwidth]{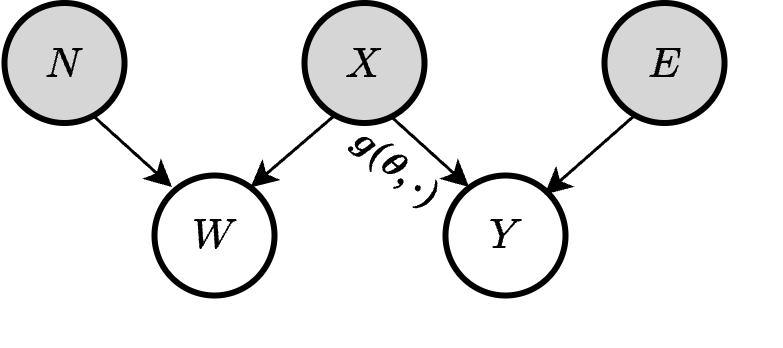}
        \hfill
        \caption{\textcolor{red}{\texttt{Classical}}}
        \label{fig:classical}
    \end{subfigure}
    \caption{Graphical representation of Berkson and Classical ME. Directed arrows indicate conditional dependency. The regression function explaining the relationship of $Y \given X$ is denoted by $g(\theta, \cdot)$. The two models differ on the dependency relationships between $N,W$ and $X$.}
    \label{fig:two_subfigures}
\end{figure}
Berkson ME occurs when the user assigns the observed covariates and they are hence independent of the ME. An example is when a group average ($W$) of the quantity of interest is assigned to each individual of the group, whereas the true individual levels ($X$) differ by an additive error. Another example occurs when a health practitioner prescribes a fixed dose of a drug ($W$) whereas the actual quantity of the drug substance absorbed by each body ($X$) differs between individuals with the same prescription. In contrast, classical ME often occurs when the ME is a source of inaccurate recording methods such as a faulty instrument. As a result, observed measurements ($W$) randomly fluctuate around the true value ($X$) and is hence correlated with the ME. For example, Classical ME often occurs in effect estimation problems met in health and epidemiology sciences. In these problems, we are interested in exposure-response relations \citep[see][]{brakenhoff2018measurement}. Exposures, often measured with error,
can include both self-reported quantities, such as patient lifestyle changes, and laboratory measurements such as serum cholesterol or the levels of aluminium in the brain \citep[][]{campbell2002potential}.

The difference between Berkson and Classical ME can also be understood from a graphical representation as illustrated in Figs. \ref{fig:berkson} and \ref{fig:classical}. In the Berkson case, $N$ directly affects $X$ whereas in the Classical case dependence flows from $N$ to $W$. 
In both cases, we assume that $N$ is random and $\F_N^0$ might not be precisely known. We aim to perform parametric inference on $\theta \in \Theta$, ensuring robustness to ME and potential misspecifications of a priori beliefs about its distribution. Below, we outline key literature on ME, emphasising its relevance to our work (see Appendix \ref{app:related} for further review).

\paragraph{Total Least Squares for ME} \label{sec-relw-tls}
The TLS method \citep{golub1979total, golub1980analysis} has been widely used to approach ME in the covariate. This is a generalisation of Ordinary Least Squares (OLS) wherein the residuals of both the response and covariate variables are minimised. These methods give useful computational tools for the ME problem; however, their use in nonlinear regression problems can be challenging. 
From a statistical point of view, TLS assumes Gaussian-distributed errors, corresponding to maximising the joint log-likelihood in the Classical ME case under the assumption that $N$ and $E$ follow zero-mean Gaussians. More details and a solution to the TLS problem can be found in Appendix \ref{app-tls-gaus}, \ref{app-tls-sol}.

\paragraph{Frequentist Methods for ME}
Most methods require knowledge of the distribution of the ME, or its variance. Deming regression \citep{deming1943statistical}, introduced for Classical ME, assumes that the ratio of the regression and measurement error variances, i.e. $\sigma^2_\epsilon/\sigma^2_{\nu}$, is known. For Berkson ME, \cite{huwang2000errors} suggested an iterative re-weighted least squares for polynomial regression, assuming Gaussian distributed error. \cite{wang2003estimation, wang2004estimation} considered the more general case of nonlinear regression with Gaussian and non-Gaussian distributed errors through a minimum distance estimation approach.  

The simulation-extrapolation method (SIMEX), suggested by \cite{cook1994simulation}, assumes knowledge of the ME variance or the ability to estimate it. SIMEX simulates additional ME to estimate the relationship between covariate and response variables. It has been studied for both Classical and Berkson error models \citep{althubaiti2016non} as well as mixtures of Classical and Berkson errors \citep{carroll2006measurement, misumi2018simulation}, while it has also been generalised to semi-parametric models \citep{apanasovich2009simex}. Since the ME variance is rarely known in practice, alternative methods for its estimation are often employed. For example, \cite{kekecc2022estimation} suggested using an expansion of Bernstein polynomials to identify the variance of Gaussian, additive, Classial ME. 

Alternative methods for Classical ME, assume access to repeated measurements of the mismeasured RV or instrumental variables (IV) \citep[e.g. see][]{liu1992efficacy,li2002robust, bowden1990instrumental}. In the former case, true covariates are treated as latent, and replicate measurements are used to estimate them. In contrast, IV methods use the IV variable, which is often difficult to identify, to overcome the bias induced by Classical ME.  
In the nonparametric literature, addressing classical ME often involves deconvolution kernel methods \citep{fan1993nonparametric, wang2011deconvolution} for kernel-based estimation of regression functions. Two key challenges lie in selecting the deconvolving kernel, which typically lacks a closed form and can be sensitive to bandwidth choice and the reliance on sufficiently smooth functions. 

\paragraph{Bayesian Approaches to ME}
Most Bayesian methods relax the assumptions discussed above but are often reliant on strong parametric assumptions. In the Classical error case, \cite{berry2002bayesian} restrict the distributions of the covariates, response error and ME to Gaussian distributions and model the regression function with smoothing splines and penalized mixtures of truncated polynomial splines.  \cite{sarkar2014bayesian} proposed the use of mixtures of B-splines and mixtures induced by Dirichlet Processes (DPs) to target the more general case of heteroscedastic errors. More recently, \cite{brown2024geometric} researched the geometric ergodicity of Gibbs samplers for multivariate Berkson and Classical ME, however the treatment is limited to Gaussian ME. More flexible models based on DPs and Gaussian mixtures have also been proposed by \cite{muller1997bayesian, carroll1999flexible}. While these overcome the simplicity of parametric forms, they rely on computationally expensive MCMC techniques. Recently, neural networks have been used in nonparametric regression with ME \citep[][]{hu2022measurement} to handle more generic regression functions in large-scale scenarios. 
However, knowledge of the ME distribution type remains necessary. In the nonparametric Bayesian literature, Gaussian Processes (GPs) have been used to address ME on the input covariate. For example, \cite{cervone2015gaussian} used GP modelling for location error on the input, employing a Hybrid Monte Carlo algorithm. Recently, \cite{zhou2023gaussian} proposed a fully Bayesian approach based on GP and Dirichlet process Gaussian mixture priors. 
GPs have also been applied in Berkson error scenarios by \cite{zhang2017modeling} in their study on tunnel deformation detection, but a general, fully Bayesian treatment of Berkson ME has not yet been offered. 

Our framework presents a Bayesian approach for regression under Berkson or Classical ME, eliminating the need for covariate replicate measurements or instrumental variables. Unlike existing methods, it doesn't assume precise ME or covariate distribution specifications.  It can incorporate prior beliefs about the ME distribution, such as its variance, and the degree of confidence in these, through a nonparametric prior on the unknown conditional distribution of $X \given W$.

\section{Bayesian Nonparametric Learning} \label{sec-npl}
A flexible way of approaching model misspecification in a Bayesian framework is through the Bayesian NPL setting \citep{lyddon2018nonparametric, fong2019scalable}. 
By deriving NPL, rather than standard Bayesian, posteriors we perform Bayesian inference under misspecified error distributions whilst fully characterising model uncertainty: prior beliefs about $X \given W$ are incorporated through a nonparametric prior. The resulting framework offers robustness against parametric distributional assumptions and facilitates the inclusion of prior beliefs regarding the error distribution.

NPL relaxes the standard Bayesian assumption that a parametric model family $\PTheta$ can correctly describe the Data-Generating Process (DGP) $\P^0$, i.e. that $\P^0 \in \PTheta$. 
The approach to incorporating uncertainty in the parameter of interest involves placing uncertainty directly on $\P^0$ and propagating it to $\theta \in \Theta$. It is accomplished using a conjugate nonparametric prior and associated posterior on $\P^0$, which then induces a posterior on the parameter space $\Theta$. This differs from traditional Bayesian inference, which conditions on the parameters rather than the DGP. This framework leads to a parallelizable computational framework to sample from this posterior on $\theta$, called the Posterior Bootstrap. 

The DGP is modelled using a DP prior, as proposed by  \cite{lyddon2018nonparametric} and \cite{fong2019scalable}: $\P \sim \DP(c, \F)$. 
Here, $c > 0$ is a concentration parameter and $\F \in \mathcal{P}_{\mathcal{Z}}$ is a centring measure. In the absence of ME, conditionally on the observations $\{(x_i,y_i\}_{i=1}^n \simiid \P^0$, we can obtain the posterior distribution of $\P$ by conjugacy:
\begin{talign} 
\begin{split}
\P \given \{(x_i,y_i)\}_{i=1}^n \sim \DP\left(c', \F' \right), \quad c' = c+n,\qquad  \F' = \frac{c}{c + n} \F + \frac{n}{c + n} \Pn \label{eq-DPPosterior}
\end{split}
\end{talign}
where $\Pn := \frac{1}{n} \sum_{i=1}^n \delta_{(x_i, y_i)}$ and $\delta_{(x,y)}$ is a Dirac measure at $(x,y) \in \mathcal{Z}$. The value of $c$ indicates the level of confidence we have in the prior centring measure $\F$. 
The NPL framework, in its standard form, is not sufficient to account for ME, since it assumes that we have access to realisations from the DGP. We introduce an approach, suitable for this setting, based on the NPL principles. 

\section{Methodology} \label{sec-method}
We present a Bayesian framework for ME models (MEMs) by means of a nonparametric prior on the conditional distributions $\{X \given W = w_i\}_{i=1}^n$. This is achieved through a flexible DP distribution, allowing the prior definition for both Berkson and Classical error models. We explore two loss functions within this framework: one based on the TLS objective and the other on the MMD. The latter extends the NPL-MMD framework \citep{dellaporta2022robust} to the ME context, enabling the handling of non-Gaussian and nonlinear models.
\subsection{NPL for Measurement Error (Robust-MEM)} \label{sec-npl-input-unc}
The standard NPL setting (Sec. \ref{sec-npl}) assumes the availability of a finite number of observations from the DGP and allows for inference on a potentially misspecified model. In the MEM setting, observations from the DGP are not available since we have observations of the pair of random variables $(W,Y)$, instead of $(X,Y)$. In this case, we are interested in setting the uncertainty directly on $X \given W$ and propagating this on $\theta$ through an NPL posterior. 
Recall the observed data $\{(w_i, y_i)\}_{i=1}^n$, such that realizations of $X$ are observed with some additive, zero mean, and constant variance error, i.e. $w_i = x_i + \nu_i$ in the Classical case or $x_i = w_i + \nu_i$ in the Berkson case, for all $i = 1, \dots, n$. Denote by $\P_{X \given w_i} \in \mathcal{P}_\X$ the distribution of $X \given W = w_i$. We set a nonparametric prior on $\P_{X \given w_i}$ through a DP:
\begin{talign} \label{eq-dp-prior}
    \P_i \sim \text{DP}(c, \F_{w_i})
\end{talign}
The prior centering measure $\F_{w_i}$ is indexed by $w_i$ to indicate the dependence on the observed covariate $w_i$. $\F_{w_i}$ encodes our prior beliefs about the true and noisy covariate relationship through distributional assumptions on the ME and the true covariate while respecting the Berkson or Classical structure of the error. We further discuss the choices of $\F_{w_i}$ and $c$ in Sections \ref{sec-prior} and \ref{sec-prior-elicit} respectively. 
Since we only have one data point $(w_i, y_i)$ for each conditional $X \given W = w_i$, the conjugate posterior DP is:
\begin{talign}  \label{eq-dp-posterior}
    \P_{i} \given w_i \sim \text{DP}\left(c + 1, \F_{w_i}'\right), \quad \F_{w_i}' = \frac{1}{c+1} \delta_{w_i} + \frac{c}{c+1} \F_{w_i}.
\end{talign}
\begin{remark}
    Although we focus on the general case when we only have one data point $(w_i, y_i)$ for each conditional $X \given W = w_i$, in the case of Berkson ME, it is common to assume that there are multiple observations associated with the same observed covariate $w_i$ but different true covariate value $x_i$. For example, this is the case when a fixed dose $w_i$ is assigned to multiple subjects of a trial \citep[see e.g.][]{burr1988errors}. In this case, the posterior DP conditioned on all the observations from the same set $i$ would be defined as $\text{DP}\left(c + n_j, \F_{w_i}'\right)$, where $n_j$ is the number of observations in set $i$ and $\F_{w_i}' = \frac{n_j}{c+n_j} \delta_{w_i} + \frac{c}{c+n_j} \F_{w_i}.$
\end{remark} 
Let $\P^{0} = \P^0_X \times \P^0_{Y \given X} \in \mathcal{P}_\mathcal{Z}$ be the true joint DGP 
where $\P^0_X \in \mathcal{P}_\mathcal{X}$ denotes the distribution of $X$ and $\P^0_{Y \given X}$ denotes the true conditional distribution of $Y \given X$. 
The target for any loss function $l: \mathcal{X} \times \mathcal{Y} \times \Theta \rightarrow \R$ is:
\begin{talign} \label{eq-target-theta}
\theta^0_l = \argmin_{\theta \in \Theta} \mathbb{E}_{(x,y) \sim \P^{0}} \left[l(x,y,\theta) \right].
\end{talign}
Of course, $\P^{0}$ is unknown and we do not have access to observations from $\P_X^0$ and consequently $\P^0$. We instead have a set of posterior distributions for each of the conditional distributions $X \given W = w_i$. The goal is to use the nonparametric posterior of (\ref{eq-dp-posterior}) to obtain a posterior on $\theta \in \Theta$. This is naturally defined through the pushforward measure: 
\begin{talign*}
    (\theta^\star_l)^{\#}\left(\mu_1, \dots, \mu_n \right)
\end{talign*}
where for probability measures $\P^1 \sim \mu_1, \dots, \P^n \sim \mu_n$ on $\mathcal{P}_\X$ and $\theta^\star_l : \mathcal{P}_\X^n \rightarrow \Theta$ we define:
\begin{talign} 
\P^i \sim \mu_i &:= \DP(c + 1,\F_{w_i}') \quad \forall i = 1, \dots n, \label{eq-mu} \\
\theta^\star_l(\P^1, \dots, \P^n)  &:= \argmin_{\theta \in \Theta} \mathbb{E}_{(x,y) \sim  \frac{1}{n} \sum_{i=1}^n \P^i \times}\delta_{y_i} \left[l(x,y,\theta) \right]. \label{eq-opt-step}
\end{talign}
Sampling from $(\theta^{\star}_l)^{\#}\left(\mu_1, \dots, \mu_n \right)$ can be performed through the principles of the Posterior Bootstrap \citep{fong2019scalable} where at iteration j:
\begin{enumerate}
    \item For $i = 1,\dots, n: \quad$ 
    sample $\P^{(i,j)} \sim \mu_i$. 
    \item $\theta^{\star}_l(\P^{(1,j)}, \dots, \P^{(n,j)}) = \argmin_{\theta \in \Theta} \mathbb{E}_{(x,y) \sim \frac{1}{n} \sum_{i=1}^n \P^{(i,j)} \times}\delta_{y_i} \left[l(x,y,\theta) \right]$.
\end{enumerate}
Exact sampling from a DP, as in step 1, is generally not feasible. The stick-breaking process \citep{sethuraman1994constructive} offers an exact representation of a DP which can be approximated using the Dirichlet approximation of the stick-breaking process \citep{muliere1996bayesian,ishwaran2002exact} which we use in this work. We denote by $\hat{\mu_i}$, the approximation of $\mu_i$, meaning that $\P^{(i,j)} \sim \hat{\mu_i}$ is such that:
\begin{talign} \label{eq-dp-approx}
    \P^{(i,j)} = \sum_{k=1}^T \xi_k^{i} \delta_{\tilde{x}^i_k} + \xi^i_{T+1} \delta_{w_i}, \quad \tilde{x}^i_{1:T} \overset{iid}{\sim} \F_{w_i}, \quad\xi^i_{1:T+1} \sim \Dir\left(\frac{c}{T}, \dots, \frac{c}{T}, 1\right).
\end{talign}
The intuition of this sampling procedure lies in the fact that different marginal distributions of $\P^{(i,j)}$ lead to different joint distributions of the form $$\frac{1}{n} \sum_{i=1}^n \P^{(i,j)} \times \delta_{y_i}.$$ Uncertainty stems from the marginal distribution of $X$ due to ME, hence by bootstrapping over different marginal distributions of $X$ we induce robustness to the inference procedure for the joint distribution of $(X,Y)$.
The taxonomy of the Robust-MEM framework for different choices of loss functions and prior DP specifications is outlined in Fig. \ref{fig:taxonomy}.

\subsection{Choice of Prior Centering Measure} \label{sec-prior}
The choice of the prior centering measure $\F_{w_i}$ in (\ref{eq-dp-prior}) represents our prior beliefs about the distribution of the true covariate conditionally on an observed noisy covariate $w_i$. This choice should respect the Berkson and Classical dependence assumptions.
\begin{figure}[ht!]
    \centering
  \includegraphics[width=0.75\textwidth]{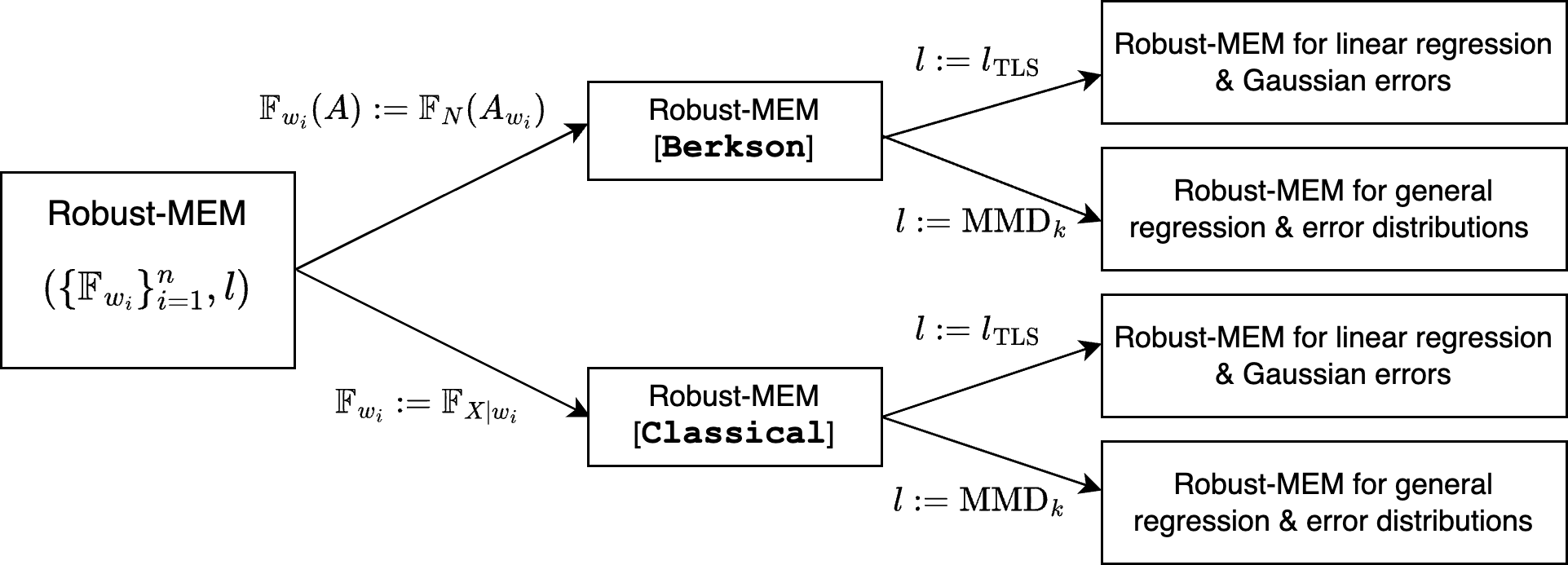}  
  \caption{Summary of different Robust-MEM frameworks based on the choices of prior centering measures $\{\F_{w_i}\}_{i=1}^n$ and loss function $l$.}
  \label{fig:taxonomy}
\end{figure}

In the Berkson ME case,  $X \notindependent N$ whereas $W \independent N$. Hence, it is sufficient to define a prior on the distribution of $X \given W = w_i$ through the prior distributional assumptions on $N$. Suppose that $\F_{N}$ denotes the prior distribution on $N$, then following the Berkson error structure we can define $\F_{w_i}$ in $(\ref{eq-dp-prior})$ such that $x \sim \F_{w_i}$ implies that $x = w_i + \nu$ where $\nu \sim \F_{N}$. More formally, for every set $A \in \sigma_{\X}$,  $\F_{w_i}(A) := \F_N(A_{w_i})$ where $A_{w_i} := \{\nu: w_i + \nu \in A \}$. 

In the Classical error case, $X \independent N$ so prior beliefs should be considered about both the distribution of $N$ and the marginal distribution of $X$. The Classical ME assumption tells us that $W \given X = x_i$ is centered at $x_i$ and is distributed according to the distribution of $N$. Hence, given a prior distribution $\F_X$ on $X$ with density $f_X$, a prior distribution $\F_{W \given X}$ on $W \given X$ such that $\E[W \given X = x] = x $ with density $f_{W \given X}$ and the observation $w_i$, we can use Bayes rule to obtain the density of $X \given W = w_i$ through:
\begin{talign} \label{eq-bayesian-post}
    f_{X \given W}(x \given w_i) = \frac{f_{W \given X}(w_i \given x)f_X(x)}{\int f_{W \given X}(w_i \given x)f_X(x) dx} 
\end{talign}
The prior centring measure can then be defined as the probability measure with corresponding density $f_{X \given W}(\cdot \given w_i)$, which we denote by $\F_{X \given w_i}$. Here, if $f_X$ is conjugate to $f_{W \given X}$, then $f_{X \given W}$ can be obtained in closed form and sampling from the corresponding probability measures $\F_{X \given w_i}$ becomes trivial. For example, if $\F_X := \mathcal{N}(\mu_x, \sigma^2_x)$ and $\F_{W \given X} := \mathcal{N}(X, \sigma^2_{\nu})$ then by conjugacy $\F_{X \given w_i} = \mathcal{N}\left(\frac{\sigma_x^2}{\sigma_{\nu}^2 + \sigma_x^2} w_i +  \frac{\sigma_{\nu}^2}{\sigma_{\nu}^2 + \sigma_x^2} \mu_x, (\frac{1}{\sigma_x^2} + \frac{1}{\sigma_\nu^2})^{-1}\right)$.
\begin{remark}
    The Classical ME structure naturally arises in the study of DP mixture models. Since $X \independent N$ one could choose to set uncertainty directly to the distribution of $X$ through a DP and define a mixture model as follows: 
    \begin{talign*}
        w_i \given x_i \sim F(x_i), \quad 
        x_i \given G &\sim G, \quad
        G \sim \DP(c, \F),
    \end{talign*}
    where $F$ denotes the distribution of the error conditionally on $x_i$ and $G$ is the mixing distribution over $x$. However, such a structure does not lead to a conjugate posterior for $G$ and more complicated sampling procedures should be used \citep{neal2000markov}. 
\end{remark}

Understanding the different prior centering measures in each scenario is facilitated by Algorithms \ref{alg:berk} and \ref{alg:class} of the Robust-MEM framework. In the Berkson case (Algorithm \ref{alg:berk}), sampling from $\F_{w_i}$ involves directly sampling from the prior centering distribution for $N$, denoted by $\F_{N}$. In the Classical case (Algorithm \ref{alg:class}), sampling is conducted directly through the probability measure $\F_{X \given w_i}$ defined via the density $f_{X \given W}(\cdot \given w_i)$ in (\ref{eq-bayesian-post}). Notably, the probability measures $\F_{N}$ in the Berkson case and $\F_{W \given X}$ and $\F_X$ in the Classical are the same across all $i = 1, \dots, n$ representing prior beliefs about the distributions of $X$ and $N$. Hence, although different DPs are set for each conditional $X \given W = w_i$, some information is shared across all of them via the prior centering specification. 

\begin{minipage}{0.45\textwidth}
% \vspace{0.3cm}
\begin{algorithm}[H] 
\SetAlgorithmName{Algorithm}{}{}
\SetAlgoLined
\SetKwInOut{Input}{input}
 \Input{$\{(w_i,y_i)\}_{i=1}^n$, $T$, $B$, $c$, $l$, $\F_{N}$}
  % \vspace{-0.4cm}
 \For{$j\gets1$ \KwTo $B$}{
  % \vspace{-0.4cm}
 \For{$i\gets1$ \KwTo $n$}{
 Sample $\tilde{\nu}_{1:T}^{(i,j)} \simiid \F_{N}$ \\
  % \vspace{-0.3cm}
  Set $\tilde{x}_{1:T}^{(i,j)} = w_i + \tilde{\nu}_{1:T}^{(i,j)}$ \\
  % \vspace{-0.22cm}
 Sample $\xi_{1:(T+1)}^{(i,j)} \sim \Dir\left(\frac{c}{T}, \dots, \frac{c}{T}, 1\right)$ \\
 % \vspace{-0.22cm}
 Set $\P^{(i,j)} := \sum_{k=1}^T \xi^{(i,j)}_k \delta_{\tilde{x}^{(i,j)}_k} + \xi^{(i,j)}_{T+1} \delta_{w_i}$ 
 }
$\theta^{(j)} = \theta^\ast_{l}(\P^{(1,j)}, \dots, \P^{(n,j)})$ in (\ref{eq-opt-step})  \\
 }
 \Return{Posterior sample $\theta^{(1:B)}$}
 \caption{Robust-MEM (\textcolor{blue}{\texttt{Berkson}})}
 \label{alg:berk}
\end{algorithm}
\end{minipage}
% \hfill
\begin{minipage}{0.46\textwidth}
 % \vspace{0.3cm}
\begin{algorithm}[H]
\SetAlgorithmName{Algorithm}{}{}
\SetAlgoLined
\SetKwInOut{Input}{input}
 \Input{$\{(w_i,y_i)\}_{i=1}^n$, $T$, $B$, $c$, $l$, $\F_{N}$, $\F_{X}$}
 \vspace{0.07cm}
 \For{$j\gets1$ \KwTo $B$}{
 % \vspace{-0.4cm}
 \For{$i\gets1$ \KwTo $n$}{
  Obtain $\F_{X \given w_i}$ from (\ref{eq-bayesian-post}) \\ 
 % \vspace{-0.2cm}
 Set $\tilde{x}_{1:T}^{(i,j)} \simiid \F_{X \given w_i}$ \\
 % \vspace{-0.2cm}
 Sample $\xi_{1:(T+1)}^{(i,j)} \sim \Dir\left(\frac{c}{T}, \dots, \frac{c}{T}, 1\right)$ \\
 % \vspace{-0.22cm}
 Set $\P^{(i,j)} := \sum_{k=1}^T \xi^{(i,j)}_k \delta_{\tilde{x}^{(i,j)}_k} + \xi^{(i,j)}_{T+1} \delta_{w_i}$ \\
 }
$\theta^{(j)} = \theta^\ast_{l}(\P^{(1,j)}, \dots, \P^{(n,j)})$ in (\ref{eq-opt-step}) 
 }
 \Return{Posterior sample $\theta^{(1:B)}$}
 \caption{Robust-MEM (\textcolor{red}{\texttt{Classical}})}
 \label{alg:class}
\end{algorithm}
\end{minipage}
\subsection{Prior Elicitation} 
\label{sec-prior-elicit}
We now focus on the concentration parameter of the DP prior in (\ref{eq-dp-prior}). It is clear from the posterior in (\ref{eq-dp-posterior}) that $c = 0$ corresponds to the prior belief that there is no ME present on the covariates resulting in all the mass of the posterior centering measure concentrating at $\delta_{w_i}$. 
This is in contrast to the standard NPL setting (Sec. \ref{sec-npl}) where $c=0$ corresponds to a noninformative prior. Moreover, it is evident that $c=1$ results in splitting the weight of the DP posterior centering measure $\F_{w_i}'$ equally between the prior centering measure $\F_{w_i}$ and the Dirac measure on the observation $w_i$. 
Lastly, when $c \rightarrow \infty$ we see that all the weight is placed in the prior belief. 
The interplay between parameter $c$, ME variance, and prior specification can be further demonstrated with an example of polynomial regression with Berkson, normally distributed ME in Fig. \ref{fig:prior-elicit}.
\begin{figure}[ht!]
    \centering
  \includegraphics[width=0.95\textwidth]{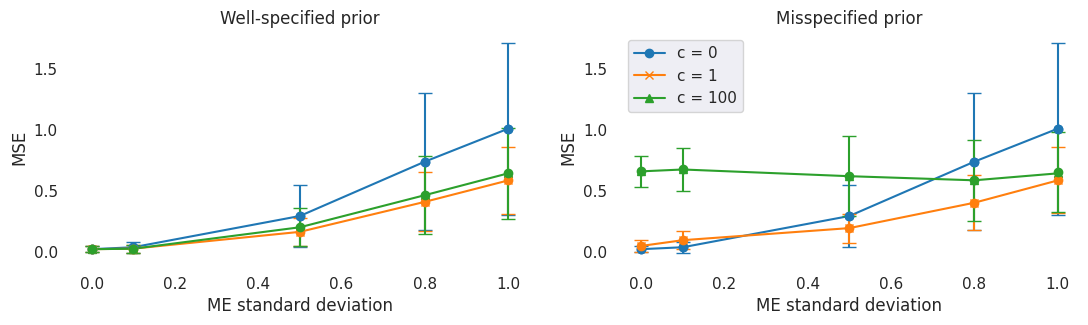}  
  \caption{Average MSE in polynomial regression with Berkson, Gaussian ME over 50 replications. Left: Robust-MEM method with a well-specified prior. Right: Robust-MEM method with a misspecified (Student-t with 3 degrees of freedom) prior.}
  \label{fig:prior-elicit}
\end{figure}
Here, the average MSE over the three parameters of interest is reported for the three limiting choices of $c$, for an increasing ME variance. In the case of a well-specified prior (left), when the ME variance is close to zero, i.e. when there is no ME, all three $c$ values perform similarly since the prior correctly indicates no ME, making inference data-driven. The choices of $c > 0$ perform similarly, yielding a lower MSE for nonzero ME variance. Conversely, in the case where the true ME distribution is Gaussian and the prior centering measure $\F_{N}^0$ is a Student-t distribution with 3 degrees of freedom (right plot), $c = 1$ corresponds to the lowest average MSE, as expected. A setting of $c = 1$ appropriately balances prior beliefs and observations, while $c = 0$ wrongly assumes no ME, and $c = 100$ overly weights the misspecified prior. In general, overlapping standard deviations in most cases indicate the robustness of the method to $c$ values. In the absence of strong prior beliefs on the ME distribution, setting $c = 1$ splits the weight sensibly between prior beliefs and observations.
\subsection{Choice of Loss Function}
In this section, we explore the adaptability of the Robust-MEM framework to various prevalent cases of ME. We propose two distinct loss functions suitable to address these scenarios comprehensively. The initial choice stems from the TLS objective, offering rapid and robust solutions, particularly tailored for scenarios such as linear regression and Gaussian-distributed ME. To ensure broad applicability across diverse model families and regression relationships, we introduce the MMD-based loss function as a further approach.
\subsubsection{Robust-MEM with the Total Least Squares} \label{sec-npl-tls}
We first suggest a loss function based on the TLS objective function. Consider a linear regression model. Then we can define the loss function $l_{\TLS}:\Z \times \Theta \rightarrow \mathbb{R}$ as
$l_{\TLS}(x_i,y_i,\theta) := \|\nu_i\|_2^2 +  (y_i - \theta^T (x_i + \nu_i))^2. $
Using this loss function in (\ref{eq-opt-step}), the minimization objective for each set of DP samples $(\P^{(1,j)}, \dots, \P^{(n,j)})$ as in (\ref{eq-dp-approx}) is
\begin{talign} 
    &\theta^{\star}_{TLS}(\P^{(1,j)}, \dots, \P^{(n,j)}) \nonumber\\
    &\quad \quad = \argmin_{\theta \in \Theta } \mathbb{E}_{(x,y) \sim \frac{1}{n} \sum_{i=1}^n \P^{(i,j)} \times \delta_{y_i}} \left[l_{\TLS}(x,y,\theta) \right] \nonumber\\ 
    &\quad \quad = \argmin_{\theta \in \Theta} 
    \frac{1}{n} \sum_{i=1}^n \left[ \sum_{t=1}^T \xi_t^{(i,j)} l_{\TLS}(\tilde{x}_t^{(i,j)}, y_i, \theta) + \xi_{T+1}^{(i,j)} l_{\TLS}(w_i,y_i,\theta) \right] \label{eq-tls-obj2}
    \end{talign}
where in (\ref{eq-tls-obj2}), optimisation over $\nu$ is omitted since the parameter of interest is $\theta$. This loss function is equivalent to minimizing the joint likelihood between $(W,Y)$ conditionally on $X$ in the case of Gaussian distributed errors, i.e. when $\F_{E}^0 \equiv N(0, \sigma^2_{\epsilon})$ and $\F_{N}^0 \equiv N(0, \sigma^2_{\nu})$ (see Appendix \ref{app-tls-gaus}). 
The pseudo-code for the Robust-MEM with the TLS objective is outlined in Algorithm \ref{alg:TLS_posterior_bootstrap} (Appendix \ref{app-exp-dets}).
This framework allows for flexible Bayesian parameter inference in MEMs while accounting for prior beliefs on the ME. However, it only allows for inference on Gaussian distributed errors and linear regression functions, which can be limiting and prone to misleading outcomes if the assumptions are misspecified. To address this limitation, we introduce a loss function based on the MMD.

\subsubsection{Robust-MEM with the Maximum Mean Discrepancy} \label{sec-npl-mmd-input-unc}
The use of distance-based loss functions in the NPL framework was suggested by \cite{dellaporta2022robust} to perform robust inference for simulator-based models with the Maximum Mean Discrepancy (MMD). The MMD is defined in a Reproducing Kernel Hilbert Space (RKHS) $\mathcal{H}_k$ with associated kernel $k: \mathcal{Z} \times \mathcal{Z} \rightarrow \mathbb{R}$ and norm $\| \cdot \|_{\mathcal{H}_k}$. The MMD has already been used in the context of universal robust regression by \cite{alquier2024universal} and is defined through the notion of kernel mean embeddings.
For a function $f \in \mathcal{H}_k$, the kernel mean embedding $\mu_{\P} \in \mathcal{H}_k$ with respect to probability measure $\P \in \mathcal{P}_{\mathcal{Z}}$ is defined as $\mu_{\P}(\cdot) := \mathbb{E}_{Z \sim \P}[k(Z,\cdot)] \in \mathcal{H}_k$.
The MMD between the probability measures $\P$ and $\Q$ is defined as $\MMD_k(\P, \Q) = \| \mu_{\P} - \mu_{\Q} \|_{\mathcal{H}_k}$. For a parametric model family $\{\P_{\theta} : \theta \in \Theta\}$ and a sample $\{z_i\}_{i=1}^n \simiid \P^0$, the squared MMD can be used as a loss function in (\ref{eq-target-theta}), resulting in a minimum MMD estimator \citep{briol2019statistical} between $\hat{\P}_n = \frac{1}{n} \sum_{i=1}^n \delta_{z_i}$ and $\P_{\theta}$, given by:
\begin{talign*}
\theta^n_{\MMD} &:= \arg \inf_{\theta \in \Theta} \MMD_k^2\left(\hat{\P}_n, \P_{\theta} \right) \\
    &= \arg \inf_{\theta \in \Theta} \int_{\mathcal{Z} \times \mathcal{Z}} k(z_1, z_2) \P_{\theta}(dz_1) \P_{\theta}(dz_2) - \frac{2}{n} \sum_{i=1}^n \int_{\mathcal{Z}} k(z_i, z_1) \P_{\theta}(dz_1).
\end{talign*}
In practice, this is approximated using a U-statistic and optimized with gradient-based methods. 
This loss function is suitable for general regression models and hence allows for the specification of nonlinear regression functions $g(\theta, x)$ as well as non-Gaussian distributed errors for both $X$ and $Y$.

Recall that $\P^0$ denotes the joint DGP of $X$ and $Y$, such that $\P^0 = \P^0_X \times \P^0_{Y \given X}$. Then the target parameter of interest is:
$$\theta^0_{\MMD} := \arg\inf_{\theta \in \Theta} \MMD_k(\P^{0}, \P_X \times \P_{g(\theta, \cdot)}),$$
representing the parameter value that best explains the relationship between $X$ and $Y$. 
The uncertainty is twofold: firstly $\P^{0}$ is unknown and secondly, we do not have access to samples from $\P_X$, as we only observe realizations $\{(w_i, y_i)\}_{i=1}^n$ of:
$
\P^n := \frac{1}{n} \sum_{i=1}^n \delta_{w_i} \times \P_{Y \given w_i}. 
$
We follow the same nonparametric framework for each of the distributions $\P_{X \given w_i}$ as in (\ref{eq-dp-prior}) and (\ref{eq-dp-posterior}). This leads to the NPL posterior density for $\theta$ through the push-forward measure $(\theta^{\star}_{\MMD})^{\#}\left(\mu_1, \dots, \mu_n \right)$ with $\mu_i$ as in (\ref{eq-mu}). The Posterior bootstrap algorithm (Algorithm \ref{alg:TLS_posterior_bootstrap_MMD}) facilitates sampling from this measure by iteratively optimizing the following objective:
\begin{talign} \label{eq-mmd-obj}%, \P^j_x
    &\theta^{\star}_{\MMD}\left( \P^{(1,j)}, \dots, \P^{(n,j)} \right) \nonumber\\
    &\quad \quad = \arg\inf_{\theta \in \Theta} \MMD_k \left( \frac{1}{n} \sum_{i=1}^n \P^{(i,j)} \times \delta_{y_i}, \frac{1}{n} \sum_{i=1}^n \P^{(i,j)} \times\P_{g(\theta, \cdot)} \right) 
\end{talign}
where $\P^{(i,j)} \sim \hat{\mu_i}$ for each $i = 1, \dots, n$ and if $(X,Y) \sim \P^{(i,j)} \times \P_{g(\theta, \cdot)}$ then $X \sim \P^{(i,j)}$ and $Y \given X = x \sim \P_{g(\theta, x)}$ mirroring the definition in \cite{alquier2024universal}. 

\section{Generalization Error of Robust-MEM with the MMD}  \label{sec-theory}
In this section, we present generalisation error bounds for the Robust-MEM framework with the MMD. We define the target distribution for the Berkson and Classical cases and outline key technical assumptions in section \ref{sec-theory-target}. The main results are presented in section \ref{sec-theory-main}. Finally, we analyze the theoretical results under common ME assumptions, such as known error distribution and Gaussian-distributed ME, in Section \ref{sec-theory-cor}. Before we start the theretical assesment we summarise the notation used for each error structure in Table \ref{tab:prob-meas-defs}. 
\begin{table}[h]
\small
    \centering
    \caption{Summary of notation for Berkson and Classical error models.}
    \label{tab:prob-meas-defs}
    \begin{tabular}{ccc}
    \hline 
         & \textcolor{blue}{\texttt{Berkson}} & \textcolor{red}{\texttt{Classical}} \\
        \hline
        Observed distribution of $W$ & $\P_W^n := \frac{1}{n} \sum_{i=1}^n \delta_{w_i}$ & $\P_W^n := \frac{1}{n} \sum_{i=1}^n \delta_{w_i}$\Tstrut\Bstrut \\
        Target distribution of $X$ & $\P^{0,n}_X := \frac{1}{n} \sum_{i=1}^n \P^0_{X \given w_i}$ & $\P_X^0$\Tstrut\Bstrut \\
        Target distribution of $(X,Y)$ & $\P^0 := \P^{0,n}_X \times \P^0_{Y \given X}$ & $\P^0 := \P^0_X \times\P^0_{Y \given X}$\Tstrut\Bstrut \\
        Joint model of $(X, Y)$ & $\P^0_{\theta} := \P^{0,n}_X \times \P_{g(\theta, \cdot)}$ & $\P^0_{\theta} := \P^{0}_X \times \P_{g(\theta, \cdot)}$\Tstrut\Bstrut \\
        DP posterior for $\P_{X \given w_i}$ & $\mu_i := \DP(c + 1, \F_{w_i}')$ & $\mu_i := \DP(c + 1, \F_{w_i}')$\Tstrut\Bstrut \\
        Prior centering measure & $\F_{w_i}(A) := \F_N(A_{w_i})$  & $\F_{w_i} := \F_{X \given w_i}$\Tstrut\Bstrut \\
        \hline
    \end{tabular}
\end{table}

\subsection{Target Distribution \& Assumptions} \label{sec-theory-target}
The goal is to bound the expected error in terms of the MMD between the true joint distribution of $X, Y$, namely $\P^0$, and the estimated model. 
The Berkson and Classical error cases lead to different forms of $\P^0$, due to the dependence structures of $X$, $N$, $W$.

Recall that in the Berkson error case, $N \notindependent X$ but $N \independent W$. Since the focus is on the true relationship between $X$ and $Y$, we consider the fixed design case by treating the observations $\{w_i\}_{i=1}^n$ as fixed points. This is a common assumption in the Berkson case as the values are often assigned by an expert \citep[see e.g.][]{delaigle2006nonparametric}. Recall that by assumption, $Y \independent W \given X$ and let $\P^0_{X \given W}$ and $\P^0_{Y \given X}$ denote the true conditional distributions of $X \given W$ and $Y \given X$, respectively. Then the true joint distribution for any sets $dx, dw \in \mathfrak{S}_{\X}$ and $dy \in \mathfrak{S}_{\Y}$ is:
$
    \P^0(W \in dw, X \in dx, Y \in dy) = \frac{1}{n} \sum_{i=1}^n \delta_{w_i} (dw) \times \P^0_{X \given w_i}  (dx) \times\P^0_{Y \given X} (dy).
$
By marginalizing out $W$ we obtain the target joint distribution of $X,Y$ as: 
\begin{talign} 
   \P^0(X \in dx, Y \in dy) &:= \P^{0,n}_X (dx) \times \P^0_{Y \given X} (dy),  \label{eq:target-berk-joint}\\ 
  \text{where} \quad \P_X^{0,n} (dx) &:= \frac{1}{n} \sum_{i=1}^n \P^0_{X \given w_i} (dx). \label{eq:target-berk-x}
\end{talign}
Further notice that by the assumption of the Berkson error, for each $i = 1, \dots n$, $\P^0_{X \given w_i}$ is such that $x_i \sim \P^0_{X \given w_i}$ means that $x_i = w_i + \nu_i$ and $\nu_i \sim \F_{N}^0$ for the true error distribution $\F_N^0$ or equivalently for any $A \in \sigma_{\X}$ we have $\P^0_{X \given w_i}(A) = \F_N^0(A_{w_i})$ where $A_{w_i} := \{\nu: w_i + \nu \in A\}$. 

In the Classical error case, $X \independent N$ hence it is not sufficient to define the joint distribution via the conditional distribution of $X \given W$. Instead, we need to go beyond the fixed design case and consider the true marginal distribution of $X$. In this case, the joint distribution of $W, X, Y$ is $\P^0(X \in dx, W \in dw, Y \in dy) = \P^0_{X} (dx) \times \P^0_{W \given X} (dw) \times\P^0_{Y \given X} (dy) $
whereas the target joint distribution of $X,Y$ is simply 

\begin{talign}
    \P^0(X \in dx, Y \in dy) = \P^0_{X} (dx) \times \P^0_{Y \given X} (dy).
\end{talign}
For random samples $\{\P^{i} \sim \mu_i\}_{i=1}^n$, with $\mu_i$ as in (\ref{eq-mu}) define $\P := (\P^1, \dots, \P^n)$. The estimated model $\P_{\theta^{\star}(\P)}$, relative to the noise-free covariates, for the Berkson and Classical cases is:  
\begin{talign} \label{eq:estimates-model}
  \text{\tt (Berkson)} \quad \P^{0}_{\theta^{\star}(\P)} (dx \times dy) &:= \P^{0,n}_X (dx) \times\P_{g(\theta^{\star}(\P), \cdot)} (dy) \quad \text{or} \\ \text{\tt (Classical)} \quad \P^{0}_{\theta^{\star}(\P)} (dx \times dy) &:= \P^0_X (dx) \times\P_{g(\theta^{\star}(\P), \cdot)} (dy)
\end{talign}
The goal is to examine how well $\P^{0}_{\theta^{\star}(\P)}$ approximates $\P^{0}$ in expectation if we had access to the distribution of the true covariates. That is, we check if the resulting model captures the relationship between the true covariates and the response variable.

In what follows, we consider the exact DP case for clarity. 
This is represented through the stick-breaking process \citep{sethuraman1994constructive} such that for each $\P^i \sim \mu_i := \DP(c+1,\F_{w_i}') $:
\begin{talign} \label{eq:p-i}
\P^i = \sum_{k=1}^{\infty} \xi_k^i \delta_{\tilde{x}_k^i}, \quad \tilde{x}^i_{1:\infty} \simiid \F'_{w_i}, \quad \xi^i_{1:\infty} \sim \GEM(c+1).
\end{talign}
Using this notation, we also define the joint model under samples from the DP posteriors:
\begin{talign}\label{eq-bar-p-theta-exact}
    \P_{\theta} := \frac{1}{n} \sum_{i=1}^n \P^i  \times\P_{g(\theta, \cdot)} = \frac{1}{n} \sum_{i=1}^n \sum_{k=1}^{\infty} \xi_k^i \delta_{\tilde{x}_k^i} \times \P_{g(\theta, \tilde{x}_k^i)}.
\end{talign}
Recall that $\mathcal{X}$ and  $\mathcal{Y}$ are two topological   spaces, equipped respectively with the Borel $\sigma$-algebra $\mathfrak{S}_\mathcal{X}$ and $\mathfrak{S}_\mathcal{Y}$, and let $\mathcal{Z}= \mathcal{X}\times \mathcal{Y}$, $\mathfrak{S}_\mathcal{Z}=\mathfrak{S}_\mathcal{X}\otimes \mathfrak{S}_\mathcal{Y}$ and $\mathcal{P}_\mathcal{Z}$ be the set of  probability distributions on $(\mathcal{Z},\mathfrak{S}_\mathcal{Z})$. 
We let $k:\mathcal{Z}^2 \rightarrow \mathbb{R}$ be a reproducing kernel on $\mathcal{Z}$ and  denote by $\mathcal{H}_k$   the reproducing kernel Hilbert space (RKHS) over $\mathcal{Z}$ with $k$. 
Examples of kernels include the \textit{Radial Basis Function (RBF)} or \textit{Gaussian} kernel defined as $k(x, y) := \exp\left\{-\frac{1}{2 l^2} \|x-y\|^2 \right\}$ for lengthscale $l$, the \textit{polynomial} kernel of degree $d$ defined as $k(x,y) := (\left<x,y \right> + c)^d$ for some constant $c$ and the \textit{exponential} kernel defined as $k(x,y) := \exp\{\left<x,y \right>\}$. Apart from the reproducing property, there are two other useful properties of kernels that we use in this paper. First, \textit{characteristic} kernels are kernels that induce an injection between probability measures and kernel mean embeddings, more formally:
\begin{definition}[\textbf{Characteristic kernel}. \cite{sriperumbudur2011universality}]
A bounded and measurable kernel $k$ is called characteristic if the map $f: \mathcal{P}_{\Z} \rightarrow \mathcal{H}_k$ defined as 
$
    f(\P) = \mu_{\P}
$
is an injection. 
\end{definition}
This property guarantees that the MMD is a divergence, meaning that 
\begin{talign*}
\MMD_k(\P, \Q) = 0 \iff \P \equiv \Q.
\end{talign*}
Another useful property that a kernel may possess is \textit{translation invariance}. This occurs when a linear perturbation of both inputs does not change the kernel value. 
\begin{definition}[\textbf{Translation invariant kernel}.]
The kernel $k$ is translation invariant if there exists a bounded, continuous, real-valued, positive definite function $\psi$  on $\X \subseteq \R^d$ such that $k(x,y) = \psi(x-y)$ for all $x,y \in \X$.
\end{definition}

We assume  throughout this work that the following standard condition on $k$ holds:
\begin{assumption}
 \label{asm:k}
 The kernel $k: \Z \times \Z \rightarrow \R$ is $\mathfrak{S}_\mathcal{Z}$-measurable and satisfies $|k(z,z^\prime)|\leq 1$ for any $z,z^\prime \in \Z$, where $|\cdot|$ denotes the euclidean norm . Moreover, $\phi: \mathcal{Z} \rightarrow \mathcal{H}_k$ is such that $k((x,y),(x',y')) = \left<\phi(x,y), \phi(x',y') \right>_{\mathcal{H}_k}$.
\end{assumption}
Next,  we let $k_X$ be a kernel on $\mathcal{X}$, $k_Y$ be a kernel  on $\mathcal{Y}$ and we denote by $k_X\otimes k_Y$  the product kernel on $\mathcal{Z}$ such that $k_X\otimes k_Y((x,y),(x',y')) = k_X(x,x') k_Y (y,y')$ for all $(x,y),(x',y')\in \mathcal{Z}$.
With this notation in place, we further assume:
\begin{assumption}
 \label{asm:kxy}
 There is a continuous kernel $k_X: \X \times \X \rightarrow \R$ and a kernel $k_Y: \Y \times \Y \rightarrow \R$ such that $|k_X(x, x^\prime)|\leq 1 $ for all $ x,x^\prime \in \X$, $|k_Y(y, y^\prime)|\leq 1 $ for all $ y,y^\prime \in \Y$ and $k=k_X\otimes k_Y$.
\end{assumption}
As noted in \cite{alquier2024universal}, from Theorems 3-4 in \cite{szabo2017characteristic}  it follows that under Assumption \ref{asm:kxy}, $k$ is characteristic if  $k_X$ and $k_Y$   are characteristic, continuous, translation invariant and bounded. When $\X=\R^{d_x}$ and $\mathcal{Y}=\R^{d_y}$ examples of such kernels $k_X$ and $k_Y$ include the widely-used RBF, exponential kernel and  Mat\'ern kernels. Here we restrict our results to the case where $k_X$ is a translation-invariant kernel:
\begin{assumption} \label{asm-kx}
    There exists a bounded, continuous, real-valued, positive-definite function $\psi$  on $\mathbb{R}^{d_x}$ such that $k_X(x,x') = \psi(x-x')$ for all $x,x^\prime \in \R^{d_\X}$. Here positive-definite means that for every points
$x_1,\dots,x_n\in\mathbb{R}^{d_\X}$, and coefficients $c_1,\dots,c_n\in\mathbb{R}$,
\[
\sum_{i=1}^n\sum_{j=1}^n c_i c_j\,\psi(x_i-x_j)\;\ge\;0.
\]
\end{assumption}
While this assumption is not strictly necessary for deriving generalisation error bounds, it demonstrates the error decomposition relative to the MMD between the true ME distribution and the prior centering measure. This choice also makes intuitive sense for the case of additive ME, since under a translation invariant kernel the distance between any two points $x + \nu_1$ and $x + \nu_2$ is characterized by the distance between $\nu_1$ and $\nu_2$.

Lastly, we make an assumption about the existence of a bounded conditional mean operator for the model $(\P_{g(\theta,x)})_{x \in \mathcal{X}}$, which was introduced in \cite{alquier2024universal}. Define $\P_X := \frac{1}{n} \sum_{i=1}^n \P^i$ for $\P^i$ as in (\ref{eq:p-i}). The conditional mean embedding operator is used to bound the MMD between $\P_{\theta} := \P_X \times \P_{g(\theta, \cdot)}$ as defined in (\ref{eq-bar-p-theta-exact}) and ${\P^0_{\theta}}^{\textcolor{red}{\texttt{Classical}}}  = \P^0_X \times \P_{g(\theta,\cdot)}$ or  ${\P^0_{\theta}}^{\textcolor{blue}{\texttt{Berkson}}} = \P^{0,n}_X \times\P_{g(\theta,\cdot)}$, in terms of the MMD between $\P_X$ and $\P_X^0$ or $\P_X$ and $\P_X^{0,n}$ respectively. Under Assumptions \ref{asm:k} and \ref{asm:kxy}, the conditional mean embedding operator of a probability measure $\P = \P_X \times\P_{Y \given X} \in \mathcal{P}_\mathcal{Z}$ is a mapping $\mathcal{C}_{Y \given X} : \HkX \rightarrow \HkY$ such that $\mathcal{C}_{Y \given X} k_{X}(x, \cdot) = \mu_{\P_{Y \given X=x}}$ where $\mu_{\P_{Y \given X=x}}$ is the kernel mean embedding of $\P_{Y \given X = x}$. 
Using this definition we make the following assumption:
\begin{assumption} \label{asm-lambda}
    For all $\theta \in \Theta$, there exists a bounded linear conditional mean operator $\mathcal{C}_{\theta} : \mathcal{H}_{\mathcal{X}} \rightarrow \mathcal{H}_{\mathcal{Y}}$ for $(\P_{g(\theta,x)})_{x \in \mathcal{X}}$ and $\Lambda := \sup_{\theta \in \Theta} \|\mathcal{C}_{\theta} \|_{o} < \infty$, for $\| \cdot \|_o$ the operator norm. 
\end{assumption}
A more detailed analysis of kernel mean embeddings can be found in \cite{klebanov2020rigorous}. We refer the reader to \cite{alquier2024universal} for a thorough discussion on sufficient conditions for the existence of a bounded linear conditional mean operator. 
Table \ref{tab:prob-meas-defs} summarises the notation for the target joint and marginal distributions alongside the equivalent joint model probability measures for each of the two types of ME. 

\subsection{Main results} \label{sec-theory-main}
We now bound the generalisation error of our method. Proofs for all theoretical results of this section are provided in Appendix \ref{app-proofs}. Note that even when the functional relationship $g$ is known, $\P_{g(\theta, \cdot)}$ might still be misspecified so the expected error is trivially lower bounded by $\inf_{\theta \in \Theta} \MMD_k(\P^{0}, \P_{\theta}^{0})$. Although this type of model misspecification is not directly targeted in this work, we expect to naturally remain robust to such misspecifications due to the robust nature of minimum MMD estimators as studied in \cite{briol2019statistical, cherief2022finite, dellaporta2022robust}. 
%

%%%
\begin{theorem}[\textcolor{blue}{\texttt{Berkson}}] \label{thm-main}
Let $k^2_X = k_X \otimes k_X$ and $\P^0 = \P_X^{0,n} \times \P_{Y \given X}^0$. Then under Assumptions \ref{asm:k}-\ref{asm-lambda} and $\Lambda$ as in Assumption \ref{asm-lambda} we have: 
  \begin{talign*} 
&\overbrace{\mathbb{E}_{(x_i,y_i)_{1:n} \sim \P^{0}} \left[ \E_{\P \sim \mu} \left[ \MMD_k(\P^{0}, \P_{\theta^\star(\P)}^{0}) \right] \right] - \inf_{\theta \in \Theta}  \MMD_k( \P^{0}, \P^{0}_{\theta})}^{\text{Generalisation error}}\\
    &\quad \leq \frac{1}{\sqrt{n}} \left(2 + \frac{2(1+\Lambda)}{\sqrt{c+2}} \right) 
    + \underbrace{\frac{2c}{c+1} \left( \Lambda  \MMD_{k^2_X}(\F_N, \F_N^0) + \sqrt{{\text{Var}}_{\mathcal{H}_{k_X}}[\F_N^0] + \MMD^2_{k_X}(\F_N^0, \F_N)}\right)}_{\text{Prior specification \& ME RKHS variance}} \\
    &\quad \quad  +\underbrace{\frac{2}{c+1} \left(\Lambda \MMD_{k^2_X}(\F_N^{0}, \delta_{0})  + \sqrt{2} \sqrt{\MMD_{k_X}(\F_N^0, \delta_0)}\right)}_{\text{ME deviation from 0}}.
\end{talign*}
where $\P := (\P^1, \dots, \P^n) \sim \mu$, $\E_{\P \sim \mu}$ denotes expectation under: $\{\P^i \sim \mu_i\}_{1:n}$ and $\text{Var}_{\mathcal{H}_{k_X}}[\F_N^0]$ denotes the RKHS variance of $\F_N^0$, i.e. ${\text{Var}}_{\mathcal{H}_{k_X}}[\F_N^0] := \E_{\nu \sim \F_N^0}[k_X(\nu, \nu)] - \left<\mu_{\F_N^0}, \mu_{\F_N^0} \right>_{\mathcal{H}_{k_X}}$.
\end{theorem}
Recall that for any set $A \in \sigma_{\X}$, $\P^0_{X \given w_i}(A) = \F_N^0(A_{w_i})$. As expected, the upper bound of the generalisation error depends on the DP prior both through the concentration parameter $c$ and the quality of the prior $\F_{N}$ with respect to $\F_N^0$ as well as the deviation of the true ME from zero and its RKHS variance. Theorem \ref{thm-main} highlights an error trade-off based on the selection of $c$ and the prior beliefs about $X \given W = w_i$, providing a mathematical interpretation of the prior elicitation discussion in Section \ref{sec-prior-elicit}. The three discussed cases can be understood as follows: when $c \rightarrow \infty$, all of the posterior weight in (\ref{eq-dp-posterior}) is placed on the prior centering measure $\F_{N}$ and consequently the generalisation error upper bound depends on the quality of $\F_{N}$ as well as a measure of dispersion of the true ME distribution. This is because if the ME is more dispersed, this leads to the distribution of $X \mid W$ being more dispersed, hence making the problem harder. When $c = 0$ the posterior centring measure is concentrated on the observed data, hence the generalisation error upper bound depends on the $\sqrt{n}$ term and the last term, which quantifies the distance of the true ME terms $\nu^0_{1:n}$ from zero. Finally, $c = 1$ corresponds to when the weight in (\ref{eq-dp-posterior}) has been equally divided between $\F_N$ and the observations. Similarly, in Theorem \ref{thm-main}, the last two terms are equally weighted, and indexed by the same constant.

In the Classical error case, we get a similar bound which we present in Theorem $\ref{thm-main-class}$. Unsurprisingly, the error bound now depends on the quality of $\F_{X \given w_i}$ with respect to $\P_X^0$.
\begin{theorem}[\textcolor{red}{\texttt{Classical}}] \label{thm-main-class}
Let $k^2_X = k_X \otimes k_X$, $\P^0 = \P_X^0 \times \P_{Y \given X}^0$ and $\F := \frac{1}{n} \sum_{i=1}^n \F_{X \given w_i}$. Then under Assumptions \ref{asm:k}-\ref{asm-lambda} and $\Lambda$ as in Assumption \ref{asm-lambda} we have: 
  \begin{talign*}               &\overbrace{\mathbb{E}_{(x,y)_{1:n} \sim \P^{0}} \left[ \E_{\P \sim \mu} \left[ \MMD_k(\P^{0}, \P_{\theta^\star(\P)}^{0}) \right] \right] - \inf_{\theta \in \Theta}  \MMD_k( \P^{0}, \P^{0}_{\theta})}^{\text{Generalisation error}}\\
    &  \leq\frac{1}{\sqrt{n}} \left(\frac{2 \Lambda}{c+1} + \frac{2 \Lambda + 2}{\sqrt{c+2}} \right) \\
    &\quad + \underbrace{\frac{2c}{c+1} \left(\Lambda  \MMD_{k^2_X}(\P_X^0, \F) + \frac{1}{n} \sum_{i=1}^n \sqrt{{\text{Var}}_{\mathcal{H}_{k_X}}[\P_X^0] + \MMD^2_{k_X}(\P_X^0, \F_{X \mid w_i})} \right)}_{\text{Prior specification \& RKHS variance of $X$}} \\
    &\quad + \underbrace{\frac{2}{c+1} \left(\MMD_{k^2_X}(\F_N^{0}, \delta_{0}) + \sqrt{2} \sqrt{\MMD_{k_X}(\F_N^0, \delta_0)} \right)}_{\text{ME deviation from 0}}.
\end{talign*}
where $\P := (\P^1, \dots, \P^n) \sim \mu$ and $\E_{\P \sim \mu}$ denotes expectation under: $\{\P^i \sim \mu_i\}_{1:n}$.
\end{theorem}
Similarly to the Berkson case, the bound on Theorem \ref{thm-main-class} depends on the prior specification. $\F$ incorporates prior beliefs about both the marginal distribution of $X$ through $\F_X$ and about the ME distribution through $\F_{W \given X}$. Hence, the quality of the prior depends on how well $\F_{X \given w_i}$ describes $\P_X^0$. Moreover, similarly to the Berkson case, the prior specification part of the upper bound also depends on an RKHS variance term. In this case, the error depends on the variance of the marginal distribution of $X$. Intuitively, this is unsurprising since a larger dispersion of the distribution of $X$ leads to a larger variance of the distribution of $W \mid X$ hence making the estimation problem harder. This is more clearly illustrated in the special case of the Gaussian example presented in the next section.

\subsection{Special cases} \label{sec-theory-cor}
Theorems \ref{thm-main} and \ref{thm-main-class} provide generalisation error bounds for any prior beliefs about the distributions of $X, X \given W$ or $N$. 
In many real-world cases \citep[see e.g.][]{berry2002bayesian, fung1999measurement}, the distribution of ME is known or can be estimated. In the Berkson error case, when $\F^0_{N}$ is known or can be estimated, we can obtain a tighter bound for the generalisation error, which vanishes as the data sample size grows large, provided that we set $c$ large enough. 
\begin{corollary}[\textcolor{blue}{\texttt{Berkson}}] \label{cor-known-dist}
   Let $\F_{N} \equiv \F_{N}^0$. Then, if $k_X$ is characteristic:
    \begin{talign*} 
    &\mathbb{E}_{(x,y)_{1:n} \sim \P^{0}} \left[ \E_{\P \sim \mu} \left[  \MMD_k(\P^{0}, \P_{\theta^\star(\P)}^{0}) \right] \right] - \inf_{\theta \in \Theta} \MMD_k(\P^{0}, \P^{0}_{\theta}) \nonumber \\
     &\quad \leq  \frac{1}{\sqrt{n}} \left(2 + \frac{2(1+\Lambda)}{\sqrt{c+2}} \right) + \frac{2c}{c+1} \sqrt{{\text{Var}}_{\mathcal{H}_{k_X}}[\F_N^0]}  \\
     &\quad + \frac{2}{c+1} \left(\Lambda \MMD_{k^2_X}(\F_N^{0}, \delta_{0})  + \sqrt{2} \sqrt{\MMD_{k_X}(\F_N^0, \delta_0)}\right). 
\end{talign*}
\end{corollary}
In the above corollary, $k_X$ being characteristic is needed to ensure that MMD is a probability metric, i.e. $\MMD_{k_{X}}(\P, \Q) = 0 \iff \P \equiv \Q$ \citep{sriperumbudur2010hilbert}. We note that if $\F_{N}^0$ is known a priori, then it would make sense to set $c$ large so that the nonparametric posterior is centered around the well-specified prior. Taking the limiting case of $c \rightarrow \infty$, we observe that the upper bound of the generalisation error approaches  $\frac{1}{\sqrt{n}} \left(2 + \frac{2(1+\Lambda)}{\sqrt{c+2}} \right) + 2 \sqrt{{\text{Var}}_{\mathcal{H}_{k_X}}[\F_N^0]} $ and hence that the error decreases at a $\sqrt{n}$ rate with an additional error due to the RKHS variance of the true ME distribution. In the Classical error case, since the target joint distribution involves the marginal distribution of $X$ a similar assumption has to be made about $\P_X^0$. 
\begin{corollary}[\textcolor{red}{\texttt{Classical}}] \label{cor-known-dist-class}
Let $\P_{X | w_i} \equiv \P_X^0 \forall i \in \{1, \dots, n\}$. Then, if $k_X$ is characteristic:
    \begin{talign*} 
    &\mathbb{E}_{(x,y)_{1:n} \sim \P^{0}} \left[ \E_{\P \sim \mu} \left[  \MMD_k(\P^{0}, \P_{\theta^\star(\P)}^{0}) \right] \right] - \inf_{\theta \in \Theta} \MMD_k(\P^{0}, \P^{0}_{\theta}) \nonumber \\
     &\quad     \leq\frac{1}{\sqrt{n}} \left(\frac{2 \Lambda}{c+1} + \frac{2 \Lambda + 2}{\sqrt{c+2}} \right) + \frac{2c}{c+1}  \sqrt{{\text{Var}}_{\mathcal{H}_{k_X}}[\P_X^0]} \\
    &\quad \quad + \frac{2}{c+1} \left(\MMD_{k^2_X}(\F_N^{0}, \delta_{0}) + \sqrt{2} \sqrt{\MMD_{k_X}(\F_N^0, \delta_0)} \right).
\end{talign*}
\end{corollary}
Another common assumption is that the ME or the true covariate follows a Gaussian distribution \citep[e.g.][]{spiegelman1991cost}. 
% cochran1968errors,
While not an assumption of our method, we explore this specific scenario when $k_{X}$ is the Gaussian kernel with parameter $l> 0$. 
This allows for a closed form of the bound in Theorem \ref{thm-main} for $\X \subseteq \R^{d_\X}$. 
\begin{corollary}[\textcolor{blue}{\texttt{Berkson}}] \label{cor-rbf-gauss}
   Let $k_{X}$ be the RBF kernel with lengthscale $l > 0$, $\F_{N}^0 = \mathcal{N}(0, \sigma_1^2 I)$ and $\F_{N} = \mathcal{N}(0, \sigma_2^2 I)$. Then, 
   \begin{talign*}
           &\mathbb{E}_{(x,y)_{1:n} \sim \P^{0}} \left[ \E_{\P \sim \mu} \left[  \MMD_k(\P^{0}, \P_{\theta^\star(\P)}^{0}) \right] \right] - \inf_{\theta \in \Theta}  \MMD_k(\P^{0}, \P^{0}_{\theta}) \\
       &\quad \quad \leq  \frac{1}{\sqrt{n}} \left(2 + \frac{2(1+\Lambda)}{\sqrt{c+2}} \right)  + \frac{2c}{c+1} \left( \Lambda  C_1 + C_2\right) + \frac{2}{c+1} \left(\Lambda C_3  + \sqrt{2} \sqrt{C_4}\right) 
\end{talign*}

where 
\begin{align*}
    C_1^2 &:= \left(\frac{l^2}{l^2 + 4\sigma_1^2}\right)^{\frac{d_\X}{2}} -  2\left(\frac{l^2}{l^2 + 2\sigma_1^2 + 2\sigma_2^2}\right)^{\frac{d_\X}{2}} + \left(\frac{l^2}{l^2 + 4\sigma_2^2}\right)^{\frac{d_\X}{2}} \\
    C_3^2 &:= \left(\frac{l^2}{l^2 + 4 \sigma^2_1} \right)^{\frac{d_\X}{2}} -2 \left(\frac{l^2}{l^2 + 2 \sigma^2_1} \right)^{\frac{d_\X}{2}} + 1 \\
    C_2^2 &:= 1 -  2\left(\frac{l^2}{l^2 + \sigma_1^2 + \sigma_2^2}\right)^{\frac{d_\X}{2}} + \left(\frac{l^2}{l^2 + 2\sigma_2^2}\right)^{\frac{d_\X}{2}} \\
    C_4^2 &:= \left(\frac{l^2}{l^2 + 2 \sigma^2_1} \right)^{\frac{d_\X}{2}} -2 \left(\frac{l^2}{l^2 +  \sigma^2_1} \right)^{\frac{d_\X}{2}} + 1. 
\end{align*}

\end{corollary}
As expected, in the absence of ME, that is when $\sigma_1 \approx 0$, we have $C_3 \approx 0$ and $C_4 \approx 0$. Similarly, the closer $\sigma_1$ is to $\sigma_2$, the smaller $C_1$ and $C_2$ are. This occurs because $\F_{N}$ is closer to the true ME distribution $\F_{N}^0$. Under the further assumption that the ME variance is known, i.e. $\sigma_1 = \sigma_2$, as made by various existing methods (see Section \ref{sec:intro}), we have $C_1 = 0$. In the Classical case, we get the following closed form when the true marginal distribution of $X$ is a Gaussian and the Gaussian priors are chosen for both the ME distribution and the marginal of $X$ as in the example in (\ref{eq-bayesian-post}).

\begin{corollary}[\textcolor{red}{\texttt{Classical}}] \label{cor-rbf-gauss-class}
    Let $k_X$ be the RBF kernel with lengthscale $l > 0$, $\F_{N}^0 = \mathcal{N}(0, \sigma_1^2 I)$, $\P_X^0 := \mathcal{N}(\mu_X, \Sigma_X)$ and $\F_{X \mid w_i} := \mathcal{N}(m_i, \sigma_{\text{post}}^2 I)$ obtained by the closed-form posterior in (\ref{eq-bayesian-post}) using the prior $X \sim \mathcal{N}(m_X, \sigma^2 I)$ and $w_i \mid X \sim \mathcal{N}(X, \sigma^2_{N}I)$, i.e. $m_i := \frac{\sigma^2}{\sigma^2 + \sigma^2_N} w_i + \left(1 - \frac{\sigma^2}{\sigma^2 + \sigma^2_N}  \right) m_X$ and $\sigma^2_{\text{post}} = \frac{\sigma^2 \sigma^2_N}{\sigma^2 + \sigma^2_N} $. Then, 
    \begin{talign*}
        &\mathbb{E}_{(x,y)_{1:n} \sim \P^{0}} \left[ \E_{\P \sim \mu} \left[  \MMD_k(\P^{0}, \P_{\theta^\star(\P)}^{0}) \right] \right] - \inf_{\theta \in \Theta} \MMD_k(\P^{0}, \P^{0}_{\theta})\\
        &\quad \leq \frac{1}{\sqrt{n}} \left(\frac{2 \Lambda}{c+1} + \frac{2 \Lambda + 2}{\sqrt{c+2}} \right) \\
        & \quad \quad + \frac{2c}{c+1} \left(\Lambda C_1 + \frac{1}{n} \sum_{i=1}^n \sqrt{C_2 + C_i^2} \right) + \frac{2}{c+1} \left(\Lambda C_3 + \sqrt{2} \sqrt{C_4}\right)
    \end{talign*}
    where 
    \begin{talign*}
        C_i^2 &:= \left| I + \tfrac{2\sigma_{\text{post}}^2}{\ell^2} I \right|^{-1/2}
+ \left| I + \tfrac{2}{\ell^2}\Sigma_X \right|^{-1/2}
\\
&\quad- 2 \left| I + \tfrac{1}{\ell^2}\big(\sigma_{\text{post}}^2 I + \Sigma_X\big) \right|^{-1/2}
\exp\!\left(
  -\tfrac{1}{2}\Delta_i^\top
  \big(\sigma_{\text{post}}^2 I + \Sigma_X + \ell^2 I\big)^{-1}
  \Delta_i
\right)
    \end{talign*}
for $\Delta_i = m_i - \mu_X$, 
\begin{talign*}
    C_2 &:= \left| I + \tfrac{2}{\ell^2}\Sigma_X \right|^{-1/2} \\
    C_1^2 &:= \frac{1}{n^2}\sum_{i=1}^n\sum_{j=1}^n \Big(1+\tfrac{4\sigma_{\text{post}}^2}{\ell^2}\Big)^{-d/2}
   \exp\!\left(-\tfrac{\|m_i-m_j\|^2}{4\sigma_{\text{post}}^2+\ell^2}\right)
\\ 
&\quad + \Big(1+\tfrac{4\sigma^2}{\ell^2}\Big)^{-d/2}
- \frac{2}{n}\sum_{i=1}^n
\Big(1+\tfrac{2(\sigma_{\text{post}}^2+\sigma^2)}{\ell^2}\Big)^{-d/2}
   \exp\!\left(-\tfrac{\|m_i - m_X\|^2}{2(\sigma_{\text{post}}^2+\sigma^2)+\ell^2}\right)
\end{talign*}
and $C_3, C_4$ are as in Corollary \ref{cor-rbf-gauss}
\end{corollary}

Finally, under the assumption of RBF kernels, we prove that the generalisation error bounds presented in Theorems \ref{thm-main} and \ref{thm-main-class} are minimax optimal in terms of the rate in $n$ by employing Le Cam's method \citep[see][Lemma 15.9]{Wainwright_2019}. Intuitively, this result follows from the fact that both bounds rely on a vanishing term of rate $\mathcal{O}(n^{- \frac{1}{2}})$ depending on the concentration rate of a minimum MMD estimator and two error terms which are independent of the sample size $n$ and dependent on the ME structure and the DP prior choice. As a result, we expect the $n$-rate to be at best the rate obtained in standard MMD estimation, without the added ME structure, which is $\mathcal{O}(n^{-1/2})$ \citep[as in][]{briol2019statistical, cherief2022finite, dellaporta2022robust}. Indeed, the minimax rate optimality of an MMD estimator under an RBF kernel has already been established by \cite{tolstikhin2016minimax}. Here, we adapt this result to our setting and show minimax optimality of our generalisation error bounds in the case of an RBF kernel. 
\begin{proposition}\label{prp:minimax}
   Let Assumption \ref{asm:kxy} hold and assume both $k_X$ and $k_Y$ are RBF kernels. Then for any inferential procedure that returns $\hat{\theta}_n$, given observations $\{(w_i, y_i)\}_{i=1}^n$ we have:
    \begin{align*}
        \inf_{\hat{\theta}_n} \sup_{\P\in \mathcal{P}_{\X\times \Z}} \E_{(w_i,x_i,y_i)_{1:n} \sim \P}[\MMD_k(\P^0, \P_{\hat{\theta}_n})] \gtrsim n^{-\frac{1}{2}}.
    \end{align*}
\end{proposition}
Note that here, $\P\in \mathcal{P}_{\X\times \Z}$ denotes any data distribution of the random variables $(X,W,Y)$ and hence contains all the possible ME structures, including the Berkson and Classical case. On the other hand, the DGP $\P_0$ of $(X,Y)$ is obtained from $\P$ after marginalising out $W$, similarly to the objective of Theorems \ref{thm-main} and \ref{thm-main-class}. 

\section{Empirical Assesment} \label{sec:sim}
In this section, we empirically evaluate the performance of our method in the presence of Berkson and Classical measurement errors with synthetic data. Code to reproduce all experiments can be found at \url{https://github.com/haritadell/me-models}. 
\paragraph{Linear regression} \label{sec-exp-lin-reg}
We first examine the performance  of Robust-MEM with the TLS loss. As is widely known \citep[see e.g.][]{berkson1950there}, Berkson ME does not generally lead to biased inference outcomes in linear regression so we consider the Classical error case. However, it's worth noting that there are cases where ignoring Berkson ME in linear regression would still introduce bias, for example in the presence of unobserved confounders as discussed in \cite{haber2021bias}. Consider the linear regression model:
$
    y = \theta_1 + \theta_2 x + \epsilon, \quad
    w = x + \nu$ where $ \epsilon \sim N(0, \sigma^2_{\epsilon}), \quad \nu \sim N(0, \sigma^2_{\nu}).
$
We sample $n = 100$ points from a Gaussian distribution $x_{1:n} \sim N(10,4)$ and perform inference for the parameters $\theta_0 = (\theta_1, \theta_2) = (1,2)$. In addition to Robust-MEM with the TLS loss, we include results from frequentist OLS and TLS estimators, as well as the SIMEX algorithm \citep{cook1994simulation} provided by the \texttt{simex} R package. We examine cases where the ME distribution is not precisely known, using a prior centering measure $\F_N = N(0,0.2)$ for various true values of $\sigma_{\nu}^2$ and the same a priori variance of $0.2$ in the SIMEX algorithm. Results for a well-specified prior are provided in Appendix \ref{app:sec-lin-wellsp}. 
In Fig. \ref{fig:class-lin} we observe that the Robust-MEM (TLS) method remains robust to ME even under a misspecified prior variance and high values of ME variance. In the absence of ME, OLS performs best with all methods performing well. More detailed results, including the MSE and variance can be found in Table \ref{tab:class-lin} (Appendix \ref{app:sec-lin-wellsp}).
\begin{figure}[ht!]
    \center
\includegraphics[width=0.9\textwidth]{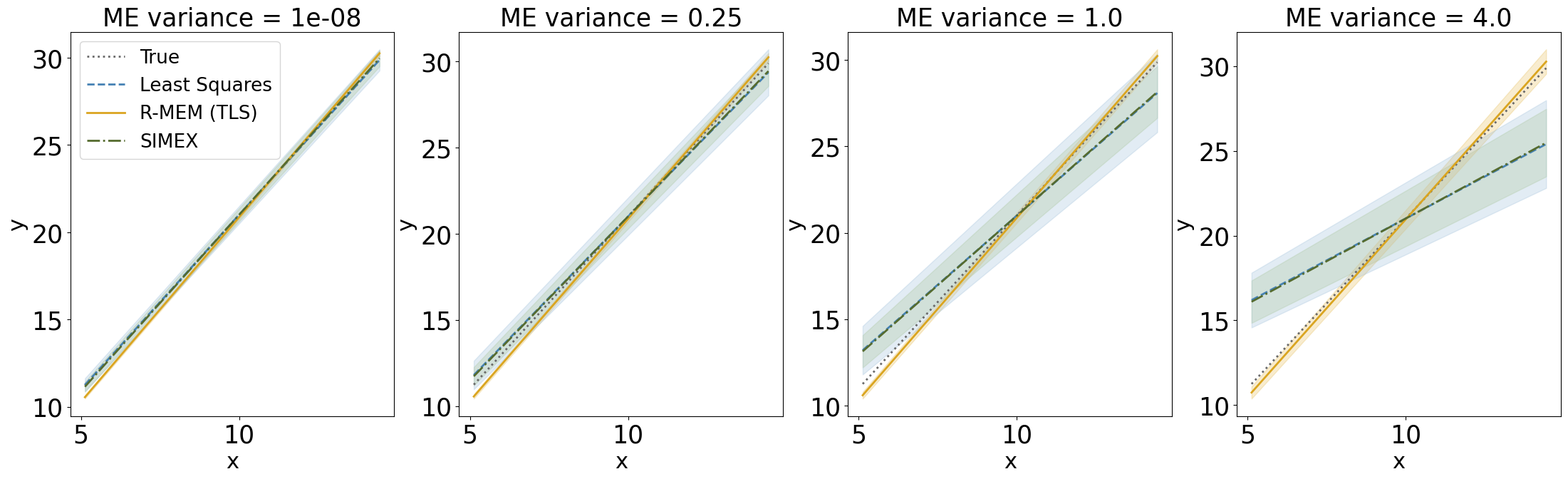}
 \caption{Fitted linear regression line over 100 simulations with a misspecified ME variance. The true line is shown along with (R-MEM (TLS)), nonlinear LS (Least Squares) and SIMEX.}
  \label{fig:class-lin}
\end{figure}
 \paragraph{Nonlinear regression} \label{sec-nonlin-reg}
Directly addressing the TLS objective in nonlinear regression can result in inconsistent parameter estimates \citep{kukush2002consistent, van2004total} so we consider the Robust-MEM (MMD) with a sigmoid-shaped curve defined by:
$$
    y = \frac{\exp(a + bx)}{1 + \exp(a + bx)} + \epsilon, \quad \epsilon \sim N(0, \sigma^2_{\epsilon}) 
$$
with parameter of interest $\theta = (a, b)$. We generate $n=200$ observations with $\theta_0 = (1,3)$ and $\sigma_{\epsilon} = 0.5$. We further consider the frequentist minimum MMD estimator of \cite{briol2019statistical} to investigate how much of the robustness is gained by the NPL framework in comparison to the choice of loss function alone. We consider Berkson ME where the true ME distribution is Gaussian  with increasing variance whereas the prior centering measure is set to be a Student-t distribution with 3 degrees of freedom. The results for the well-specified prior and the Classical error cases are reported in Appendix \ref{app:sec-nonlin-berk-wellsp}, \ref{app-nonlin-class}. 
We observe (Fig. \ref{fig:berk-nonlin-misp} \& Tab. \ref{tab:berk-mse-misp}) that for small values of ME variance ($\leq 0.25$), all methods perform satisfactorily, but as the variance increases Robust-MEM (MMD) remains the most robust to ME while also maintaining good coverage probability.   
\begin{figure}[ht!]
    \center
\includegraphics[width=0.9\textwidth]{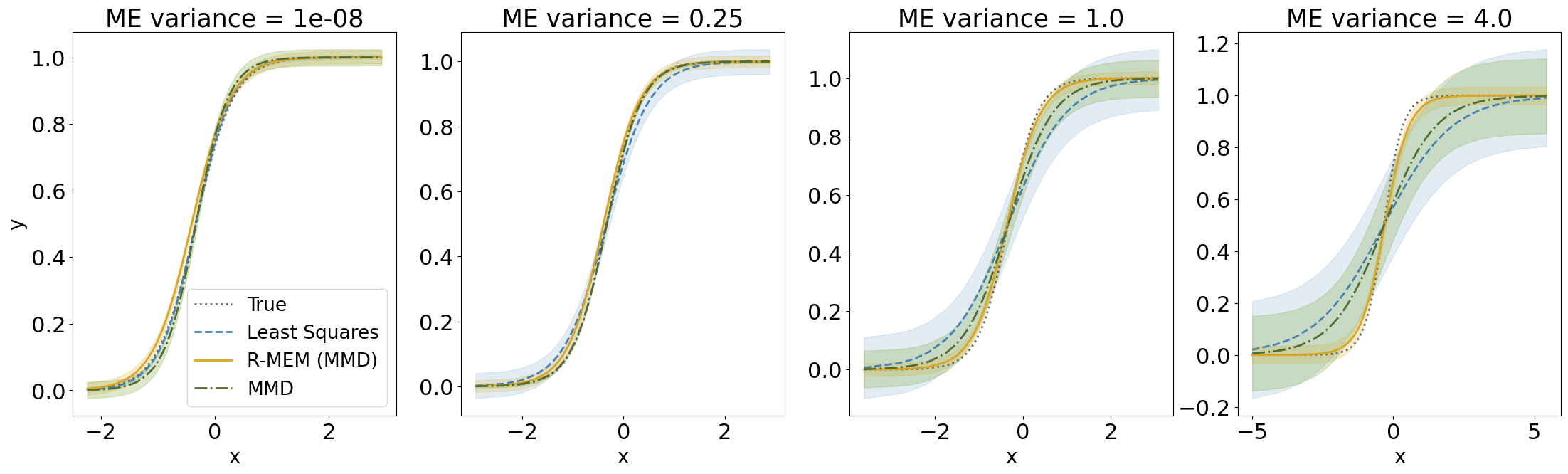}
 \caption{Model fit for the sigmoid curve regression function with Berkson ME for an increasing ME variance over 100 simulations. The true model fit is shown along with our method (R-MEM (MMD)), nonlinear least squares (Least Squares) and minimum MMD estimator (MMD).}
  \label{fig:berk-nonlin-misp}
\end{figure}
\begin{table}[h]
\small
\centering
\caption{Nonlinear Regression example with Berkson ME (over 100 simulations) with a misspecified prior. Top: MSE with associated std. for $\theta = (\theta_1, \theta_2)$. Bottom: Coverage probability (based on the 90\% posterior credible intervals) for the Robust-MEM (MMD) method.}
\label{tab:berk-mse-misp}
\begin{tabular}{llllll}
\hline
Method/Metric &  & \multicolumn{4}{c}{ME variance}  \\
 &  & $1e-08$ & $0.25$ & $1.0$ & $4.0$ \\
\hline
\multirow[t]{2}{*}{\smaller{MMD/}} & $\theta_1$ & 0.25 (0.50) & 0.19 (0.50) & 0.21 (0.22) & 0.49 (0.34) \\
 \smaller{MSE} & $\theta_2$ & 1.34 (2.11) & 0.95 (1.39) & 1.53 (1.22) & 3.85 (1.52) \\
% \cline{1-6}
\smaller{R-MEM (MMD)/} & $\theta_1$ & \textbf{0.15 (0.26)} & \textbf{0.13 (0.23)} & \textbf{0.13 (0.18)} & \textbf{0.25 (0.29)} \\
\smaller{MSE} & $\theta_2$ & \textbf{0.19 (0.24)} & \textbf{0.27 (0.43)} & \textbf{0.43 (0.65)} & \textbf{0.77 (0.96)} \\
\cline{1-6}
 \smaller{R-MEM (MMD)/} & $	\theta_1$ & 0.86 & 0.92 & 0.94 & 0.87 \\
 \smaller{Coverage prob.}& $	\theta_2$ & 1 & 0.99 & 0.99 & 0.99 \\
\hline
\end{tabular}
\end{table}
This further demonstrates that robustness arises from the NPL framework rather than solely from the robustness properties of the MMD-based loss function.

\paragraph
{Average Computational Times}
In addition to the empirical performance results, we also recorded the 
average computational times for each method across our experiments. The experiments for the Robust-MEM and minimum MMD estimation methods were run on a \texttt{GPU}. The SIMEX method was run in \texttt{R} using the code provided by the \texttt{simex} \texttt{R} package.
The average runtimes (in seconds) per dataset are summarised in 
Table~\ref{tab:comp-times}.
\begin{table}[h!]
\centering
\caption{Average computational times (in seconds) with associated standard deviation across 100 simulation runs.}
\label{tab:comp-times}
\begin{tabular}{lcc}
\hline
Method & Linear Regression & Nonlinear Regression \\
\hline
OLS             & 0.001 (0.0003) & --      \\
TLS             & 0.003 (0.0008) & --      \\
SIMEX           & 0.29 (0.03) \\
MMD & -- &   1.78 (0.23)   \\
Robust-MEM (TLS)& 5.41 (0.45) & --      \\
Robust-MEM (MMD)& -- & 13.25 (3.52)  \\
\hline
\end{tabular}
\end{table}
These results indicate that while the proposed Robust-MEM methods incur 
a higher computational cost relative to simpler approaches, such as OLS or 
TLS, the runtimes remain practical for the problem sizes considered due to the parallelisation property of the Posterior Bootstrap.

\section{Berkson and Classical Measurement Error Studies} \label{sec:stu}
In this section, we examine the application of our framework in case studies from Econometrics and Mental Health Science in the presence of either Berkson or Classical ME.
\paragraph{California Test Scores} \label{sec-exp-cas}
We first study the  \texttt{CASchools} dataset, available from the \texttt{AER} R package, which provides data on various school districts in California between 1998 and 1999. 
We employ a polynomial regression model to study the effect of school's income on students' test scores, as studied in \cite{hanck2021introduction}, by considering a single outcome variable `Score' to be the average score between maths and reading. To introduce Berkson ME in the dataset we consider the case of group averages which arises when the individual measurements of a covariate are missing but we have data on the group level. People commonly assign individuals a group value based on the group they belong to, disregarding the resulting Berkson ME. To mimic this, we divide the income values range into 15 levels and we assign the within-group average income to all schools of the same income level. This creates Berkson ME in the covariates as each school in level income $i \in \{1, \dots, 15\}$ has some true income $x$, but we assign to it the group average income $w_i$ which means that $x = w_i + \nu$ for some unknown error $\nu$, independent of $w_i$. We consider a second degree polynomial regression model, i.e.     
$Y = a + bx + cx^2 + \epsilon$ where $\epsilon \sim N(0, \sigma_{\epsilon}^2)$. 
The regression fit for 200 posterior samples along with the true and observed covariates is shown in Fig. \ref{fig:cas}. As an oracle, we also plot a standard nonlinear Least Squares fit to the clean data, i.e. the data without ME. To assess the coverage of the suggested methods we repeat the procedure of obtaining posterior samples 100 times with different random seeds. We then compute the coverage probability as the percentage of times that the 90 \% credible interval of the posterior included the true parameter values. 
\begin{figure}[ht!]
     \centering
         \centering
         \includegraphics[width=0.85\textwidth]{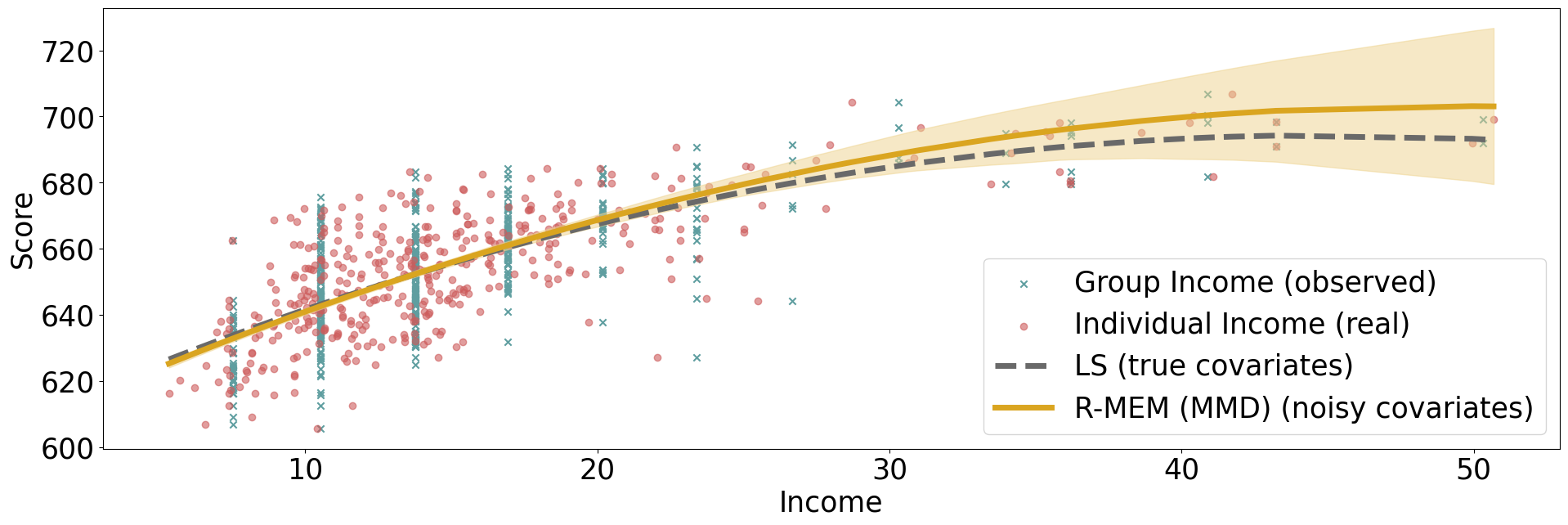}
        \caption{Posterior mean estimates with pointwise 90\% credible sets obtained by Robust-MEM (MMD). Scattered points indicate the observed group level average incomes (w) and the unobserved individual incomes (x) with associated scores (y).} 
        \label{fig:cas}
\end{figure}
Again, as true parameter values here we take the coefficients obtained by fitting a standard nonlinear Least Squares method to the non noisy data, i.e. the individual incomes and scores. The coverage probability for both parameters over 100 simulations was 100\% at a 90\% credible level whereas the average mean squared error of the residuals conditional on the true covariates over 100 runs was 0.25 with a standard deviation of 0.37. 
\paragraph{Mental Health Study} \label{sec-exp-berry}
 We examine a treatment effect estimation problem in the presence of Classical ME, in a mental health study introduced in \cite{berry2002bayesian}. The associated dataset and code for inference were later released in \cite{harezlak2018semiparametric} and code for the semiparametric Bayesian approach with penalized splines can be found in the \texttt{HRW} R package. The study consists of data from a six-week clinical trial of a drug versus a placebo. A physician-assessed score of the patient's mental health is given at baseline and the score is noted again at the end of treatment (Fig. \ref{fig:berry}). Let $I_i$ be the indicator function for the observation $(w_i, y_i)$ being in the placebo group, i.e. $I_i = 0$ if patient $i$ is in the treatment group and $I_i = 1$ otherwise. 
The aim is to model the relationship between the scores at the end of treatment and at baseline, for the treatment and placebo groups, in order to estimate the average treatment effect (ATE) of the drug. The ATE here is defined as the difference between the mean response scores conditional on the true baseline scores, denoted by: 
\begin{align*}
   ATE(x) &= \E[Y \given X = x, I = 0] - \E[Y \given X = x, I = 1] :=g^1(\theta^1,x) - g^0(\theta^0,x)
\end{align*}
where $g^0(\theta^0,x)$ and $g^1(\theta^1,x)$ correspond to the mean response scores, conditional on the baseline scores, for the placebo and treatment groups respectively, indexed by parameters of interest $\theta = (\theta^0, \theta^1)$. 
Full details of the model can be found in Appendix \ref{app-bcr-model}.

Following \cite{berry2002bayesian} and \cite{harezlak2018semiparametric} we assume that the covariates $x_i$ are observed with ME such that for each $i = 1, \dots, n$ we observe $w_i$ where $x_i \given w_i \sim N(w_i, 0.35)$. We hence set $\F_{W \given X} = N(X, 0.35)$ in the DP prior and we further set $\F_X = N(0, 1)$. By conjugacy we can obtain the prior centering measure $\F_{X \given w_i}$ as in (\ref{eq-bayesian-post}). The ATE for posterior value estimates for $\theta$ is plotted in Fig. \ref{fig:berry} with pointwise 90\% credible sets along with the \texttt{HRW} approach denoted by BCR. 
\begin{figure}[ht!]
     \centering
         \centering
         \includegraphics[width=0.85\textwidth]{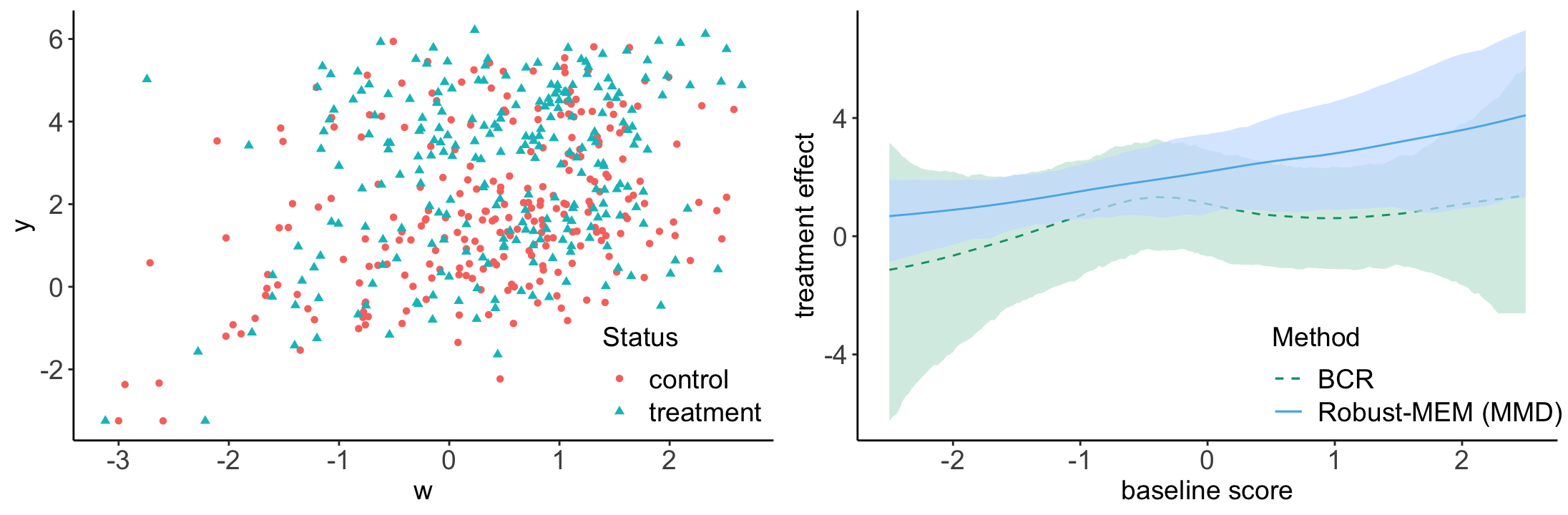}
        \caption{Left: Observed baseline scores (w) with associated end-of-treatment scores (y) for the control and treatment groups in the mental health study. Right: Posterior model estimate for $ATE$ with pointwise 90\% credible sets obtained by Robust-MEM (MMD) and the BCR method of \cite{berry2002bayesian} and \cite{harezlak2018semiparametric}.} 
        \label{fig:berry}
\end{figure}
We observe that the two estimates of the ATE differ especially for higher baseline scores. Some differences in performance with this method were also discussed by \cite{sarkar2014bayesian} who note that this might be due to the strong parametric assumptions in the model considered by \cite{berry2002bayesian}. 
\begin{table}[h]
\small
\centering
\caption{Simulation of mental health study using posterior mean parameters from the BCR method (100 simulations): residuals MSE (with std.), parameter MSE (MSE with std. compared to true values, averaged over 64 parameters), coverage probability (averaged over 64 parameters).}
\label{tab:mh}
\begin{tabular}{lrrrrr}
\hline
& \smaller{Residuals MSE} & \smaller{Parameter MSE} & \smaller{Coverage} \\
\hline
\smaller{Robust-MEM (MMD)} & 29.66 (18.85) & 0.107 (0.024) & 94\%\\
\smaller{BCR} & 101.91 (392.5) & 0.37 (1.55) & 99\%\\
\hline
\end{tabular}
\end{table}
To further assess our method and examine whether these assumptions might be a potential downfall, we simulate data from the obtained parameters of the BCR method and obtain posterior samples from each of the two methods. 
Specifically, we utilize the posterior means for each parameter obtained from the BCR method to generate new data ${(w_i,y_i)}_{i=1}^n$ according to the suggested model. In this scenario, the ground truth parameters are known, enabling us to compute the MSE of the residuals and the MSE of the posterior mean parameters compared to the ground truth. Table \ref{tab:mh} presents the average MSE of the residuals across 100 simulations of the experiment, along with their associated standard deviations. Additionally, we report the MSE with associated standard deviations of the parameter estimates, where each method's estimated parameter is taken as the posterior mean. We further evaluate the uncertainty quantification of the methods by providing the coverage probabilities of the 90\% credible posterior intervals for each parameter over 100 simulations. We observe that the suggested method outperforms BCR in terms of both residuals and parameter MSE, albeit achieving slightly lower coverage probabilities. 
\begin{remark}
    In practice there's a trade-off between MSE and coverage due to the initialization of the optimisation step in (\ref{eq-opt-step}). To achieve targeted coverage probabilities, random initialization from a wide uniform distribution is effective, increasing posterior uncertainty. Conversely, suppose a low MSE is a priority, then a small pre-processing step (Appendix \ref{app-bcr-model}) can be used to identify an initialization point.
\end{remark}
\section{Conclusion} \label{sec-conclusion}
We introduce a Bayesian NPL framework that is robust to Berkson or Classical ME and misspecified priors. 
Empirical findings demonstrate the robustness to ME and improved estimation accuracy of the Robust-MEM (TLS) method compared to prior art. We further instantiate our framework with the MMD-based loss enabling robust inference for nonlinear and non-Gaussian models, offering theoretical guarantees by bounding the generalisation error of the Robust-MEM (MMD) method for both Berkson and Classical error models.

The proposed framework is not limited to specific parametric assumptions or availability of replicate measurements but it is flexible to incorporate such information, if it becomes available. For example, replicate measurements could be incorporated into the prior centering measures through an empirical measure. This work opens up the possibility to explore further specialized loss functions suitable for different data spaces, such as discrete, and allows the possibility to handle heteroscedastic and systematic ME scenarios through alternative nonparametric prior structures.
Future work could also focus on the study of distributional properties of posterior NPL estimators, for example through extending the results of loss-likelihood bootstrap methods \citep{lyddon2019general}. Such results would further impact other related works based on the NPL framework.

\section*{Acknowledgements}
The authors would like to thank Mengqi Chen for spotting an error in a proof in an earlier version of the paper. CD acknowledges support from EPSRC grant [EP/T51794X/1] as part of the Warwick CDT in Mathematics and Statistics and The Alan Turing Institute’s Enrichment Scheme. TD acknowledges support from a UKRI Turing AI acceleration Fellowship [EP/V02678X/1]. For the purpose of open access, the author has applied a Creative Commons Attribution (CC-BY) license to any Author Accepted Manuscript version arising from this submission.

\bibliography{arxiv_submission/bib}

\newpage
\section*{Appendix}

\appendix
\allowdisplaybreaks
\section{Outline}
 In Section \ref{app:related} we give a detailed overview of related literature for inference with ME. In Section \ref{app-tls} we provide further details on the TLS problem and in Section \ref{app-proofs} we provide proofs to all theoretical results presented in the main text. Finally in Section \ref{app-exp-dets} we give further experimental details on the simulations and real-world applications as well as additional experimental results.

\section{Further Related Work} \label{app:related}
An overview of related work was given in Section \ref{sec-related} of the main text. There are a number of alternative approaches that have been developed to deal with ME such as Instrumental Variables methods, mainly in the Econometrics literature, Deconvolution Kernel methods, as well as methodology in the Causal Inference literature. Similarly to the approaches we discussed in the main text, these methods are often limited by strong assumptions or the requirement of additional measurements. 

\paragraph{Instrumental Variable (IV) methods} IV methods, for example, \cite{newhouse1998econometrics, bowden1990instrumental} in the Classical setting, require access to observations from the so-called IV, which is correlated with the covariate, but not necessarily a replicate of it. The required conditional independence between the IV and the response variable, given the unobserved covariate, overcomes the bias induced by the correlation between the regression error and the covariate in the standard setting. The IV setting has also been generalised to a kernel-based nonparametric estimation of causal effect with Classical ME in the cause, in the presence of unobserved continuous confounding by \cite{zhu2022causal}. In the Berkson error case, \cite{schennach2013regressions} deals with nonparametric regression relying on the availability of an IV. Even in the cases where a candidate RV is available, it is hard to ensure that the necessary IV conditions hold.

\paragraph{Deconvolution Kernel Methods}
Deconvolution kernel methods \cite{fan1993nonparametric, wang2011deconvolution} were introduced in Section \ref{sec-related} for kernel-based estimation of regression functions. One of the main challenges in these methods is the choice of the deconvolving kernel as these do not generally admit a closed form and can be sensitive to the selection of bandwidth. \cite{lutkenoner2015family} introduced a kernel family enabling analytical representations of deconvolving kernels for normally distributed ME. \cite{delaigle2016methodology} relaxed assumptions and proposed a deconvolution method for unknown, symmetric error distributions. Similarly, \cite{kato2018uniform} developed uniform confidence bands based on a multiplier-bootstrap approach, assuming unknown error distributions but requiring an independent error distribution sample. Despite desirable theore-tical properties, these methods often apply to specific function classes, such as sufficiently smooth functions. For Berkson-type error models, \cite{delaigle2006nonparametric} suggested a flexible nonparametric approach based on trigonometric functions, assuming some knowledge of the ME distribution or replicate measurements of the error or covariates.

\paragraph{Causal Inference} In causal inference, a biased estimate can lead to false conclusions about the causal relationship between a covariate and the response variable. This problem has been explored in the causal inference literature by reframing it as a causal effect estimation problem with unobserved confounders  \cite[see][]{kuroki2014measurement,adams2019learning,finkelstein2021partial}. Such causal effect estimation problems are met in health and epidemiology sciences where we are interested in exposure-response relations \cite[see][]{brakenhoff2018measurement}. Exposures, often measured with error,
can include both self-reported quantities, such as patient lifestyle changes, and laboratory measurements such as serum cholesterol or the levels of aluminum in the brain \cite[][]{campbell2002potential}. 

\section{Further details on the TLS problem} \label{app-tls}
The main idea of TLS is to expand OLS to consider perturbations in both the covariates and response variables. Consider the model described in (\ref{eq-eps-dist}). The special case of linear regression corresponds to $g(\theta_0, x) = \theta_0^T x$. For observations $w = (w_1, \dots, w_n)^T \in \mathbb{R}^{n \times d_x}$ and $y = (y_1, \dots, y_n)^T \in \mathbb{R}^{n \times 1}$ with associated unobserved errors $\nu = (\nu_1, \dots, \nu_n)^T \in \mathbb{R}^{n \times d_x}$ and $\epsilon = (\epsilon_1, \dots, \epsilon_n)^T \in \mathbb{R}^{n \times 1}$ such that $\nu_{1:n} \simiid \F_{N}^0$ and $\epsilon_{1:n}\simiid \F^0_{E}$, we have that:
\begin{align*} 
    \forall i = 1, \dots, n: \quad y_i = \theta^T x_i  + \epsilon_i, \quad \quad
    x_i = w_i + \nu_i 
\end{align*} 
and we wish to correctly estimate $\theta \in \mathbb{R}^{d_x \times 1}$. It is straightforward to see that if we naively use OLS to fit $\{(w_i, y_i)\}_{i=1}^n$ with the linear model, ignoring ME, the resulting estimator for $\theta$ will be biased because each $w_i$ is correlated with $\nu_i$. This bias is not present in the TLS case. Denote by $\|\cdot \|_F$ the Frobenius norm, then TLS corresponds to the optimisation problem:
\begin{talign*}
    \min_{\epsilon, \nu, \theta} \left\| [\nu \quad \epsilon  ] \right\|_F^2 \quad \text{subject to} \quad y =  (w + \nu) \theta+ \epsilon
\end{talign*}
or equivalently:
\begin{talign}\label{eq-tls}
    \min_{\nu, \theta} \sum_{i=1}^n \| \nu_i \|_2^2 + \| y - (w + \nu)\theta \|_2^2.
\end{talign}
The ordinary TLS algorithm has been extended to, among others, the regularised TLS \cite[][]{guo2002regularized} to induce regularisation on the set of admissible solutions and convex TLS \cite[][]{malioutov2014convex} to deal with varying ME, known error structure and outliers. Some finite sample statistical properties of TLS estimators were given by \cite{xu2012total}. For a more thorough review of TLS methods, we refer the reader to \cite{van2004total} and \cite{markovsky2007overview}. 

\subsection{Connection to Gaussian distributed errors} \label{app-tls-gaus}
We describe how the TLS objective function is equivalent to maximising the joint log-likelihood in the Gaussian distributed errors case. For simplicity we are going to present the univariate Gaussian case, but the multivariate case follows easily. Consider a linear MEM model with Gaussian noise:
    \begin{talign}
        y_i &= \alpha x_i + \epsilon_i, \quad \epsilon_i \sim N(0, \sigma_{\epsilon}^2) \label{eq-y-model}\\
        w_i &= x_i + \nu_i , \quad \nu_i \sim N(0, \sigma_{\nu}^2)\label{eq-x-model}.
    \end{talign}
    Let $Y = (y_1, \dots, y_n) \in \mathbb{R}^n$, $w = ( w_1, \dots, w_n) \in \mathbb{R}^n$, $x = (x_1, \dots, x_n) \in \mathbb{R}^n$, $\epsilon = (\epsilon_1, \dots, \epsilon_n) \in \mathbb{R}^n$ and $\nu = (\nu_1, \dots, \nu_n) \in \mathbb{R}^n$. Then it follows that
    \begin{talign*}
        Y \given \alpha, x, \sigma_\epsilon &\sim N(Y | \alpha x, \sigma^2_{\epsilon} I)  \\
        w \given x, \sigma_{\nu} &\sim N(w | x, \sigma^2_{\nu} I)
    \end{talign*}
    hence under the assumption that $W\independent{Y} \given X$ the joint negative log-likelihood is 
    \begin{talign}
        - l(w,Y \given \alpha, \sigma_\epsilon, \sigma_\nu, x) = &n \log (\sigma_\nu) + n \log(\sigma_\epsilon) + n \log(2 \pi) + \frac{1}{2 \sigma^2_\nu} \sum_{i=1}^n (w_i - x_i)^2 \nonumber \\
        &\quad + \frac{1}{2 \sigma^2_\epsilon} \sum_{i=1}^n (y_i - \alpha x_i)^2.
    \end{talign}
    Substituting (\ref{eq-y-model}) and (\ref{eq-x-model}) above we see that minimising the negative joint log-likelihood above is equivalent to minimising the objective 
    \begin{talign} \label{eq-obj-joint}
        n \log (\sigma_\nu) + n \log(\sigma_\epsilon) + n \log(2 \pi) + \frac{1}{2 \sigma^2_\nu} \sum_{i=1}^n \nu_i^2 + \frac{1}{2 \sigma^2_\epsilon} \sum_{i=1}^n (y_i - \alpha (x_i + \nu_i))^2
    \end{talign}
    So for parameters of interest $\theta = (\alpha, \nu)$ then the MLE estimator would be 
    \begin{talign*}
        \hat{\theta} = \argmin_{\nu, \alpha}  \sum_{i=1}^n \left[ \nu_i^2 +  (y_i - \alpha (x_i + \nu_i))^2 \right]
    \end{talign*}
    which is precisely the TLS objective. If we don't have any extra information on $\nu$, this is generally a non-identifiable problem hence TLS uses numerical optimisation techniques to overcome this. 

\subsection{Solution of the TLS problem} \label{app-tls-sol}
The solution to the problem in (\ref{eq-tls}) has been derived for real RVs using the Singular Value Decomposition (SVD) by \cite{golub1979total}, and \cite{ golub1980analysis} and later generalised by \cite{van1988analysis} who relaxed some of the necessary numerical assumptions for the existence of a solution. \cite{diao2019total} later improved the computational time of solving the TLS problem through fast sketching methods. Here, we provide the SVD solution which we also use in the experiments.

Let $x \in \mathbb{R}^{n \times d_x}$ denote the matrix with $w_i \in \mathbb{R}^{d_x}$ as its i'th row and $Y \in \mathbb{R}^{n \times d_y}$ be the matrix with $y_i \in \mathbb{R}^{d_y}$ in the i'th row. Further denote by $D \in \R^{n \times (d_x + d_y)}$ the data matrix $D := [x \quad Y]$. Let $D = u \Sigma V^T$ be the SVD of $D$ where $\Sigma = \text{diag}(\sigma_1, \dots, \sigma_{d_x + d_y})$ with $\sigma_1 \geq \dots \geq \sigma_{d_x+d_y}$ the singular values of $D$. Moreover denote by $V_{22} \in \mathbb{R}^{d_y \times d_y}$ the bottom right block of $V$. Then the TLS optimisation problem in (\ref{eq-tls}) has a unique solution if and only if $V_{22}$ is non-singular and $\sigma_{d_x} \neq \sigma_{d_x + 1}$. The solution is given by $A = - V_{12} V_{22}^{-1}$ where $V_{12} \in \mathbb{R}^{d_y \times d_x}$ is the top right partition of $V$. A full derivation of the solution can be found at \cite{golub1980analysis}.

\section{Proofs} \label{app-proofs}
We now provide proofs for all technical results of Section \ref{sec-theory}. We begin with all theoretical results for the Berkson case in Section \ref{app:sec-proofs-berk}. The proofs for the results for the Classical error case are in Section \ref{app:sec-proofs-class}.

\subsection{Proofs of theoretical results for the \textcolor{blue}{\texttt{Berkson}} error model} \label{app:sec-proofs-berk}
\subsubsection{Proof of Theorem \ref{thm-main} (\textcolor{blue}{\texttt{Berkson}})}  \label{app-proof-thm}
We start with the following Lemma provided in \cite{alquier2024universal}. 
\begin{lemma}(\cite{alquier2024universal}, Lemma 10)
\label{lemma:preliminary:MMD}
 Let $\mathcal{S}$ be a set (equipped with a $\sigma$-algebra). Let $S_1,\dots,S_n$ be independent RVs on $\mathcal{S}$ with respective distributions $Q_1,\dots,Q_n$. Define $\bar{Q}=(1/n)\sum_{i=1}^n Q_i$ and $\hat{Q}=(1/n)\sum_{i=1}^n \delta_{S_i} $. Then under assumptions \ref{asm:k} and \ref{asm:kxy} we have:
 \begin{talign*}
  \mathbb{E}_{S_{1:n} \sim Q_{1:n}}\left[ \MMD_k(\bar{Q},\hat{Q}) \right] \leq \frac{1}{\sqrt{n}} \text{ and } \mathbb{E}_{S_{1:n} \sim Q_{1:n}}\left[ \MMD_k^2(\bar{Q},\hat{Q}) \right] \leq \frac{1}{n}. 
  \end{talign*}
 where we denote with $\mathbb{E}_{S_{1:n} \sim Q_{1:n}}$ the expectation under $\{S_i \sim Q_i\}_{i=1}^n$.
\end{lemma}
We now prove an important lemma bounding the expected MMD between the target marginal distribution of $X$, namely 
\begin{talign*}
\mathbb{P}_X^{0,n} = \frac{1}{n} \sum_{i=1}^n \mathbb{P}_{X \given w_i}^{0}
\end{talign*}
where $\mathbb{P}_{X \given w_i}^{0}$ is as defined in (\ref{eq:target-berk-x}), and the sample from the exact DPs $\P_X$. 
Recall that we denote by $\F_N^0$ the true ME distribution, i.e. $$\nu_{1:n}^0 \simiid \F_N^0$$.
\begin{lemma} (\textcolor{blue}{\texttt{Berkson}}) \label{lemma-marg-X}
For $i = 1, \dots, n$ let $\mathbb{P}^{i}$ be a sample from the exact Dirichlet process $\mu_i = \DP(c+1, \F'_{w_i})$ and $\P_X = \frac{1}{n} \sum_{i=1}^n \P^i$. 
Then for any kernel $k_X$ on $\mathcal{X}$ satisfying Assumptions \ref{asm:kxy} and \ref{asm-kx} the following statement holds: 
 \begin{talign*}
     &\mathbb{E}_{x_{1:n} \sim \P_X^{0,n}} [ \E_{\P \sim \mu} [\MMD_{k_X}(\P_X^{0,n}, \P_X)]] \leq \\
     &\quad \quad \quad \frac{1}{\sqrt{n(c+2)}} + \frac{c}{c+1} \MMD_{k_X}(\F_N^0, \F_N) + \frac{1}{c+1} \MMD_{k_X}(\F_N^{0}, \delta_0).
 \end{talign*}
\end{lemma}
\begin{proof}
We use the stick-breaking process representation \cite{sethuraman1994constructive} of a DP, meaning that for each $i \in \{1, \dots, n\}$,  $\P^i = \sum_{j=1}^{\infty} \xi_j^i \delta_{z_j^i}$ where $\xi_{1:\infty} \sim \GEM(c+1)$ and $z_{1:\infty} \simiid \F'_{w_i}$. For clarity we use the separate expectations $\E_{\xi^i_{1:\infty} \sim \GEM(c+1)}$ and $\mathbb{E}_{z^i_{1:\infty} \simiid \F'_{w_i}}$ induced by $\mu_i$ for each $i \in \{1, \dots,n\}$. Further let $\F' := \frac{1}{n} \sum_{i=1}^n \F'_{w_i}$. Using the triangle inequality, we have
\begin{talign} \label{eq-target}
   \MMD_{k_X}(\mathbb{P}_{X}^{0,n},\mathbb{P}_X) \leq  \MMD_{k_X}(\P_{X}, \F') + \MMD_{k_X}(\mathbb{P}_{X}^{0,n}, \F') 
\end{talign}
%%%
We first bound the expected MMD between $\P_X$ and $\F'$. We have that:
\begin{talign*}
    &\mathbb{E}_{\P \sim \mu} [\MMD_{k_X}^2(\P_X, \F')] \nonumber\\
    &\quad \quad = \mathbb{E}_{\P \sim \mu} [\| \mu_{\P_X} - \mu_{\F'} \|_{\mathcal{H}_{k_X}}^2] \nonumber\\
    &\quad\quad=  \mathbb{E}_{\P \sim \mu} \left[\left\| \frac{1}{n} \sum_{i=1}^n (\mu_{\P^i} - \mu_{\F'_{w_i}}) \right\|^2_{\mathcal{H}_{k_X}} \right] \nonumber \\
    &\quad \quad= \frac{1}{n^2} \sum_{i=1}^n \mathbb{E}_{\P \sim \mu} \left[ \| \mu_{\P^i} - \mu_{\F^\prime_{w_i}} \|^2_{\HkX}\right] \\
    &\quad \quad \quad \quad + \frac{2}{n^2} \sum_{i \neq j} \E_{\P \sim \mu} \left[\left<\mu_{\P^i} - \mu_{\F^{\prime}_{w_i}}, \mu_{\P}^{j} - \mu_{\F^{\prime}_{w_j}} \right>_{\mathcal{H}_{k_X}} \right] 
\end{talign*}
Notice that for each $i \neq j$, since $\E_{\P^i \sim \mu_i}[\P^i] = \F^\prime_{w_i}$ by the definition of the mean of a DP, we have $\E_{\P \sim \mu}[\mu_{\P^i}] = \mu_{\F^\prime_{w_i}}$ and hence
\begin{talign}
& \E_{\P\sim \mu} \left[\left< \mu_{\P^i} - \mu_{\F^\prime_{w_i}}, \mu_{\P^j} - \mu_{\F^\prime_{w_j}} \right>_{\mathcal{H}_{k_X}} \right]\nonumber \\
&= \E_{\P \sim \mu} [\left<\mu_{\P^i}, \mu_{\P^j} \right>] - 2 \E_{\P\sim \mu} \left[\left<\mu_{\P^i}, \mu_{\F^\prime_{w_j}} \right>_{\mathcal{H}_{k_X}}\right] + \left<\mu_{\F^\prime_{w_i}}, \mu_{\F^\prime_{w_j}} \right>_{\mathcal{H}_{k_X}} \nonumber\\
&= 0. \label{eq:zeroeq}
\end{talign}
Therefore, we have:
\begin{talign}
    &\E_{\P \sim \mu} [\MMD_{k_X}^2(\P_X, \F')] \nonumber\\
    &\quad \quad= \frac{1}{n^2} \sum_{i=1}^n \mathbb{E}_{\P \sim \mu} \left[ \| \mu_{\P^i} - \mu_{\F^\prime_{w_i}} \|^2_{\HkX}\right] \nonumber\\
    &\quad \quad= \frac{1}{n^2} \sum_{i=1}^n \E_{\xi^{i}_{1:\infty} \sim \GEM(c+1)} \left[ \mathbb{E}_{z^{i}_{1:\infty} \simiid \F'_{w_i}} \left[\left\| \sum_{j=1}^{\infty} \xi_j^i k_X(z_j^i, \cdot) - \mu_{\F'_{w_i}}  \right\|^2_{\mathcal{H}_{k_X}} \right] \right] \nonumber \\
    &\quad\quad= \frac{1}{n^2} \sum_{i=1}^n \E_{\xi^{i}_{1:\infty} \sim \GEM(c+1)} \left[ \mathbb{E}_{z^{i}_{1:\infty} \simiid \F'_{w_i}} \left[ \sum_{j=1}^{\infty} (\xi_j^i)^2 \left\| k_X(z_j^i, \cdot) - \mu_{\F'_{w_i}} \right\|^2_{\mathcal{H}_{k_X}} \right.\right. \nonumber \\
    &\quad\quad\quad \quad \left. \left. + 2 \sum_{j \neq t} \xi_j^i \xi_t^i \left<k_X(z_j, \cdot) - \mu_{\F'_{w_i}}, k_X(z_t, \cdot) - \mu_{\F'_{w_i}} \right>_{\HkX} \right] \right]. \label{eq:pxf}
\end{talign}
For any $j \neq t$ and each $i \in \{1. \dots, n\}$ we have:
\begin{talign}
    & \E_{z_j, z_t \simiid \F'_{w_i}} \left[ \left< k_X(z_j, \cdot) - \mu_{\F'_{w_i}}, k_X(z_t, \cdot) - \mu_{\F'_{w_i}} \right>_{\HkX} \right] \nonumber \\
    &= \E_{z_j, z_t \simiid \F'_{w_i}} \left[\left< k_X(z_j, \cdot) - \E_{z_j \sim \F'_{w_i}}[k_X(z_j, \cdot)], \right. \right. \nonumber\\
    &\left. \left.\quad \quad \quad k_X(z_t, \cdot) - \E_{z_t \sim \F'}[(k_X(z_t, \cdot)] \right>_{\HkX} \right] \nonumber \\
    &= \E_{z_j, z_t \simiid \F'_{w_i}}[k_X(z_j, z_t)] - \E_{z_j \sim \F'_{w_i}}[k_X(z_j, \cdot)] \E_{z_t \sim \F'_{w_i}}[k_X(z_t, \cdot)] \notag \nonumber\\
     &\qquad - \E_{z_t \sim \F_{w_i}}[k_X(z_t, \cdot)] \E_{z_j \sim \F_{w_i}}[k_X(z_j, \cdot)] \nonumber\\
     &\quad \quad + \E_{z_t \sim \F'_{w_i}}[k_X(z_t, \cdot)] \E_{z_j \sim \F'_{w_i}}[k_X(z_j, \cdot)] \nonumber\\
    &= \E_{z_j, z_t \simiid \F'_{w_i}}[k_X(z_j, z_t)] - \E_{z_j, z_t \simiid \F'_{w_i}}[k_X(z_j, z_t)] \nonumber \\
    &\quad\quad- \E_{z_j, z_t \simiid \F'_{w_i}}[k_X(z_j, z_t)] + \E_{z_j, z_t \simiid \F'_{w_i}}[k_X(z_j, z_t)] \nonumber \\
    &= 0. \label{eq:zeroeq2}
\end{talign}
Moreover, for any $j = 1,2, \dots$ and each $i \in \{1,\dots,n\}$ we have 
\begin{talign}
    &\E_{z_k \sim \F'_{w_i}} \left[ \| k_X(z_k, \cdot) - \mu_{\F'_{w_i}} \|_{\HkX}^2    \right] \nonumber \\
    &\quad= \E_{z_k \sim \F'_{w_i}} \left[ \| k_X(z_k, \cdot) - \E_{z_k \sim \F'_{w_i}}[ k_X(z_k, \cdot)] \|_{\HkX}^2 \right] \nonumber \\
    &\quad= \E_{z_k \sim \F'_{w_i}} \left[ \| k_X(z_k, \cdot) \|_{\HkX}^2 - 2 \left< k_X(z_k, \cdot), \E_{z_k \sim \F_{w_i}}[k_X(z_k, \cdot)] \right> \right. \nonumber \\
    &\quad\quad \quad \left. + \|\E_{z_k \sim \F'_{w_i}} [k_X(z_k, \cdot)] \|_{\HkX}^2 \right] \nonumber\\
    &\quad= \E_{z_k \sim \F'_{w_i}} \left[ \| k_X(z_k, \cdot) \|_{\HkX}^2 \right] - 2 \left< \E_{z_k \sim \F'_{w_i}} [k_X(z_k, \cdot)], \right.  \nonumber \\
    &\quad\quad \quad \left. \E_{z_k \sim \F'_{w_i}}[k_X(z_k, \cdot)] \right> + \|\E_{z_k \sim \F'_{w_i}} [k_X(z_k, \cdot)] \|_{\HkX}^2 \nonumber\\
    &\quad= \E_{z_k \sim \F'_{w_i}} \left[ \| k(z_k, \cdot) \|_{\HkX}^2 \right] - 2   \|\E_{z_k \sim \F'_{w_i}} [k_X(z_k, \cdot)] \|_{\HkX}^2 \nonumber \\
    &\quad\quad \quad + \|\E_{z_k \sim \F'} [k_X(z_k, \cdot)] \|_{\HkX}^2 \nonumber\\
    &\quad= \E_{z_k \sim \F'_{w_i}} [\| k_X(z_k, \cdot) \|_{\HkX}^2] - \|\E_{z_k \sim \F'_{w_i}} [k_X(z_k, \cdot)] \|_{\HkX}^2 \nonumber\\
    &\quad\leq \E_{z_k \sim \F'_{w_i}} [\| k_X(z_k, \cdot) \|_{\HkX}^2] \nonumber\\
    &\quad= \E_{z_k \sim \F'_{w_i}} [ | k_X(z_k, z_k) |] \label{eq:ineq1}.
\end{talign}
Substituting (\ref{eq:zeroeq2}) and (\ref{eq:ineq1}) in (\ref{eq:pxf}):
\begin{talign*} 
    \mathbb{E}_{\P \sim \mu} [\MMD_{k_X}^2(\P_X, \F')] &= \frac{1}{n^2} \sum_{i=1}^n \sum_{j=1}^{\infty} \E_{\xi_{1:\infty}^i \sim \GEM(c+1)}[(\xi_j^i)^2] \\
    &\quad \quad\E_{z_k \sim \F'_{w_i}} [ | k_X(z_k, z_k) |] \\
    &\leq \frac{1}{n^2} \sum_{i=1}^n \sum_{j=1}^{\infty} \E_{\xi_{1:\infty}^i \sim \GEM(c+1)}[(\xi_j^i)^2] \\
    &= \frac{1}{n}\sum_{j=1}^{\infty} \frac{2(c+1)^{j-1}}{(c+3)^j(c+2)} \\
    &= \frac{1}{n}\frac{2}{(c+3)(c+2)} \sum_{j=1}^{\infty} \left(\frac{c+1}{c+3} \right)^{j-1} \\
    &= \frac{1}{n}\frac{2}{(c+3)(c+2)} \frac{1}{1- \frac{c+1}{c+3}} \\
    &= \frac{1}{n(c+2)}
\end{talign*}
where in the third equality we used the second moment of the GEM distribution namely $\E_{\xi_{1:\infty} \sim \GEM(c+1)}[\xi_j^2] = \frac{2(c+1)^{j-1}}{(c+3)^j(c+2)}$. By Jensen's inequality we have that 
\begin{align} \label{eq-px-fprime}
    \mathbb{E}_{\P \sim \mu} [\MMD_{k_X}(\P_X, \F')] \leq \sqrt{ \mathbb{E}_{\P \sim \mu} [\MMD_{k_X}^2(\P_X, \F')]} \leq \frac{1}{\sqrt{n(c+2)}}.
\end{align}

% %
We now bound the second term in (\ref{eq-target}). By definition of $\F'$ it follows that
$$
\F' =  \frac{1}{n} \sum_{i=1}^n \F'_{w_i} = \frac{1}{n} \sum_{i=1}^n \left[\frac{c}{c+1} \F_{w_i} + \frac{1}{c+1} \delta_{w_i} \right] = \frac{c}{c+1} \F + \frac{1}{c+1} \P_W^n 
$$
where $\F := \frac{1}{n} \sum_{i=1}^n \F_{w_i}$ and $\P_W := \frac{1}{n} \sum_{i=1} \delta_{w_i}$ is the empirical measure of the observed, noisy covariates. By linearity of expectations it follows that $\E_{\F'}[\cdot] = \frac{c}{c+1} \E_{\F}[\cdot] + \frac{1}{c+1} \E_{\P_W^n} [\cdot]$ and the kernel mean embedding of $\F'$ is 
\begin{talign*}
    \mu_{\F'} &:= \mathbb{E}_{z \sim \F'}[k(z, \cdot)] = \frac{c}{c+1} \E_{z \sim \F}[k(z,\cdot)] + \frac{1}{c+1} \mathbb{E}_{z \sim \P_W^n}[k(z,\cdot)] \\
    &= \frac{c}{c+1} \mu_{\F} + \frac{1}{c+1} \mu_{\P_W^n}.
\end{talign*}
Using this we can first bound the objective as:
\begin{talign}
    \MMD_{k_X}(\P_X^{0,n}, \F') &= \left\| \mu_{\P_X^{0,n}} - \frac{c}{c+1} \mu_{\F} - \frac{1}{c+1} \mu_{\P_W^n} \right\|_{\HkX} \nonumber \\
    &\leq \frac{c}{c+1} \| \mu_{\P_X^{0,n}} - \mu_{\F} \|_{\HkX} + \frac{1}{c+1} \| \mu_{\P_X^{0,n}} - \mu_{\P_W^n} \|_{\HkX} \nonumber \\
    &= \frac{c}{c+1} \MMD_{k_X}(\P_X^{0,n}, \F) + \frac{1}{c+1} \MMD_{k_X}(\P_X^{0,n}, \P_W^{n}) \label{eq-three-terms}
\end{talign}
where we used the triangle inequality for both inequalities. 
%%%%%
Let $\phi_X$ be the feature map associated with the kernel function $k_X$. Then,
\begin{align*}
    \MMD_{k_X}( \P_X^{0,n}, \F) &= \| \mathbb{E}_{X \sim \P_X^{0,n}} [\phi_X(X)] - \mathbb{E}_{X \sim \F} [\phi_X(X)] \|_{\mathcal{H}_{k_X}} \\
    &= \| \mathbb{E}_{X \sim \frac{1}{n} \sum_{i=1}^n \P^0_{X \given w_{i}}} [\phi_X(X)] - \mathbb{E}_{X \sim \frac{1}{n} \sum_{i=1}^n \F_{w_i}} [\phi_X(X)] \|_{\mathcal{H}_{k_X}} \\
    &= \left\| \frac{1}{n} \sum_{i=1}^n \mathbb{E}_{X \sim \P^0_{X \given w_{i}}} [\phi_X(X)] - \frac{1}{n} \sum_{i=1}^n \mathbb{E}_{X \sim \F_{w_i}} [\phi_X(X)] \right\|_{\mathcal{H}_{k_X}} \\
     &= \left\| \frac{1}{n} \sum_{i=1}^n \mathbb{E}_{\nu \sim \F_N^0} [\phi_X(w_i + \nu)] - \frac{1}{n} \sum_{i=1}^n \mathbb{E}_{\nu \sim \F_N} [\phi_X(w_i + \nu)] \right\|_{\mathcal{H}_{k_X}} \\
     &\leq \frac{1}{n} \sum_{i=1}^n \| \mathbb{E}_{\nu \sim \F_N^0} [\phi_X(w_i + \nu)] - \mathbb{E}_{\nu \sim \F_N} [\phi_X(w_i + \nu)] \|_{\mathcal{H}_{k_X}} \\
     &= \frac{1}{n} \sum_{i=1}^n \MMD_{k_X^i} (\F_N^0, \F_N) 
\end{align*}
where for each $i \in \{1, \dots, n\}$, $k_{X}^i(\nu, \nu') = k_X(\nu + w_i, \nu' + w_i)$ and $w_1, \dots, w_n$ denote the observations. By assumption \ref{asm-kx}, $k_X$ is translation invariant so for any $i \in \{1, \dots, n\}$,
$$
k_X^i(\nu, \nu') = k_X(\nu+w_i, \nu'+w_i) = \psi(\nu + w_i - \nu' - w_i) = \psi(\nu - \nu') = k_X(\nu,\nu')
$$
hence
\begin{talign*}
    &\MMD^2_{k_X^i}(\F_N^0, \F_N) \\
    & =  \mathbb{E}_{\nu \sim \F_N^0, \nu' \sim \F_N^0} [k_X^i(\nu, \nu')] -2 \mathbb{E}_{\nu \sim \F_N^0, \nu' \sim \F_N} [k_X^i(\nu, \nu')] + \mathbb{E}_{\nu \sim \F_N, \nu' \sim \F_N} [k_X^i(\nu, \nu')] \\
    & = \mathbb{E}_{\nu \sim \F_N^0, \nu' \sim \F_N^0} [k_X(\nu, \nu')] -2 \mathbb{E}_{\nu \sim \F_N^0, \nu' \sim \F_N} [k_X(\nu, \nu')] + \mathbb{E}_{\nu \sim \F_N, \nu' \sim \F_N} [k_X(\nu, \nu')] \\
    & = \MMD^2_{k_X}(\F_N^0, \F_N)
\end{talign*}
and 
\begin{talign*}
    \MMD_{k_X}(\P_X^0, \F) \leq \frac{1}{n} \sum_{i=1}^n \MMD_{k_X^i}(\F_N^0, \F_N) = \MMD_{k_X}(\F_N^0, \F_N).
\end{talign*}
Similarly,
\begin{talign*}
     \MMD_{k_X}(\P_X^{0,n}, \P_W^n) &= \left\| \frac{1}{n} \sum_{i=1}^n \left[ \mathbb{E}_{x \sim \P^0_{X \given w_i}} [\phi_X(x)] - \phi_X(w_i)  \right] \right\|_{\mathcal{H}_{k_X}} \\
    &= \left\| \frac{1}{n} \sum_{i=1}^n \left[ \mathbb{E}_{\nu \sim \F_N^0} [\phi_X(w_i + \nu)] - \phi_X(w_i)  \right] \right\|_{\mathcal{H}_{k_X}} \\
    &= \left\| \frac{1}{n} \sum_{i=1}^n \left[ \mathbb{E}_{\nu \sim \F_N^0} [k^i_X(\nu, \cdot)] - k^i_X(0,\cdot)  \right] \right\|_{\mathcal{H}_{k_X}} \\
    &\leq \frac{1}{n} \sum_{i=1}^n \left\| \mathbb{E}_{\nu \sim \F_N^0} [k_X^i(\nu, \cdot)] - k_X^i(0, \cdot) \right\|_{\mathcal{H}_{k_X}} \\
    &= \frac{1}{n} \sum_{i=1}^n \MMD_{k^i_X}(\F_N^{0}, \delta_{0}) \\
    &= \MMD_{k_X}(\F_N^{0}, \delta_{0}). 
\end{talign*}
Substituting in (\ref{eq-three-terms}), we obtain:
\begin{talign} \label{eq-pxon-fprime}
    \mathbb{E}_{x_{1:n} \sim \P_X^{0,n}}[\MMD_{k_X}(\P_X^{0,n}, \F')] \leq  \frac{c}{c+1} \MMD_{k_X}(\F_N^0, \F_N) + \frac{1}{c+1} \MMD_{k_X}(\F_N^{0}, \delta_{0}). 
 \end{talign}
 Finally, using (\ref{eq-px-fprime}) and (\ref{eq-pxon-fprime}) in (\ref{eq-target}) we obtain the required result:
 \begin{talign*}
     &\mathbb{E}_{x_{1:n} \sim \P_X^{0,n}} [ \E_{\P \sim \mu} [\MMD_{k_X}(\P_X^{0,n}, \P_X)]] \leq \\
     &\quad \quad \quad \frac{1}{\sqrt{n(c+2)}} + \frac{c}{c+1} \MMD_{k_X}(\F_N^0, \F_N) + \frac{1}{c+1} \MMD_{k_X}(\F_N^{0}, \delta_{0})
 \end{talign*}
\end{proof}
We now restate a Lemma provided by \cite{alquier2024universal}:
\begin{lemma}[\cite{alquier2024universal}, Lemma 3] \label{lem-alq2}
   For $\P_X', \P_X'' \in \mathcal{P}_\X$ let $\P^\prime = \P'_X \times \P_{Y \given \cdot}$ and $\P'' = \P''_X \times \P_{Y \given \cdot}$. Assume that $(\P_{Y \given x})_{x \in \X}$ admits a bounded linear conditional mean embedding operator $\mathcal{C}_{Y \given X}$. Then 
   \begin{talign*}
     \MMD_{k}(\P', \P'') \leq \| \mathcal{C}_{Y \given X} \|_{o} \MMD_{k^2_X}(\P'_X, \P''_X).  
   \end{talign*}
    where $k^2_X = k_X \otimes k_X$.
\end{lemma}
%
%%%%%%%%%%%%%%%%%%%%%%%%%%%%%%%%%%%%%%%%%%%
\begin{lemma} (\textcolor{blue}{\texttt{Berkson}}) \label{lem-joint-emp}
Let $\mathbb{P}^{0,n} := \P_X^n \times \P_Y^n = \frac{1}{n} \sum_{i=1}^n \delta_{(x_i, y_i)}$ and $\mathbb{P}_{XY} = \P_X \times \P_Y^{n} = \frac{1}{n} \sum_{i=1}^n \sum_{j=1}^{\infty} \xi_j^i \delta_{(z_j,y_i)}$ where $z^i_{1:\infty} \overset{iid}{\sim} \F'_{w_i}, \xi^i_{1:\infty} \sim \GEM(c+1)$. Then under Assumptions \ref{asm:k}-\ref{asm-kx}:
\begin{talign*}
    &\mathbb{E}_{(x,y)_{1:n} \sim \P^0} [ \E_{\P \sim \mu} [\MMD_k(\mathbb{P}^{0,n}, \mathbb{P}_{XY})]] \\
    &\leq \frac{1}{\sqrt{n(c+2)}} + \frac{c}{c+1} \sqrt{{\text{Var}}_{\mathcal{H}_{k_X}}[\F_N^0] + \MMD^2_{k_X}(\F_N^0, \F_N)} + \frac{\sqrt{2}
    }{c+1} \sqrt{\MMD_{k_X}(\F_N^0, \delta_0)}. 
\end{talign*}
\end{lemma}
\begin{proof}
By the triangle inequality, we have:
\begin{talign} \label{eq:idk}
    \MMD_k(\P^{0,n}, \P_{XY}) &\leq \MMD_k\left(\P^{0,n}, \frac{1}{n} \sum_{i=1}^n \mathbb{F}^\prime_{w_i} \times \delta_{y_i} \right)  + \MMD_k\left(\P_{X,Y}, \frac{1}{n} \sum_{i=1}^n \mathbb{F}^\prime_{w_i} \times \delta_{y_i}\right).
\end{talign}
Let $\phi_X, \phi_Y$ be the feature maps associated with $k_X$ and $k_Y$ respectively. For the first term, we have:
\begin{talign*}
    &\MMD_k\left(\P^{0,n}, \frac{1}{n} \sum_{i=1}^n \mathbb{F}^\prime_{w_i} \times \delta_{y_i} \right)\\
    &= 
    \left\| \frac{1}{n} \sum_{i} \left[\phi_X(x_i) \otimes \phi_Y(y_i) - \frac{c}{c+1} \mu_{\mathbb{F}_{w_i}} \otimes \phi_Y(y_i) - \frac{1}{c+1} \phi_X(w_i) \otimes \phi_Y(y_i)\right] \right\|_{\mathcal{H}_k} \\
    &= \left \| \frac{1}{n} \sum_i \left[ 
    \frac{c}{c+1} (\phi_X(x_i) - \mu_{\mathbb{F}_{w_i}}) \otimes \phi_Y(y_i) + \frac{1}{c+1} (\phi_X(x_i) - \phi_X(w_i)) \otimes \phi_Y(y_i) 
    \right] \right\|_{\mathcal{H}_k} \\
    &\leq \frac{1}{n} \sum_i \left[ \frac{c}{c+1}\left\|(\phi_X(x_i) - \mu_{\mathbb{F}_{w_i}}) \otimes \phi_Y(y_i)  \right\|_{\mathcal{H}_k} \right.\\
    &\left. \quad + \frac{1}{c+1} \left\| (\phi_X(x_i) - \phi_X(w_i)) \otimes \phi_Y(y_i)\right\|_{\mathcal{H}_k} \right] \\
    &= \frac{1}{n} \sum_i \left[ \frac{c}{c+1}\left\|(\phi_X(x_i) - \mu_{\mathbb{F}_{w_i}}) \right\|_{\mathcal{H}_{k_X}} \left\| \phi_Y(y_i)  \right\|_{\mathcal{H}_{k_Y}} \right] \\
    &\quad + \frac{1}{n} \sum_{i=1}^n 
    \left[\frac{1}{c+1} \left\| (\phi_X(x_i) - \phi_X(w_i)) \right\|_{\mathcal{H}_{k_X}} \left\| \phi_Y(y_i)\right\|_{\mathcal{H}_{k_Y}} \right].
\end{talign*}
    Now by taking expectations and using the Cauchy-Schwarz inequality, we get: 
    \begin{talign}
        &\mathbb{E}_{w_i,x_i,y_i}\left[ \MMD_k\left(\P^{0,n}, \frac{1}{n} \sum_{i=1}^n \mathbb{F}^\prime_{w_i} \times \delta_{y_i} \right)  \right] \nonumber\\
        &\leq \frac{c}{n(c+1)} \sum_{i=1}^n \mathbb{E}_{w_i,x_i,y_i} [\| \phi_X(x_i) - \mu_{\F_{w_i}} \|_{\mathcal{H}_{k_X}} \| \phi_Y(y_i) \|_{\mathcal{H}_{k_Y}}] \nonumber\\
        &\quad + \frac{1}{n(c+1)} \sum_{i=1}^n \mathbb{E}_{w_i,x_i,y_i}  [\left\| (\phi_X(x_i) - \phi_X(w_i)) \right\|_{\mathcal{H}_{k_X}} \left\| \phi_Y(y_i)\right\|_{\mathcal{H}_{k_Y}} ] \nonumber\\
        &\leq \frac{c}{n(c+1)} \sum_{i} \sqrt{\mathbb{E}_{w_i,x_i}[\|\phi_X(x_i) - \mu_{F_{w_i}} \|_{\mathcal{H}_{k_X}}^2]} \sqrt{\mathbb{E}_{w_i,x_i,y_i}[\|\phi_Y(y_i) \|_{\mathcal{H}_{k_Y}}^2]} \nonumber\\
        &\quad + \frac{1}{n(c+1)} \sum_{i} \sqrt{\mathbb{E}_{w_i,x_i}[\|\phi_X(x_i) - \phi_X(w_i) \|_{\mathcal{H}_{k_X}}^2]} \sqrt{\mathbb{E}_{w_i,x_i,y_i}[\|\phi_Y(y_i) \|_{\mathcal{H}_{k_Y}}^2]} \nonumber\\
        &\leq \frac{c}{n(c+1)} \sum_i \sqrt{\mathbb{E}_{w_i,x_i}[\|\phi_X(x_i) - \mu_{F_{w_i}} \|_{\mathcal{H}_{k_X}}^2]}\nonumber\\
        &\quad + \frac{1}{n(c+1)}  \sum_{i} \sqrt{\mathbb{E}_{w_i,x_i}[\|\phi_X(x_i) - \phi_X(w_i) \|_{\mathcal{H}_{k_X}}^2]}. \label{eq:idk2}
    \end{talign}
For the second term, by Assumption \ref{asm-kx} we have that $k_X$ is translation invariant so $\exists \psi_X$ such that $k_X(x,x^\prime) = \psi_X(x-x^\prime)$ hence:
\begin{talign*}
    &\mathbb{E}_{w_i,x_i}[\|\phi_X(x_i) - \phi_X(w_i) \|_{\mathcal{H}_{k_X}}^2]\\
    &= \E_{\nu \sim \F^0_N}[\| \phi_X(w_i + \nu) - \phi_X(w_i)\|^2_{\mathcal{H}_{k_X}}] \\
    &= \E_{\nu \sim \F^0_N}[k_X(w_i + \nu, w_i + \nu) + k_X(w_i, w_i) - 2 k_X(w_i+\nu, w_i)] \\
    &= 2 \psi_X(0) - 2 \E_{\nu \sim \F_N^0}[\psi_X(\nu)]\\
    &= 2 - 2\E_{\nu \sim \F_N^0}[\psi_X(\nu)]
\end{talign*}
where in the last equality we used the fact that $\psi_X(0) = \sup_{x,x^\prime \in \X}k_X(x,x^\prime) = 1$ since $k_X$ is positive definite and hence $\psi_X$ needs to have its supremum at $0$. We hence also have that:
\begin{talign*}
    \MMD_{k_X}^2(\F_N^0, \delta_0) &= \| \E_{\nu \sim \F_N^0}[\phi_X(\nu)] - \phi_X(0) \|^2_{\mathcal{H}_{k_X}} \\
    &= \| \E_{\nu \sim \F_N^0}[\phi_X(\nu)] \|^2_{\mathcal{H}_{k_X}} - 2 \E_{\nu \sim \F_N^0}[\psi_X(\nu)] + \psi_X(0)\\
    &= (1 - \E_{\F_N^0}[\phi_X(N)])^2 \\
    &= \frac{1}{4} (2 - 2\E_{\nu \sim \F_N^0}[\psi_X(\nu)])^2 \\
    &= \frac{1}{4} \left(\mathbb{E}_{w_i,x_i}[\|\phi_X(x_i) - \phi_X(w_i) \|_{\mathcal{H}_{k_X}}^2] \right)^2
\end{talign*}
and hence 
\begin{align}
    \mathbb{E}_{w_i,x_i}[\|\phi_X(x_i) - \phi_X(w_i) \|_{\mathcal{H}_{k_X}}^2] = 2 \MMD_{k_X}(\F_N^0, \delta_0).
\end{align}
For the first term, note that 
\begin{talign*}
    &\E_{w_i, x_i}[\| \phi_X(x_i) - \mathbb{E}_{x \sim \F_{w_i}}[\phi_X(x)\|^2] \\
    &= \E_{x_i \sim \P^0_{X | w_i}}[k(x_i, x_i)] - 2 \E_{x_i \sim \P^0_{X|w_i}}[\left<\phi_X(x_i), \E_{x \sim \F_{w_i}} \phi_X(x) \right>] \\
    &\quad \quad + \left< \E_{\F_{w_i}}[\phi_X(x)], \E_{\F_{w_i}}[\phi_X(x)] \right>_{\mathcal{H}_{k_X}}
    \\&=\E_{\nu \sim \F_N^0}[k_i(\nu, \nu)] - 2 \E_{\nu \sim \F_N^0}[\left<\phi_i(\nu), \E_{\nu^\prime \sim \F_N} \phi_i(\nu^\prime) \right>_{\mathcal{H}_{k_X}}] \\
    &\quad \quad + \left< \E_{\F_N}[\phi_i(\nu)], \E_{\F_N}[\phi_i(\nu)] \right> \\
    &= \E_{\nu \sim \F_N^0}[k_i(\nu, \nu)] - \left<\mu_{\F_N^0}, \mu_{\F_N^0} \right> + \MMD^2_{k_i}(\F_N^0, \F_N) \\
    &= {\text{Var}}_{\mathcal{H}_{k_i}}[\F_N^0] + \MMD^2_{k_i}(\F_N^0, \F_N)\\
    &= {\text{Var}}_{\mathcal{H}_{k_X}}[\F_N^0] + \MMD^2_{k_X}(\F_N^0, \F_N)
\end{talign*}
where $\phi_i$ denotes the feature map of $k_i$ which is defined as $k_i(\nu,\nu^\prime) := k_X(w_i + \nu, w_i + \nu^\prime)$ and the last equality follows from the fact that $k_X$ is translation invariant. We further used the notation: 
${\text{Var}}_{\mathcal{H}_{k_X}}[\F_N^0] := \E_{\nu \sim \F_N^0}[k_X(\nu, \nu)] - \left<\mu_{\F_N^0}, \mu_{\F_N^0} \right>_{\mathcal{H}_{k_X}}$ which is the effective RKHS variance of $\F_N^0$. When there is no ME, this term will go to zero, whereas otherwise it will quantify the error induced by higher ME distribution variances, making the problem harder. Substituting into (\ref{eq:idk2}) we obtain:
\begin{talign*}
    &\mathbb{E}_{(x,y)_{1:n} \sim \P^0} \left[ \E_{\P \sim \mu} \left[ \MMD_k\left(\P^{0,n}, \frac{1}{n} \sum_{i=1}^n \mathbb{F}^\prime_{w_i} \times \delta_{y_i} \right)  \right] \right] \\
    &\leq \frac{c}{c+1} \sqrt{{\text{Var}}_{\mathcal{H}_{k_X}}[\F_N^0] + \MMD^2_{k_X}(\F_N^0, \F_N)} + \frac{\sqrt{2}
    }{c+1} \sqrt{\MMD_{k_X}(\F_N^0, \delta_0)}.
\end{talign*}
We now turn to the second term of (\ref{eq:idk}). Expanding the expected MMD we obtain:
\begin{talign*}
    &\mathbb{E}_{(x,y)_{1:n} \sim \P^0} \left[ \E_{\P \sim \mu} [\left[\MMD_k^2(\P_{XY}, \frac{1}{n} \sum_{i=1}^n \mathbb{F}^\prime_{w_i} \times \delta_{y_i}) \right] \right]\\
    &= \mathbb{E}_{(x,y)_{1:n} \sim \P^0} \left[ \E_{\P \sim \mu} \left[ \left\| \frac{1}{n} \sum_{i=1}^n \left(\mu_{\F^\prime_{w_i}} \otimes \phi_Y(y_i) - \mu_{\P^i} \otimes \phi_Y(y_i) \right) \right\|^2 \right] \right] \\
    &= \frac{1}{n^2} \sum_{i=1}^n \mathbb{E}_{w_i, x_i,y_i} \mathbb{E}_{\P^i} \left[ \| \mu_{\F^\prime_{w_i}} \otimes \phi_Y(y_i) - \mu_{\P^i} \otimes \phi_Y(y_i) \|^2\right] \\
    &\quad + \frac{2}{n^2} \sum_{i \neq j} \mathbb{E}_{w_i, x_i,y_i, w_j, x_j, y_j} \mathbb{E}_{\P^i} \left[\left< \mu_{F^\prime_{w_i}} \otimes \phi_Y(y_i) - \mu_{\P^i} \otimes \phi_Y(y_i) , \right. \right.\\
    &\left. \left. \quad \quad \mu_{\F^\prime_{w_j}} \otimes \phi_Y(y_j) - \mu_{\P^j} \otimes \phi_Y(y_j)\right> \right]
\end{talign*}
The cross terms vanish since $\mathbb{E}_{\P \sim \mu}[\P^i] = \F^\prime_{w_i}$ for each $i \in \{1, \dots, n\}$ and conditionally on $(w_i, w_j, x_i, x_j, y_i, y_j)$ $\mu_{\P^i} \otimes \phi_Y(y_i)$ and $\mu_{\P^j} \otimes \phi_Y(y_j)$ are independent and hence: 
\begin{talign*}
    &\mathbb{E}_{w_i, x_i,y_i, w_j, x_j, y_j} \mathbb{E}_{\P} \left[\left< \mu_{\F^\prime_{w_i}} \otimes \phi_Y(y_i) - \mu_{\P_i} \otimes \phi_Y(y_i), \right. \right.  \\
    &\left. \left.  \quad \mu_{\F^\prime_{w_j}} \otimes \phi_Y(y_j) - \mu_{\P_j} \otimes \phi_Y(y_j)\right> \right] \\
    &\quad = \mathbb{E}_{w_i, x_i,y_i, w_j, x_j, y_j} \left[\left<\mu_{\F^\prime_{w_i}}, \mu_{\F^\prime_{w_j}}\right>- 2 \mathbb{E}_{\P_i} \mathbb{E}_{\P_j} \left[\left<\mu_{\F_{w_i}}^\prime \otimes 
    \phi_Y(y_i), \mu_{\P_j} \otimes \phi_Y(y_j)\right>\right] \right.\\
    &\left. \quad + \mathbb{E}_{\P_i} \mathbb{E}_{\P_j} [\left<\mu_{\P_i} \otimes \phi_Y(y_i), \mu_{\P_j} \otimes \phi_Y(y_j) \right>]\right] \\
    &= \mathbb{E}_{w_i,x_i,y_i, w_j, x_j, y_j} \left[ 
        \left<\mu_{\F^\prime_{w_i}}, \mu_{\F^\prime_{w_j}}\right>- 2 \left<\mu_{\F^\prime_{w_i}} \otimes \phi_Y(y_i), \mu_{\F^\prime_{w_j}} \otimes \phi_Y(y_j)\right> \right.\\
        &\left. \quad + \left<\mu_{\F^\prime_{w_i}} \otimes \phi_Y(y_i), \mu_{\F^\prime_{w_j}} \otimes \phi_Y(y_j)\right> 
    \right] \\
    &= 0. 
\end{talign*}
For the diagonal terms, we have:
\begin{talign*}
    &\mathbb{E}_{w_i,x_i,y_i} \mathbb{E}_{\P_i} \left[ \| \mu_{\F^\prime_{w_i}} \otimes \phi_Y(y_i) - \mu_{\P_i} \otimes \phi_Y(y_i) \|^2\right] \\
    &=  \mathbb{E}_{w_i,x_i,y_i} \mathbb{E}_{\P_i} \left[\left \| \sum_{j=1}^\infty  
        \xi_j \left[ \phi_X(z_j) \otimes \phi_Y(y_i) - \mu_{\F^\prime_{w_i}} \otimes \phi_Y(y_i)
    \right] \right\|^2 \right] \\
    &= \mathbb{E}_{w_i,x_i,y_i} \mathbb{E}_{\P_i} \left[
        \sum_{j=1}^\infty \xi_j^2 \| (\phi_X(z_j) - \mu_{\F^\prime_{w_i}}) \otimes \phi_Y(y_i) \|^2 \right.\\
        &\left. \quad + 2 \sum_{j \neq t} \xi_j \xi_t \left<(\phi_X(z_j) - \mu_{\F^\prime_{w_i}}) \otimes \phi_Y(y_i), (\phi_X(z_t) - \mu_{\F^\prime_{w_i}}) \otimes \phi_Y(y_i) \right> 
    \right]
\end{talign*}
Let $\Delta_j := (\phi_X(z_j) - \mu_{\F^\prime_{w_i}}) \otimes \phi_Y(y_i)$. Then 
\begin{talign*}
    \mathbb{E}_{w_i,x_i,y_i} \mathbb{E}_{\P_i} [\Delta_j] &= \mathbb{E}_{w_i,x_i,y_i} \left[\mathbb{E}_{z_j \sim \F^\prime_{w_i}} [(\phi_X(z_j) - \mu_{\F^\prime_{w_i}}) \otimes \phi_Y(y_i)] \right]\\
    &= \mathbb{E}_{w_i,x_i,y_i} \left[ (\mu_{\F^\prime_{w_i}} - \mu_{\F^\prime_{w_i}}) \otimes \phi_Y(y_i) \right] \\
    &= 0
\end{talign*}
and $\Delta_j$ is independent of $\Delta_t$ conditionally on $w_i,x_i,y_i$ for $j \neq t$ so the cross terms vanish. Finally, we have:
\begin{talign*}
    &\mathbb{E}_{w_i,x_i,y_i} \mathbb{E}_{\P_i} \left[ \| \Delta_j\|^2 \right] \\
    &= \mathbb{E}_{w_i, x_i,y_i} \mathbb{E}_{\P^i} \left[ \left< \phi_X(z_j) \otimes \phi_Y(y_i), \phi_X(z_j) \otimes \phi_Y(y_i)\right> \right.\\
    &\left. \quad - 2 \left< \phi_X(z_j) \otimes \phi_Y(y_i) , \mu_{\F^\prime_{w_i}} \otimes \phi_Y(y_i) \right> \right.\\
    &\quad \left. + \left< \mu_{\F^\prime_{w_i}} \otimes \phi_Y(y_i), \mu_{\F^\prime_{w_i}} \otimes \phi_Y(y_i) \right> \right] \\
    &\leq \mathbb{E}_{w_i,x_i,y_i} \mathbb{E}_{z_j \sim \F^\prime_{w_i}}[k((z_j,y_i),(z_j,y_i))] \\
    &\leq 1. 
\end{talign*}
Therefore, putting it all together, we have:
\begin{talign*}
&\mathbb{E}\left[\MMD_k^2(\P_{XY}, \frac{1}{n} \sum_{i=1}^n \mathbb{F}^\prime_{w_i} \times \delta_{y_i}) \right] \\
    &= \frac{1}{n^2} \sum_{i=1}^n \mathbb{E}_{w_i,x_i,y_i} \mathbb{E}_{\P^i} \left[
        \sum_{j=1}^\infty \xi_j^2 \| (\phi_X(z_j) - \mu_{\F^\prime_{w_i}}) \otimes \phi_Y(y_i) \|^2 \right] \\
        &\leq \frac{1}{n^2} \sum_{i=1}^n \mathbb{E}_{\xi \sim GEM} \left[\sum_{j=1}^\infty \xi_j^2\right] \\
        &= \frac{1}{n^2} \sum_{i=1}^n \sum_{j=1}^\infty \frac{2(c+1)^{j-1}}{(c+3)^j(c+2)} \\
        &= \frac{1}{n (c+2)}.
\end{talign*}
By Jensen's, 
\begin{talign*}
    \mathbb{E}\left[\MMD_k(\P_{XY}, \frac{1}{n} \sum_{i=1}^n \mathbb{F}^\prime_{w_i} \times \delta_{y_i}) \right] \leq \frac{1}{\sqrt{n (c+2)}}
\end{talign*}
and the final result is:
\begin{talign*}
    &\E[\MMD_k(\P^{0,n}, \P_{XY})] \leq \frac{1}{\sqrt{n(c+2)}} + \frac{c}{c+1} \sqrt{{\text{Var}}_{\mathcal{H}_{k_X}}[\F_N^0] + \MMD^2_{k_X}(\F_N^0, \F_N)} \\
    &\quad + \frac{\sqrt{2}
    }{c+1} \sqrt{\MMD_{k_X}(\F_N^0, \delta_0)}. 
\end{talign*}
\end{proof}
\begin{lemma} (\textcolor{blue}{\texttt{Berkson}}) \label{lem-joint-model}
Let $\P^{0}_{\theta} := \P_X^{0,n} \times \P_{g(\theta,\cdot)}$ and $\P_{\theta} := \P_{X} \times\P_{g(\theta, \cdot)}$. Under Assumptions \ref{asm:k}-\ref{asm-lambda}, for any $\theta \in \Theta$:
\begin{talign*}
   &\mathbb{E}_{x_{1:n} \sim \P_X^{0,n}} [ \E_{\P \sim \mu} [\MMD_{k}(\P_X^{0,n} \times\P_{g(\theta, \cdot)}, \P_X \times \P_{g(\theta, \cdot)})]] \\
   &\quad \quad \quad \leq \Lambda \left( \frac{1}{\sqrt{n(c+2)}} + \frac{c}{c+1} \MMD_{k^2_X}(\F_N^0, \F_N) + \frac{1}{c+1} \MMD_{k^2_X}(\F_N^{0}, \delta_{0}) \right)
\end{talign*}
\end{lemma}
\begin{proof}
Using Lemma \ref{lem-alq2} and Assumption 4 we have that for any $\theta \in \Theta$:
\begin{talign*}
    &\mathbb{E}_{x_{1:n} \sim \P_X^{0,n}} [ \E_{\P \sim \mu} [\MMD_{k}(\P_X^{0,n}  \times\P_{g(\theta, \cdot)}, \P_X \times\P_{g(\theta, \cdot)})]] \\
    & \quad \quad  \leq \Lambda \mathbb{E}_{x_{1:n} \sim \P_X^{0,n}} [ \E_{\P \sim \mu} [\MMD_{k^2_X}(\P_X^{0,n}, \P_X) ]] 
\end{talign*}
and the result follows by Lemma \ref{lemma-marg-X} for $k^2_X$.
\end{proof}
We are now ready to prove Theorem \ref{thm-main}:
\begin{proof}
We have that for any $\theta \in \Theta$:
\begin{talign*}
    &\MMD_k(\P^{0}, \P_{\theta^\star(\P)}^{0}) \\
    &\quad \leq \MMD_k(\P_{\theta^\star(\P)}^{0}, \P_{XY}) + \MMD_k(\P^{0}, \P_{XY}) \\
    &\quad\leq \MMD_k(\P_{\theta^\star(\P)}, \P_{XY}) + \MMD_k(\P_{\theta^\star(\P)}^{0}, \P_{\theta^\star(\P)}) \\
    &\quad \quad + \MMD_k(\P^{0}, \P_{XY}) \\
    &\quad\leq \MMD_k(\P_{\theta}, \P_{XY}) + \MMD_k(\P_{\theta^\star(\P)}^{0}, \P_{\theta^\star(\P)})  \\
    &\quad \quad + \MMD_k(\P^0, \P^{0,n}) + \MMD_k(\P^{0,n}, \P_{XY}) \\
    &\quad\leq  \MMD_k(\P_{\theta}, \P^{0}) + \MMD_k(\P_{XY}, \P^{0}) + \MMD_k(\P_{\theta^\star(\P)}^{0}, \P_{\theta^\star(\P)})  \\
    & \quad \quad + \MMD_k(\P^0, \P^{0,n}) + \MMD_k(\P^{0,n}, \P_{XY}) \\
    &\quad\leq  \MMD_k(\P_{\theta}, \P^{0}) + \MMD_k(\P_{\theta^\star(\P)}^{0}, \P_{\theta^\star(\P)}) + 2 \MMD_k(\P^0, \P^{0,n}) \\
    & \quad \quad + 2\MMD_k(\P^{0,n}, \P_{XY}) \\
    &\quad\leq  \MMD_k(\P^{0}_{\theta}, \P^{0}) + \MMD_k(\P_{\theta}^{0}, \P_{\theta}) + \MMD_k(\P_{\theta^\star(\P)}^{0}, \P_{\theta^\star(\P)}) \\
    &\quad \quad + 2 \MMD_k(\P^0, \P^{0,n}) + 2\MMD_k(\P^{0,n}, \P_{XY}) 
\end{talign*}
where we used the triangle inequality in the first two and last three inequalities and the definition of $\theta^{\star}(\P)$ in the third inequality. Since this holds for any $\theta \in \Theta$ we have that 
\begin{talign*}
    \MMD_k(\P^{0}, \P_{\theta^\star(\P)}^{0}) &\leq  \inf_{\theta \in \Theta}  \MMD_k( \P^{0}, \P^{0}_{\theta})  + \MMD_k(\P_{\theta}^{0}, \P_{\theta})  \\
    &\quad + \MMD_k(\P_{\theta^\star(\P)}^{0}, \P_{\theta^\star(\P)}) + 2 \MMD_k(\P^0, \P^{0,n})\\
    &\quad + 2\MMD_k(\P^{0,n}, \P_{XY}). 
\end{talign*}
Using lemma \ref{lemma:preliminary:MMD} with $S_i = (x_i, y_i) \sim Q_i := \P^0_{X \given w_i} \P_{Y \given x}^0$ we have that: 
\begin{talign} \label{eq:mmd-emp-x}
    \mathbb{E}_{(x,y)_{1:n} \sim \P^0}[\MMD(\P^0, \P^{0,n}) ] \leq \frac{1}{\sqrt{n}}.
\end{talign}
Using (\ref{eq:mmd-emp-x}) and Lemmas \ref{lem-joint-emp} and \ref{lem-joint-model}, we have that: 
\begin{talign*}
   &\mathbb{E}_{(x,y)_{1:n} \sim \P^0} \left[ \E_{\P \sim \mu} \left[ \MMD_k(\P^{0}, \P_{\theta^\star(\P)}^{0}) \right] \right] - \inf_{\theta \in \Theta}  \MMD_k( \P^{0}, \P^{0}_{\theta})\\
    &\leq \frac{1}{\sqrt{n}} \left(2 + \frac{2(1+\Lambda)}{\sqrt{c+2}} \right) \\
     &\quad + \frac{2c}{c+1} \left( \Lambda  \MMD_{k^2_X}(\F_N, \F_N^0) + \sqrt{{\text{Var}}_{\mathcal{H}_{k_X}}[\F_N^0] + \MMD^2_{k_X}(\F_N^0, \F_N)}\right) \\
     &\quad + \frac{2}{c+1} \left(\Lambda \MMD_{k^2_X}(\F_N^{0}, \delta_{0})  + \sqrt{2} \sqrt{\MMD_{k_X}(\F_N^0, \delta_0)}\right). 
\end{talign*}
\end{proof}
%%%%%
\subsubsection{Proof of Corollary \ref{cor-known-dist} (\textcolor{blue}{\texttt{Berkson}})} \label{app-proof-cor1}
\begin{proof}
The result follows from theorem \ref{thm-main} by noting that for $k^2_X$ characteristic, 
$$
\MMD_{k^2_{X}}(\F_N, \F_{N}^0) = 0 \quad \text{if and only if} \quad \F_N = \F_{N}^0.
$$
\end{proof}
%%%%%
\subsubsection{Proof of Corollary \ref{cor-rbf-gauss} (\textcolor{blue}{\texttt{Berkson}})} \label{app-proof-cor2}
\begin{proof}
First we have that for any $x,x' \in \R^{d_\X}$,
$$
k^2_{X}(x,x') = \left(\exp\left\{- \frac{1}{2l^2} \| x - x'\|^2\right\}\right)^2 = \exp\left\{- \frac{1}{l^2} \| x - x'\|^2\right\}
$$
By definition of the MMD,
\begin{talign}\label{eq-mmd-integrals}
    &\MMD^2_{k^2_X}(\F_N, \F_N^0) \nonumber\\
    &\quad = \int_{\R^{d_\X}} \int_{\R^{d_\X}} k^2_X(x,x') d\F_N^0 d\F_N^0 - 2 \int_{\R^{d_\X}} \int_{\R^{d_\X}} k^2_X(x,x') d\F_N^0 d\F_N \nonumber \\
    &\quad \quad + \int_{\R^{d_\X}} \int_{\R^{d_\X}} k^2_X(x,x') d\F_N d\F_N 
\end{talign}
Let $l' = \frac{l}{\sqrt{2}}$. Following the argument provided in Section 3.1 of the supplementary material in \cite{rustamov2021closed} 
and the definition of the Gaussian kernel and Gaussian density function, the second term above is:
\begin{talign*}
    &\int_{\R^{d_\X}} \int_{\R^{d_\X}} k^2_X(x,x') d\F_N^0 d\F_N \\
    &\quad= \int_{\R^{d_\X}} \int_{\R^{d_\X}} \exp\left\{- \frac{1}{l^2} \| x - x'\|^2\right\} (2 \pi \sigma_1^2)^{-d_\X/2} \exp\left\{- \frac{1}{2\sigma_1^2} \|x\|^2\right\} \\
    &\quad \quad \quad (2 \pi \sigma_2^2)^{-d_\X/2} \exp\left\{- \frac{1}{2\sigma_2^2} \|x'\|^2\right\} dx dx' \\
    &\quad= \int_{\R^{d_\X}} \int_{\R^{d_\X}} \exp\left\{- \frac{1}{2l'^2} \| x - x'\|^2\right\} (2 \pi \sigma_1^2)^{-d_\X/2} \exp\left\{- \frac{1}{2\sigma_1^2} \|x\|^2\right\} \\
    &\quad \quad \quad (2 \pi \sigma_2^2)^{-d_\X/2} \exp\left\{- \frac{1}{2\sigma_2^2} \|x'-w\|^2\right\} dx dx' \\
    &\quad= (2 \pi l'^2)^{d_\X/2} \int_{\R^{d_\X}} \left(\int_{\R^{d_\X}} (2 \pi l'^2)^{-d_\X/2} \exp\left\{- \frac{1}{2l'^2} \| x - x'\|^2\right\} \right. \\
     &\quad \quad \quad (2 \pi \sigma_1^2)^{-d_\X/2} \exp\left\{ \left. - \frac{1}{2\sigma_1^2} \|x\|^2\right\} dx \right) \\
     &\quad \quad \quad (2 \pi \sigma_2^2)^{-d_\X/2} \exp\left\{- \frac{1}{2\sigma_2^2} \|x'-w\|^2\right\}  dx'
\end{talign*}
where in the second equality we rewrote $$\exp\left\{- \frac{1}{2\sigma_2^2} \|x'\|^2\right\} = \exp\left\{- \frac{1}{2\sigma_2^2} \|x'-w\|^2\right\}$$ which holds true in the case $w = 0$. The inner integral is the density function of the sum of two Gaussian RVs with mean $0$ and variances $l^2$ and $\sigma_1^2$ respectively. We will denote the sum of these RVs $Z = X + Y$ where $X \sim N(0, l^2 I)$ and $Y \sim N(0, \sigma_1^2 I)$ such that the inner integral corresponds to the probability density function of $Z$. Hence the above becomes
\begin{talign*}
    &\int_{\R^{d_\X}} \int_{\R^{d_\X}} k^2_X(x,x') d\F_N^0 d\F_N \\
    &\quad = (2 \pi l'^2)^{d_\X/2} \int_{\R^{d_\X}} f_{Z}(x') (2 \pi \sigma_2^2)^{-d_\X/2} \exp\left\{- \frac{1}{2\sigma_2^2} \|x'- w\|^2\right\}  dx' 
\end{talign*}
The remaining integral corresponds to the probability density function of the sum of $W = Z + U$ where $Z = X + Y$ as above and $U \sim N(0, \sigma_2^2 I)$. Hence, the double integral is equal to the probability density function of $W = X + Y + U \sim N(0, (l'^2 + \sigma_1^2 + \sigma_2^2) I)$ at $w=0$ and 
\begin{talign*}
    &\int_{\R^{d_\X}} \int_{\R^{d_\X}} k^2_X(x,x') d\F_N^0 d\F_N \\
    &\quad= (2 \pi l'^2)^{d_\X/2} \int_{\R^{d_\X}} f_{Z}(x') (2 \pi \sigma_2^2)^{-d_\X/2} exp\left\{- \frac{1}{2\sigma_2^2} \|x'- w\|^2\right\}  dx' \\
    &\quad= (2 \pi l'^2)^{d_\X/2} (2 \pi (l'^2 + \sigma_1^2 + \sigma_2^2))^{-d_\X/2} \exp\left\{- \frac{1}{2(l'^2 + \sigma_1^2 + \sigma_2^2)} \| 0 \|^2\right\} \\
    &\quad= \left(\frac{l'^2}{l'^2 + \sigma_1^2 + \sigma_2^2}\right)^{d_\X/2} \\
    &\quad= \left(\frac{l^2}{l^2 + 2\sigma_1^2 + 2\sigma_2^2}\right)^{d_\X/2}.
\end{talign*}
We can use the same argument to derive the other two terms in (\ref{eq-mmd-integrals}) to obtain 
\begin{talign*}
    \MMD^2_{k^2_X}(\F_N, \F_N^0) = \left(\frac{l^2}{l^2 + 4\sigma_1^2}\right)^{d_\X/2} -  2\left(\frac{l^2}{l^2 + 2\sigma_1^2 + 2\sigma_2^2}\right)^{d_\X/2} + \left(\frac{l^2}{l^2 + 4\sigma_2^2}\right)^{d_\X/2}.
\end{talign*}
Note that using exactly the same argument, we obtain:
\begin{talign*}
    \MMD^2_{k_X}(\F_N, \F_N^0) = \left(\frac{l^2}{l^2 + 2\sigma_1^2}\right)^{d_\X/2} -  2\left(\frac{l^2}{l^2 + \sigma_1^2 + \sigma_2^2}\right)^{d_\X/2} + \left(\frac{l^2}{l^2 + 2\sigma_2^2}\right)^{d_\X/2}.
\end{talign*}
Similarly,
\begin{talign*}
    \MMD^2_{k^2_X}(\F_N^{0}, \delta_{0}) &=  \left\|\mathbb{E}_{\nu^0_{1:n} \simiid \F_N^0} [k_X^2(\nu, \cdot)] - k_X^2(0, \cdot) \right\|^2_{\mathcal{H}_{k_X^2}}\nonumber \\
    &= \E_{\nu \sim \F_{N}^0, \nu'\sim \F_N^0}[k_X^2(\nu, \nu')] -2 \E_{\nu \sim \F_N^0} [k_X^2(\nu, 0)] + k_X^2(0, 0) \nonumber \\
    &= \left(\frac{l^2}{l^2 + 4 \sigma^2_1} \right)^{\frac{d_\X}{2}} -2 \left(\frac{l^2}{l^2 + 2 \sigma^2_1} \right)^{\frac{d_\X}{2}} + 1
\end{talign*}
and
\begin{talign*}
    \MMD^2_{k_X}(\F_N^0, \delta_0) = \left(\frac{l^2}{l^2 + 2 \sigma^2_1} \right)^{\frac{d_\X}{2}} -2 \left(\frac{l^2}{l^2 +  \sigma^2_1} \right)^{\frac{d_\X}{2}} + 1.
\end{talign*}
For the variance term, we have:
\begin{talign*}
    \text{Var}_{\mathcal{H}_{k_X}}[\F_N^0] &= \E_{\nu \sim \F_N^0}[k_X(\nu, \nu)] - \left< \E_{\nu \sim\F_N^0}[\phi_X(\nu)], \E_{\nu \sim \F_N^0}[\phi_X(\nu)]\right> \\
    &= 1 - \E_{\nu \sim \F_{N}^0, \nu'\sim \F_N^0}[k_X(\nu, \nu')] \\
    &= 1 - \left(\frac{l^2}{l^2 + 2 \sigma^2_1} \right)^{\frac{d_\X}{2}}
\end{talign*}
and hence 
\begin{align*}
    &\MMD^2_{k^2_X}(\F_N^{0}, \delta_{0}) +  \text{Var}_{\mathcal{H}_{k_X}}[\F_N^0] \\
    &\quad= 1 -  2\left(\frac{l^2}{l^2 + \sigma_1^2 + \sigma_2^2}\right)^{\frac{d_\X}{2}} + \left(\frac{l^2}{l^2 + 2\sigma_2^2}\right)^{\frac{d_\X}{2}}. 
\end{align*}

So by Theorem \ref{thm-main} it follows that:
   \begin{talign*}
           &\mathbb{E}_{(x,y)_{1:n} \sim \P^{0}} \left[ \E_{\P \sim \mu} \left[  \MMD_k(\P^{0}, \P_{\theta^\star(\P)}^{0}) \right] \right] - \inf_{\theta \in \Theta}  \MMD_k(\P^{0}, \P^{0}_{\theta}) \\
       &\quad \quad \leq  \frac{1}{\sqrt{n}} \left(2 + \frac{2(1+\Lambda)}{\sqrt{c+2}} \right)  + \frac{2c}{c+1} \left( \Lambda  C_1 + C_3\right)  + \frac{2}{c+1} \left(\Lambda C_2  + \sqrt{2} \sqrt{C_4}\right).
\end{talign*}
where 
\begin{talign*}
    C_1^2 &:= \left(\frac{l^2}{l^2 + 4\sigma_1^2}\right)^{\frac{d_\X}{2}} -  2\left(\frac{l^2}{l^2 + 2\sigma_1^2 + 2\sigma_2^2}\right)^{\frac{d_\X}{2}} + \left(\frac{l^2}{l^2 + 4\sigma_2^2}\right)^{\frac{d_\X}{2}} \\
    C_2^2 &:= \left(\frac{l^2}{l^2 + 4 \sigma^2_1} \right)^{\frac{d_\X}{2}} -2 \left(\frac{l^2}{l^2 + 2 \sigma^2_1} \right)^{\frac{d_\X}{2}} + 1 \\
    C_3^2 &:= 1 -  2\left(\frac{l^2}{l^2 + \sigma_1^2 + \sigma_2^2}\right)^{\frac{d_\X}{2}} + \left(\frac{l^2}{l^2 + 2\sigma_2^2}\right)^{\frac{d_\X}{2}} \\
    C_4^2 &:= \left(\frac{l^2}{l^2 + 2 \sigma^2_1} \right)^{\frac{d_\X}{2}} -2 \left(\frac{l^2}{l^2 +  \sigma^2_1} \right)^{\frac{d_\X}{2}} + 1. 
\end{talign*} 
\end{proof}
\subsection{Proofs of theoretical results for the \textcolor{red}{\texttt{Classical}} error model} \label{app:sec-proofs-class}
\subsubsection{Proof of Theorem \ref{thm-main-class} (\textcolor{red}{\texttt{Classical}})}  \label{app-proof-thm-class}
We first prove a lemma bounding the expected MMD between the true empirical measure of the unobserved $\{x_i\}_{i=1}^n$, namely $\P_X^{n}$ and the empirical sample from the exact DPs $\P_X$ in the classical error case. 
\begin{lemma} (\textcolor{red}{\texttt{Classical}}) \label{lemma-marg-X-class}
For $i = 1, \dots, n$ let $\mathbb{P}^{i}$ be a sample from the exact Dirichlet process $\mu_i = \DP(c+1, \F'_{w_i})$ and $\P_X = \frac{1}{n} \sum_{i=1}^n \P^i$. 
Then for any kernel $k_X$ on $\mathcal{X}$ satisfying Assumptions \ref{asm:kxy} and \ref{asm-kx} the following statement holds: 
 \begin{talign*}
     &\mathbb{E}_{x_{1:n} \sim \P_X^{0}} [ \E_{\P \sim \mu} [\MMD_{k_X}(\P_X^{0}, \P_X)]] \leq \\
     &\quad \quad \quad \frac{c}{\sqrt{n}(c+1)} + \frac{1}{\sqrt{n(c+2)}} + \frac{c}{c+1} \MMD_{k_X}(\P_X^0, \F) + \frac{1}{c+1} \MMD_{k_X}(\F_N^{0}, \delta_0)
 \end{talign*}
\end{lemma}

\begin{proof}
Using the triangle inequality, we have
\begin{talign} \label{eq-target-m-class}
   \MMD_{k_X}(\mathbb{P}_X^{0},\mathbb{P}_X) \leq  \MMD_{k_X}(\P_X, \F') + \MMD_{k_X}(\mathbb{P}_X^{0}, \F'). 
\end{talign}
The first term can be bounded exactly in the same way as the first part of the proof of Lemma \ref{lemma-marg-X} since it does not depend on the error structure:
\begin{align} \label{eq-px-fprime-m-class}
     \mathbb{E}_{x_{1:n} \sim \P_X^{0}} [ \E_{\P \sim \mu} [\MMD_k(\P_X, \F^\prime)]] \leq \frac{1}{\sqrt{n(c+2)}}.
\end{align}

% %

We now bound the second term in (\ref{eq-target-m-class}). By definition of $\F'$ it follows that
\begin{talign*}
    \F' =  \frac{1}{n} \sum_{i=1}^n \F'_{w_i} = \frac{1}{n} \sum_{i=1}^n \left[\frac{c}{c+1} \F_{w_i} + \frac{1}{c+1} \delta_{w_i} \right] = \frac{c}{c+1} \F + \frac{1}{c+1} \P_W^n 
\end{talign*}
where $\F := \frac{1}{n} \sum_{i=1}^n \F_{w_i}$ and $\P_W := \frac{1}{n} \sum_{i=1} \delta_{w_i}$ is the empirical measure of the observed, noisy covariates. By linearity of expectations it follows that 
\begin{talign*}
    \E_{\F'}[\cdot] = \frac{c}{c+1} \E_{\F}[\cdot] + \frac{1}{c+1} \E_{\P_W^n} [\cdot]
\end{talign*}
and the kernel mean embedding of $\F'$ is 
\begin{talign*}
    \mu_{\F'} &:= \mathbb{E}_{z \sim \F'}[k(z, \cdot)] = \frac{c}{c+1} \E_{z \sim \F}[k(z,\cdot)] + \frac{1}{c+1} \mathbb{E}_{z \sim \P_W^n}[k(z,\cdot)] \\
    &= \frac{c}{c+1} \mu_{\F} + \frac{1}{c+1} \mu_{\P_W^n}.
\end{talign*}
Using this we can first bound the objective as:
\begin{talign}
    &\MMD_{k_X}(\P_X^{0}, \F') \nonumber \\
    &\quad= \left\| \mu_{\P_X^{0}} - \frac{c}{c+1} \mu_{\F} - \frac{1}{c+1} \mu_{\P_W^n} \right\|_{\HkX} \nonumber \\
    &\quad\leq \frac{c}{c+1} \| \mu_{\P_X^{0}} - \mu_{\F} \|_{\HkX} + \frac{1}{c+1} \| \mu_{\P_X^{0}} - \mu_{\P_W^n} \|_{\HkX} \nonumber \\
    &\quad\leq \frac{c}{c+1} \| \mu_{\P_X^{0}} - \mu_{\F} \|_{\HkX} + \frac{1}{c+1} \| \mu_{\P_X^n} - \mu_{\P_W^n} \|_{\HkX} + \frac{1}{c+1} \| \mu_{\P_X^{0}} - \mu_{\P_X^n} \|_{\HkX} \nonumber \\
    &\quad= \frac{c}{c+1} \MMD_{k_X}(\P_X^{0}, \F) + \frac{1}{c+1} \MMD_{k_X}(\P_X^n, \P_W^n) + \frac{1}{c+1} \MMD_{k_X}(\P_X^{0}, \P_X^{n}) \label{eq-three-terms-m}
\end{talign}
where we used the triangle inequality for both inequalities. Using Lemma \ref{lemma:preliminary:MMD} with $S_i = x_i \sim Q_i := \mathbb{P}^0_{X}$ we have 
\begin{align*}
    \mathbb{E}_{x_{1:n} \simiid \P_X^{0}}[\MMD_{k_X}(\mathbb{P}_X^{0}, \mathbb{P}_X^n)] \leq \frac{1}{\sqrt{n}}.
\end{align*}
%
%%%%%
Next, let $\phi_X$ be the feature map associated with the kernel function $k_X$ and let $k_{X}^i(\nu, \nu') = k_X(w_i - \nu, w_i - \nu')$ for each $i \in \{1, \dots, n\}$ where $\{w_i\}_{i=1}^n$ denote the observations. By assumption \ref{asm-kx}, $k_X$ is translation invariant so for any $i \in \{1, \dots, n\}$,
$$
k_X^i(\nu, \nu') = k_X(w_i - \nu, w_i - \nu') = \psi(w_i - \nu - w_i + \nu') = \psi(\nu' - \nu) = k_X(\nu',\nu).
$$
hence
\begin{talign*}
    \MMD^2_{k_X^i}(\F_N^0, \delta_0) &=  \mathbb{E}_{\nu \sim \F_N^0, \nu' \sim \F_N^0} [k_X^i(\nu, \nu')] -2 \mathbb{E}_{\nu \sim \F_N^0} [k_X^i(\nu, 0)] + k_X^i(0, 0) \\
    &= \mathbb{E}_{\nu \sim \F_N^0, \nu' \sim \F_N^0} [k_X(\nu', \nu)] -2 \mathbb{E}_{\nu \sim \F_N^0} [k_X(0, \nu )] + k_X(0, 0) \\
    &= \MMD^2_{k_X}(\F_N^0, \delta_0)
\end{talign*}
where we have used the fact that kernels are symmetric by definition of a Reproducing Kernel Hilbert Space. Using this result we have,
\begin{talign*}
     &\E_{x_{1:n} \sim \P^0_X} [\MMD_{k_X}(\P_X^n, \P_W^n)] \\
     &\quad = \E_{x_{1:n} \sim \P^0} \left[\left\| \frac{1}{n} \sum_{i=1}^n \left[\phi_X(x_i) - \phi_X(w_i)  \right] \right\|_{\mathcal{H}_{k_X}} \right]\\
    &\quad = \E_{\nu^0_{1:n} \sim \F_N^0} \left[ \left\| \frac{1}{n} \sum_{i=1}^n \left[\phi_X(w_i - \nu_i^0) - \phi_X(w_i)  \right] \right\|_{\mathcal{H}_{k_X}} \right] \\
    &\quad = \left\| \frac{1}{n} \sum_{i=1}^n \left[ \E_{\nu^0_{i} \sim \F_N^0}[k^i_X(\nu_i^0, \cdot)] - k^i_X(0,\cdot)  \right] \right\|_{\mathcal{H}_{k_X}} \\
    &\quad \leq \frac{1}{n} \sum_{i=1}^n \MMD_{k^i_X}(\F_N^{0}, \delta_{0}) \\
    &\quad = \MMD_{k_X}(\F_N^{0}, \delta_{0}). 
\end{talign*}
Substituting in (\ref{eq-three-terms-m}), we obtain:
\begin{talign} \label{eq-pxon-fprime-m}
    \mathbb{E}_{x_{1:n} \sim \P_X^{0}}[\MMD_{k_X}(\P_X^{0}, \F')] &\leq \frac{1}{c+1} \frac{1}{\sqrt{n}} + \frac{c}{c+1} \MMD_{k_X}(\P_X^0, \F) \nonumber\\
    &+ \frac{1}{c+1} \MMD_{k_X}(\F_N^{0}, \delta_{0}). 
 \end{talign}
 Finally, using (\ref{eq-px-fprime-m-class}) and (\ref{eq-pxon-fprime-m}) in (\ref{eq-target-m-class}) we obtain the required result:
  \begin{talign*}
     &\mathbb{E}_{x_{1:n} \sim \P_X^{0}} [ \E_{\P \sim \mu} [\MMD_{k_X}(\P_X^{0}, \P_X)]] \leq \\
     &\quad \quad \quad \frac{c}{\sqrt{n}(c+1)} + \frac{1}{\sqrt{n(c+2)}} + \frac{c}{c+1} \MMD_{k_X}(\P_X^0, \F) + \frac{1}{c+1} \MMD_{k_X}(\F_N^{0}, \delta_0)
 \end{talign*}
\end{proof}
%
%%%%%%%%%%%%%%%%%%%%%%%%%%%%%%%%%%%%%%%%%%%
\begin{lemma} (\textcolor{red}{\texttt{Classical}}) \label{lem-joint-emp-m}
Let 
\begin{align*}
    \mathbb{P}^{0,n} := \frac{1}{n} \sum_{i=1}^n \delta_{(x_i, y_i)} \quad \text{and} \quad \mathbb{P}_{XY} = \P_X \times \P_Y^{n} = \frac{1}{n} \sum_{i=1}^n \sum_{j=1}^{\infty} \xi_j^i \delta_{(z_j,y_i)}
\end{align*}
where $z^i_{1:\infty} \overset{iid}{\sim} \F'_{w_i}, \xi^i_{1:\infty} \sim \GEM(c+1)$. Then under Assumptions \ref{asm:k}-\ref{asm-kx}:
\begin{talign*}
    &\mathbb{E}_{(x,y)_{1:n} \sim \P^0} [ \E_{\P \sim \mu} [\MMD_k(\mathbb{P}^{0}, \mathbb{P}_{XY})]] \\
    &\quad  \leq \frac{1}{\sqrt{n(c+2)}} + \frac{c}{n(c+1)}  \sum_{i=1}^n \sqrt{{\text{Var}}_{\mathcal{H}_{k_X}}[\P_X^0] + \MMD^2_{k_X}(\P_X^0, \F_{X \mid w_i})} \\
    &\quad \quad  + \frac{\sqrt{2}
    }{c+1} \sqrt{\MMD_{k_X}(\F_N^0, \delta_0)}. 
\end{talign*}
\end{lemma}
\begin{proof}
The proof follows exactly the proof of the Berkson case in Lemma \ref{lem-joint-emp} with a few changes to account for the different error structure. To avoid repetition, we only note the differences. 
The only part that depends on the ME structure is the bound on the term 
\begin{talign*}
\sum_{i=1}^{n} \sqrt{\E_{w_i,x_i}[\|\phi_X(x_i) - \mu_{\F_{X \mid w_i}} \|^2_{\mathcal{H}_{k_X}}]}
\end{talign*}
as this depends directly on the structure of the prior centring measure $\F_{w_i}$ relative to the true distribution of each $(x_i,w_i)$. For every $i \in \{1, \dots, n\}$ we have:
\begin{talign*}
    &\E_{w_i,x_i}[\|\phi_X(x_i) - \mu_{\F_{X \mid w_i}} \|^2_{\mathcal{H}_{k_X}}] \\
    &=\E_{w_i, x_i}[\| \phi_X(x_i) - \mathbb{E}_{x \sim \F_{X \mid w_i}}[\phi_X(x)]\|^2] \\
    &= \E_{x_i \sim \P^0}[k(x_i, x_i)] - 2 \E_{x_i \sim \P^0}[\left<\phi_X(x_i), \E_{x \sim \F_{X \mid w_i}} \phi_X(x) \right>] \\
    &\quad \quad + \left< \E_{x \sim \F_{w_i}}[\phi_X(x)], \E_{x \sim \F_{X \mid w_i}}[\phi_X(x)] \right>_{\mathcal{H}_{k_X}} \\
    &= \E_{x_i \sim \P^0}[k(x_i, x_i)] - \left< \mu_{\P_X^0}, \mu_{\P_X^0}\right>_{\mathcal{H}_{k_X}} + \MMD^2_{k_X}(\P_X^0, \F_{X \mid w_i}) \\
    &= \text{Var}_{\mathcal{H}_{k_X}}[\P_X^0] + \MMD^2_{k_X}(\P_X^0, \F_{X \mid w_i})
\end{talign*}
and hence 
\begin{talign*}
&\sum_{i=1}^{n} \sqrt{\E_{w_i,x_i}[\|\phi_X(x_i) - \mu_{\F_{X \mid w_i}} \|^2_{\mathcal{H}_{k_X}}]}\\
& \quad  =\sum_{i=1}^n \sqrt{\text{Var}_{\mathcal{H}_{k_X}}[\P_X^0] + \MMD^2_{k_X}(\P_X^0, \F_{X \mid w_i})}.
\end{talign*}
Unsurprisingly, this term now depends on the RKHS variance of the true marginal of $X$ as well as the quality of the prior centring measure with respect to $\P_X^0$. This is because in the Classical error case, $X$ is not dependent on the ME, and hence if $X$ has a large variance, then $W$ will have a larger variance and the problem becomes harder. The rest of the proof is identical to Lemma \ref{lem-joint-emp}. 
\end{proof}
\begin{lemma} (\textcolor{red}{\texttt{Classical}}) \label{lem-joint-model-m}
Let $\P^{0}_{\theta} := \P_X^{0} \times\P_{g(\theta,\cdot)}$ and $\P_{\theta} := \P_{X} \times\P_{g(\theta, \cdot)}$ and suppose Assumptions \ref{asm:k}-\ref{asm-lambda} hold, then for any $\theta \in \Theta$:
\begin{talign*}
   &\mathbb{E}_{x_{1:n} \sim \P_X^{0}} [ \E_{\P \sim \mu} [\MMD_{k}(\P_X^{0} \times\P_{g(\theta, \cdot)}, \P_X \times\P_{g(\theta, \cdot)})]] \\
   &\quad  \leq \Lambda \left(\frac{1}{\sqrt{n}(c+1)} + \frac{1}{\sqrt{n(c+2)}} + \frac{c}{c+1} \MMD_{k^2_X}(\P_X^0, \F) + \frac{1}{c+1} \MMD_{k^2_X}(\F_N^{0}, \delta_{0}) \right)
\end{talign*}
\end{lemma}
\begin{proof}
Using Lemma \ref{lem-alq2} and Assumption 4 we have that for any $\theta \in \Theta$:
\begin{talign*}
    &\mathbb{E}_{x_{1:n} \sim \P_X^{0}} [ \E_{\P \sim \mu} [\MMD_{k}(\P_X^{0}  \times\P_{g(\theta, \cdot)}), \P_X \times\P_{g(\theta, \cdot)}]] \\
    & \quad \leq \Lambda \mathbb{E}_{x_{1:n} \sim \P_X^{0}} [ \E_{\P \sim \mu} [\MMD_{k^2_X}(\P_X^{0}, \P_X) ]] 
\end{talign*}
and the result follows by Lemma \ref{lemma-marg-X-class} for $k^2_X$.
\end{proof}
We are now ready to prove Theorem \ref{thm-main-class}:
\begin{proof}
We have that for any $\theta \in \Theta$:
\begin{align*}
    &\MMD_k(\P^{0}, \P_{\theta^\star(\P)}^{0}) \\
    &\quad \leq \MMD_k(\P_{\theta^\star(\P)}^{0}, \P_{XY}) + \MMD_k(\P^{0}, \P_{XY}) \\
    &\quad\leq \MMD_k(\P_{\theta^\star(\P)}, \P_{XY}) + \MMD_k(\P_{\theta^\star(\P)}^{0}, \P_{\theta^\star(\P)}) + \MMD_k(\P^{0}, \P_{XY}) \\
    &\quad\leq \MMD_k(\P_{\theta}, \P_{XY}) + \MMD_k(\P_{\theta^\star(\P)}^{0}, \P_{\theta^\star(\P)}) + \MMD_k(\P^0, \P^{0,n}) \\
    &\quad \quad + \MMD_k(\P^{0,n}, \P_{XY}) \\
    &\quad\leq  \MMD_k(\P_{\theta}, \P^{0}) + \MMD_k(\P_{XY}, \P^{0}) + \MMD_k(\P_{\theta^\star(\P)}^{0}, \P_{\theta^\star(\P)})  \\
    & \quad \quad + \MMD_k(\P^0, \P^{0,n}) + \MMD_k(\P^{0,n}, \P_{XY}) \\
    &\quad\leq  \MMD_k(\P_{\theta}, \P^{0}) + \MMD_k(\P_{\theta^\star(\P)}^{0}, \P_{\theta^\star(\P)}) + 2 \MMD_k(\P^0, \P^{0,n}) \\
    & \quad \quad + 2\MMD_k(\P^{0,n}, \P_{XY}) \\
    &\quad\leq  \MMD_k(\P^{0}_{\theta}, \P^{0}) + \MMD_k(\P_{\theta}^{0}, \P_{\theta}) + \MMD_k(\P_{\theta^\star(\P)}^{0}, \P_{\theta^\star(\P)}) \\
    &\quad \quad + 2 \MMD_k(\P^0, \P^{0,n}) + 2\MMD_k(\P^{0,n}, \P_{XY}) 
\end{align*}
where we used the triangle inequality in the first three and last two inequalities and the definition of $\theta^{\star}(\P)$ in the third inequality. 
Since this holds for any $\theta \in \Theta$ we have that 
\begin{align*}
    \MMD_k(\P^{0}, \P_{\theta^\star(\P)}^{0}) &\leq  \inf_{\theta \in \Theta}  \MMD_k( \P^{0}, \P^{0}_{\theta}) +  \MMD_k(\P^{0}_{\theta}, \P^{0}) + \MMD_k(\P_{\theta}^{0}, \P_{\theta})  \\
    &\quad + \MMD_k(\P_{\theta^\star(\P)}^{0}, \P_{\theta^\star(\P)}) + 2 \MMD_k(\P^0, \P^{0,n})\\
    &\quad + 2\MMD_k(\P^{0,n}, \P_{XY}). 
\end{align*}
Using lemma \ref{lemma:preliminary:MMD} with $S_i = (x_i, y_i) \sim Q_i := \P^0_{X \given w_i} \times \P_{Y \given x}^0$ we have that: 
\begin{align} \label{eq:mmd-emp}
    \mathbb{E}_{(x,y)_{1:n} \sim \P^0}[\MMD(\P^0, \P^{0,n}) ] \leq \frac{1}{\sqrt{n}}.
\end{align}
Using equation (\ref{eq:mmd-emp}) and lemmas \ref{lem-joint-emp-m} and \ref{lem-joint-model-m}, we have that: 
\begin{talign*}
   &\mathbb{E}_{(x,y)_{1:n} \sim \P^0} \left[ \E_{\P \sim \mu} \left[ \MMD_k(\P^{0}, \P_{\theta^\star(\P)}^{0}) \right] \right] - \inf_{\theta \in \Theta}  \MMD_k( \P^{0}, \P^{0}_{\theta})\\
    &\quad \quad \leq 
    \frac{1}{\sqrt{n}} \left(\frac{2 \Lambda}{c+1} + \frac{2 \Lambda + 2}{\sqrt{c+2}} \right) \\
    &\quad \quad + \frac{2c}{c+1} \left(\Lambda  \MMD_{k^2_X}(\P_X^0, \F) + \frac{1}{n} \sum_{i=1}^n \sqrt{{\text{Var}}_{\mathcal{H}_{k_X}}[\P_X^0] + \MMD^2_{k_X}(\P_X^0, \F_{w_i})} \right) \\
    &\quad \quad + \frac{2}{c+1} \left(\MMD_{k^2_X}(\F_N^{0}, \delta_{0}) + \sqrt{2} \sqrt{\MMD_{k_X}(\F_N^0, \delta_0)} \right).
\end{talign*}
\end{proof}
%%%%%
\subsubsection{Proof of Corollary \ref{cor-known-dist-class} (\textcolor{red}{\texttt{Classical}})} \label{app-proof-cor3}
\begin{proof}
The result follows from Theorem \ref{thm-main-class} by noting that for $k_X$ characteristic, 
$$
\MMD_{k_{X}}(\P_X^0, \F) = 0 \quad \text{if and only if} \P_X^0 = \F.
$$
and similarly for $k_X^2$.
\end{proof}

\subsubsection{Proof of Corollary \ref{cor-rbf-gauss-class} (\textcolor{red}{\textcolor{red}{\texttt{Classical}}})}
\begin{proof}
    The proof is analogous to the Berkson case, making use of the closed-form expression of the MMD between Gaussian distributions. The only difference lies in the fact that some of the inputs to the MMD terms in the Classical case are not zero-mean, meaning that the double term in the expansion of the MMD is slightly different. In particular, for two Gaussian distributions $\P_1 := \mathcal{N}(\mu_1, \Sigma_1)$ and $\P_2 := \mathcal{N}(\mu_2, \Sigma_2)$ we have:
    \begin{talign*}
        &\int_{\R^{d_\X}} \int_{\R^{d_\X}} k_X(x,x^\prime) d \P_1 d \P_2 \\
        &=  \left| I + \tfrac{1}{\ell^2}\big(\Sigma_1 + \Sigma_2\big) \right|^{-1/2}
\exp\!\left(
  -\tfrac{1}{2}(\mu_1-\mu_2)^\top
  \big(\Sigma_1+ \Sigma_2 + \ell^2 I\big)^{-1}
  (\mu_1 - \mu_2) \right).
    \end{talign*}
To see this note that \[
\iint k_X(x,x')\,d\mathbb P_1(x)\,d\mathbb P_2(x')
=\mathbb E\left[\exp\!\left(-\frac{\|x-x^\prime\|^2}{2\ell^2}\right)\right].
\]
Set $z:=x-x^\prime\sim\mathcal N(m,S)$ with $m:=\mu_1-\mu_2$ and $S:=\Sigma_1+\Sigma_2$.
Writing $f(z;m,S)$ for the $\mathcal N(m,S)$ density,
\[
\mathbb E\exp\!\left(-\frac{\|z\|^2}{2\ell^2}\right)
= \int_{\mathbb R^d}\exp\!\left(-\frac{1}{2\ell^2}z^\top z\right)f(z;m,S)\,dz.
\]
Combining exponents and completing the square we obtain:
\[
-\tfrac12(z-m)^\top S^{-1}(z-m)-\tfrac{1}{2\ell^2}z^\top z
= -\tfrac12(z-\tilde m)^\top B\,(z-\tilde m)
  -\tfrac12\,m^\top(S+\ell^2 I)^{-1}m,
\]
where
\[
B:=S^{-1}+\ell^{-2}I,\qquad
\tilde m:=B^{-1}S^{-1}m=(I+\ell^{-2}S)^{-1}m,
\]
and we used $S^{-1}-S^{-1}B^{-1}S^{-1}=(S+\ell^2 I)^{-1}$.
Thus
\[
\mathbb E\left[\exp\!\left(-\frac{\|z\|^2}{2\ell^2}\right)\right]
=(2\pi)^{-d/2}|S|^{-1/2}\,e^{-\frac12 m^\top(S+\ell^2 I)^{-1}m}
\int_{\mathbb R^d} e^{-\frac12(z-\tilde m)^\top B (z-\tilde m)}dz.
\]
The Gaussian integral equals $(2\pi)^{d/2}|B|^{-1/2}$. Since
$|B|=|S^{-1}+\ell^{-2}I|=|S^{-1}|\,|I+\ell^{-2}S|$,
we obtain
\[
|S|^{-1/2}\,|B|^{-1/2} = |I+\ell^{-2}S|^{-1/2}.
\]
Substituting $S=\Sigma_1+\Sigma_2$ and $m=\mu_1-\mu_2$ yields the stated formula.
    The rest of the terms are obtained using the same arguments as in the proof of Corollary \ref{cor-rbf-gauss} in Section \ref{app-proof-cor2}.
\end{proof}

\subsection{Extension of theoretical results to the DP approximation}
The proofs of the theoretical results in the previous sections were based on the exact stick-breaking process representation of the DP. However, in practice, we employ the Dirichlet approximation approach (see Section \ref{sec-npl-input-unc} and Equation (\ref{eq-dp-approx})). In this Section, we investigate any error incurred from this approximation. We start by examining the generalisation error results of Section \ref{sec-theory} under the approximated DP.

Note that throughout the proofs, the only terms that relied on the DP representation were the ones involving the MMD between a random DP sample (which was represented exactly via an infinite sum) and the posterior DP centring measure, i.e. the term $\mathbb{E}_{\P \sim \mu} [\MMD_{k_X}(\P_X, \F')]$ in Lemmas \ref{lemma-marg-X} and \ref{lemma-marg-X-class} and the joint term $\E_{(x,y)_{1:n} \sim \P^0}[\E_{\P \sim \mu}[\MMD_k(\P_{XY}, \frac{1}{n} \sum_{i=1}^n \F^\prime_{w_i} \times \delta_{y_i})]]$ in Lemmas \ref{lem-joint-emp} and \ref{lem-joint-emp-m}. Since these terms only rely on the DP randomness relative to the DP mean, they are the same in the Berkson and Classical cases. We demonstrate the marginal case below, as the joint case follows exactly the same. We have:
\begin{talign*}
%\mathbb{E}_{\P_{X}^{1:n} \sim \mathbb{D}_{1:n}}
    &\mathbb{E}_{\P \sim \mu} [\MMD_{k_X}^2(\P_X, \F')] \nonumber\\
    &\quad \quad = \mathbb{E}_{\P \sim \mu} [\| \mu_{\P_X} - \mu_{\F'} \|_{\mathcal{H}_{k_X}}^2] \nonumber\\
    &\quad\quad=  \mathbb{E}_{\P \sim \mu} \left[\left\| \frac{1}{n} \sum_{i=1}^n (\mu_{\P^i} - \mu_{\F'_{w_i}}) \right\|^2_{\mathcal{H}_{k_X}} \right] \nonumber \\
    &\quad \quad= \frac{1}{n^2} \sum_{i=1}^n \mathbb{E}_{\P \sim \mu} \left[ \| \mu_{\P^i} - \mu_{\F^\prime_{w_i}} \|^2_{\HkX}\right] \\
    &\quad \quad \quad \quad + \frac{2}{n^2} \sum_{i \neq j} \E_{\P \sim \mu} \left[\left<\mu_{\P^i} - \mu_{\F^{\prime}_{w_i}}, \mu_{\P^j} - \mu_{\F^{\prime}_{w_j}} \right>_{\mathcal{H}_{k_X}} \right] 
\end{talign*}
Under the approximation we have that for each $i \in \{1, \dots, n\}$, $\P^i \sim \mu_i$ means that $\P^i := \sum_{t=1}^T \xi_t^i \delta_{\tilde{x}_t^i} + \xi^i_{T+1} \delta_{w_i}$ where $\xi^i_{1:T+1} \sim \Dir \left(\frac{c}{T}, \dots, \frac{c}{T}, 1 \right)$ and $\tilde{x}^i_{1:T} \simiid \F_{w_i}$. Define by $\E_{\xi^i \sim \Dir}$ the expectation under the Dirichlet distribution of the weights $\xi^i = (\xi^i_1, \dots, \xi^i_{T+1})$. Similarly to the exact case we have:
\begin{align*}
    \E[\mu_{\P^i}] &= \E \left[\sum_{t=1}^T \xi_t^i k_X(\tilde{x}_t^i, \cdot) + \xi^i_{T+1} k_X(w_i, \cdot) \right] \\ &= \sum_{t=1}^T \E_{\xi^i \sim \Dir}[\xi_T^i] \mu_{\F_{w_i}} + \E_{\xi^i \sim \Dir} [\mu_{\delta_{w_i}}] \\
    &= \mu_{\F_{w_i}} \sum_{t=1}^T \frac{c}{T(c+1)} + \frac{1}{c+1} \mu_{\delta_{w_i}} \\
    &= \frac{c}{c+1} \mu_{\F_{w_i}} + \frac{1}{c+1} \mu_{\delta_{w_i}} \\
    &= \mu_{\F^\prime_{w_i}} 
\end{align*}
where we used the first moment of the Dirichlet distribution.

Then for each $i \neq j$, we have:
\begin{talign*}
& \E_{\P\sim \mu} \left[\left< \mu_{\P^i} - \mu_{\F^\prime_{w_i}}, \mu_{\P^j} - \mu_{\F^\prime_{w_j}} \right>_{\mathcal{H}_{k_X}} \right]\nonumber \\
&= \E_{\P \sim \mu} [\left<\mu_{\P^i}, \mu_{\P^j} \right>] - 2 \E_{\P\sim \mu} \left[\left<\mu_{\P^i}, \mu_{\F^\prime_{w_j}} \right>_{\mathcal{H}_{k_X}}\right] + \left<\mu_{\F^\prime_{w_i}}, \mu_{\F^\prime_{w_j}} \right>_{\mathcal{H}_{k_X}} \nonumber\\
&= 0.
\end{talign*}
Therefore, we have:
\begin{talign}\label{eq:idk3}
    \E_{\P \sim \mu} [\MMD_{k_X}^2(\P_X, \F')] = \frac{1}{n^2} \sum_{i=1}^n \mathbb{E}_{\P \sim \mu} \left[ \| \mu_{\P^i} - \mu_{\F^\prime_{w_i}} \|^2_{\HkX}\right] . 
\end{talign}
For each $i \in \{1, \dots, n\}$ we have:
\begin{talign*}
    &\mathbb{E}_{\P \sim \mu} \left[ \| \mu_{\P^i} - \mu_{\F^\prime_{w_i}} \|^2_{\HkX}\right] \\
    &= \mathbb{E}_{\P \sim \mu} \left[ \left< \mu_{\F^\prime_{w_i}} - \mu_{\P^i}, \mu_{\F^\prime_{w_i}} - \mu_{\P^i} \right>_{\HkX}\right] \\
    &= \mathbb{E}_{\P \sim \mu} \left[ \| \mu_{\P^i}  \|^2_{\HkX} + \| \mu_{\F^\prime_{w_i}}  \|^2_{\HkX} - 2 \left<\mu_{\F^\prime_{w_i}}, \mu_{\P^i} \right>_{\HkX}\right]\\
    &= \| \mu_{\F^\prime_{w_i}}  \|^2_{\HkX} - 2 \| \mu_{\F^\prime_{w_i}}  \|^2_{\HkX} + \E_{\P^i \sim \mu_i}[\| \mu_{\P^i}  \|^2_{\HkX}] \\
    &\leq \E_{\P^i \sim \mu_i}[\| \mu_{\P^i}  \|^2_{\HkX}] \\
    &= \E_{\xi^i \sim \Dir}\left[\E_{\tilde{x}^i_{1:T} \simiid \F_{w_i}}\left[ \left\| 
    \sum_{t=1}^T \xi^i_t \phi_X(\tilde{x}^i_t) + \xi^i_{T+1} \phi_X(w_i)
    \right\|_{\HkX}^2 \right]\right] \\
    &= \E_{\xi^i \sim \Dir}\left[\E_{\tilde{x}^i_{1:T} \simiid \F_{w_i}}\left[ 
    \sum_{t=1}^T (\xi_t^i)^2 \| \phi_X(\tilde{x}^i_t)\|^2_{\HkX} + (\xi^i_{T+1})^2 \| \phi_X(w_i)\|^2_{\HkX} \right. \right. \\
    &\left. \left. \quad + 2 \sum_{j \neq t}^T \xi^i_j \xi^i_t k_X(\tilde{x}^i_j, \tilde{x}^i_t
    ) + 2 \sum_{t=1}^T \xi^i_t \xi^i_{T+1} k(\tilde{x}^i_t, w_i)
    \right]\right] \\
    &\leq \sum_{t=1}^T \E[(\xi_t^i)^2] + \E[(\xi^i_{T+1})^2] + 2 \sum_{j \neq t} \left| \E[ \xi^i_j \xi^i_t] \right| + 2 \sum_{t=1}^T |\E[\xi^i_t \xi^i_{T+1}] | \\
    &= \frac{c \left(\frac{c}{T} + 1\right)}{(c+1)(c+2)} + \frac{2}{(c+1)(c+2)} + \frac{T-1}{T} \frac{c^2}{(c+1)(c+2)}+ \frac{2c}{(c+1)(c+2)} \\
    &= 1
\end{talign*}
where we used Assumption \ref{asm:kxy} about the boundedness of $k_X$ and the closed-form expressions of the moments of a Dirichlet distribution. Substituting into (\ref{eq:idk3}) we obtain:
\begin{talign*}
    \E_{\P \sim \mu} [\MMD_{k_X}^2(\P_X, \F')] &= \frac{1}{n^2} \sum_{i=1}^n \mathbb{E}_{\P \sim \mu} \left[ \| \mu_{\P^i} - \mu_{\F^\prime_{w_i}} \|^2_{\HkX}\right] \\
    &\leq \frac{1}{n}.
\end{talign*}
Hence, by Jensen's inequality, we have 
\begin{talign*}
    \E_{\P \sim \mu} [\MMD_{k_X}(\P_X, \F')] \nonumber= \frac{1}{n^2} \sum_{i=1}^n \mathbb{E}_{\P \sim \mu} \left[ \| \mu_{\P^i} - \mu_{\F^\prime_{w_i}} \|^2_{\HkX}\right] \leq \frac{1}{\sqrt{n}}.
\end{talign*}
This result shows that the DP approximation does not affect the rate of the generalisation error. Moreover, the truncation limit $T$ does not enter the bound, meaning that the same theoretical rate can be guaranteed irrespective of the choice of $T$.

\subsection{Proof of Minimax Optimality (Proposition \ref{prp:minimax})}
We are going to use Le Cam's method, following the results in \cite{tolstikhin2016minimax}. First, define the RBF kernels: 
    \begin{align*}
        k((x,y), (x^\prime, y^\prime)) &= k_X(x,x^\prime) k_Y(y, y^\prime), \\
        k_X(x,x^\prime) &= \exp\left(- \frac{\| x - x^\prime\|^2}{2 l_X^2} \right),\\
        k_Y(y,y^\prime) &= \exp\left(- \frac{\| y - y^\prime\|^2}{2 l_Y^2} \right).
    \end{align*}
We start with the following Lemma, which will allow us to compute the MMD between joint Gaussian distributions in closed form. 
\begin{lemma}\label{lem:e9}
    Let two probability measures $\P, \Q \in \mathcal{P}_\Z$ such that $\P := \P_X \times \P_{Y \mid X}$ and $\Q := \P_X \times \Q_{Y \mid X}$. Then if $k((x,y),(x^\prime, y^\prime)) = k_X(x,x^\prime) k_Y(y, y^\prime)$ we have:
    \begin{align*}
        \MMD_k^2(\P, \Q) = \E_{x, x^\prime \sim \P_X} \left[k(x,x^\prime) \MMD^2_{k_Y}(\P_{Y \mid X}, \Q_{Y \mid X})\right].
    \end{align*}
\begin{proof}
 Recall the standard expansion
\begin{align} \label{eq:mmd-exp}
\mathrm{MMD}^2_k(\P,\Q)
= \mathbb{E}_{z,z'\sim \P}k(z,z')
+ \mathbb{E}_{z,z'\sim \Q}k(z,z')
-2\,\mathbb{E}_{z\sim \P, z^\prime \sim \Q}k(z,z').
\end{align}
Write $z=(x,y)$ and $z'=(x',y')$.
Because $\P=\P_X\times \P_{y|x}$, the joint law of $(x,y,x',y')$ under $\P\times \P$
is the product measure
\[
\P_X(dx)\,\P_{Y|X=x}(dy)\,\P_X(dx')\,\P_{Y|X=x'}(dy').
\]
Using boundedness of $k_X,k_Y$  we have
$|k_X(x,x')k_Y(y,y')|\leq 1$, hence the integrand is absolutely integrable.
Therefore, we may apply Fubini–Tonelli to permute integrals. Hence, 
\begin{align*}
&\mathbb{E}_{z,z'\sim \P}k(z,z')\\
&=\int_{\Z}\!\!\int_{\Z} k_X(x,x')\,k_Y(y,y')\,
\P_X(dx)\,\P_{Y|X=x}(dy)\,\P_X(dx')\,\P_{Y|X=x'}(dy')\\
&=\int_{\X}\!\!\int_{\X} \Bigg[
k_X(x,x')\,
\underbrace{\int_{\Y}\!\!\int_{\Y} k_Y(y,y')\,\P_{Y|X=x}(dy)\,\P_{Y|X=x'}(dy')}_{:=\;\Gamma_{\P\P}(x,x')}
\Bigg]\,\P_X(dx)\,\P_X(dx') \\
&=\mathbb{E}_{x,x'\sim \P_X}\!\big[k_X(x,x')\,\Gamma_{\P\P}(x,x')\big].
\end{align*}
Similarly, for the second term, we get:
\begin{align*}
&\mathbb{E}_{z,z'\sim \Q}k(z,z') \\
&=\int_{\Z}\!\!\int_{\Z} k_X(x,x')\,k_Y(y,y')\,
\P_X(dx)\,\Q_{Y|X=x}(dy)\,\P_X(dx')\,\Q_{Y|X=x'}(dy')\\
&=\mathbb{E}_{x,x'\sim \P_X}\!\big[k_X(x,x')\,\Gamma_{\Q\Q}(x,x')\big],
\end{align*}
where $\displaystyle
\Gamma_{\Q\Q}(x,x'):=\int_{\Y}\!\!\int_{\Y} k_Y(y,y')\,\Q_{Y|X=x}(dy)\,\Q_{Y|X=x'}(dy')$.
For the third term, $(x,y)$ is distributed as $\P$ and $(x',y')$ as $\Q$, independently, so
\begin{align*}
&\mathbb{E}_{z\sim \P, z^\prime \sim \Q}k(z,z')\\
&=\int_{\Z}\!\!\int_{\Z} k_X(x,x')\,k_Y(y,y')\,
\P_X(dx)\,\P_{Y|X=x}(dy)\,\P_X(dx')\,\Q_{Y|X=x'}(dy')\\
&=\mathbb{E}_{x,x'\sim \P_X}\!\big[k_X(x,x')\,\Gamma_{\P\Q}(x,x')\big],
\end{align*}
where $\displaystyle
\Gamma_{\P\Q}(x,x'):=\int_{\Y}\!\!\int_{\Y} k_Y(y,y')\,\P_{Y|X=x}(dy)\,\Q_{Y|X=x'}(dy')$.
Substituting into \ref{eq:mmd-exp} we get:
\begin{align*}
\mathrm{MMD}^2_k(\P,\Q)
&=\mathbb{E}_{x,x'}\!\left[
k_X(x,x')\left(\Gamma_{\P\P}(x,x')+\Gamma_{\Q\Q}(x,x')-2\Gamma_{\P\Q}(x,x')\right)
\right].
\end{align*}
For any fixed $x,x'$, we have:
\[
\Gamma_{\P\P}(x,x')+\Gamma_{\Q\Q}(x,x')-2\Gamma_{\P\Q}(x,x')
=\MMD^2_{k_Y}\!\big(\P_{Y|X=x},\,\Q_{Y|X=x'}\big).
\]
Hence
$
\MMD_k^2(\P,\Q)
=\mathbb{E}_{x,x'\sim \P_X}\!\Big[
k_X(x,x')\;\MMD^2_{k_Y}\!\big(\P_{Y|X},\,\Q_{Y|X'}\big)
\Big],
$
as claimed.
\end{proof}
\end{lemma}
 Now we denote the distribution of $(W,X,Y)$ as $\P$ and recall that $\P^0$ denotes the joint distribution of $X,Y$. Further recall that we denote the model family for the regression relationship as $\{ \P_{g(\theta, \cdot)}: \theta \in \Theta\}$ and the associated model under the true covariate distribution as $\P^0_{\theta} := \P^0_{X} \times \P_{g(\theta, X)}$. We are now ready to prove the result. 
\begin{proof}
   Consider the case where there is no measurement error, i.e. $X \equiv W$ and we observe $\{(x_i, y_i)_{i=1}^n\} \sim \P^0$ and the model is well-specified, i.e. $\exists \theta^\star$ such that $\P_{g(\theta^\star, X)} \equiv \P^0_{Y \mid X}$. Since this is a special case, it trivially lower bounds the supremum i.e. 
   \begin{align*}
       &\inf_{\hat{\theta}_n} \sup_{\P} \E_{(w_i,x_i,y_i)_{1:n} \simiid \P}[\MMD_k(\P^0, \P_{\hat{\theta}_n})]\\ &\quad \geq \inf_{\hat{\theta}_n} \sup_{\substack{{\P^0_{\theta^\star}}: \\ \text{no ME}}} \E_{(x_i,y_i)_{1:n} \sim \P^0_{\theta^\star}}[\MMD_k(\P^0_{\theta^\star}, \P_{\hat{\theta}_n})] \\
       &\quad= \inf_{\hat{\theta}_n} \sup_{\substack{{\P^0_{\theta^\star}}}} \E_{(x_i,y_i)_{1:n} \sim \P^0_{\theta^\star}}[\MMD_k(\P^0_{\theta^\star}, \P_{\hat{\theta}_{n}})].
   \end{align*}
   We can now use Le Cam's method. Let $\P_{g(\theta, X)} := \mathcal{N}(\theta, \sigma^2)$ and let $\theta_0 = D$ and $\theta_1 = - D$ for some $D > 0$. Denote the joint likelihood across $n$ observations $(x_i,y_i)_{1:n}$ as $\P_\theta^{\times n}$. Then we have that:
   \begin{align*}
       KL(\P^{\times n}_{\theta_0} ||  \P^{\times n}_{\theta_0}) &= n KL(\P_X^0 \times \mathcal{N}(D, \sigma^2) || \P_X^0 \times \mathcal{N}(-D, \sigma^2)) \\
       &= n KL(\mathcal{N}(D, \sigma^2) || \mathcal{N}(-D, \sigma^2)) \\
       &= \frac{2n D^2}{\sigma^2}
   \end{align*}
where the first equality follows from the chain rule for the KL divergence, the second equality follows from the fact that under this construction $X$ and $Y$ are independent and hence by log properties:
\begin{align*}
    &KL(\P_X^0 \times \mathcal{N}(D, \sigma^2) || \P_X^0 \times \mathcal{N}(-D, \sigma^2)) \\&= KL(\P_X^0, \P_X^0) + KL(\mathcal{N}(D, \sigma^2) || \mathcal{N}(-D, \sigma^2))\\
    &= KL(\mathcal{N}(D, \sigma^2) || \mathcal{N}(-D, \sigma^2))
\end{align*}
and the third equality follows from the closed form expression of KL between Gaussian distributions. By Pinsker's inequality \citep[see e.g.][Lemma 2.5]{tsybakov2008nonparametric}, we hence have that:
\begin{align*}
    TV(\P^{\times n}_{\theta_0},  \P^{\times n}_{\theta_0}) \leq \sqrt{ \frac{1}{2} KL(\P^{\times n}_{\theta_0} ||  \P^{\times n}_{\theta_0})} = \sqrt{n} \frac{D}{\sigma}. 
\end{align*}
Now since $D$ is arbitrary, let us choose $D := \frac{c \sigma}{\sqrt{n}}$ for some $c > 0$. We then have $TV(\P^{\times n}_{\theta_0},  \P^{\times n}_{\theta_0}) \leq c$. 
Moreover, by Lemma \ref{lem:e9}, we have that:
\begin{align*}
    \MMD^2_k(\P_{\theta_0},  \P_{\theta_1}) &= \MMD^2_k(\P_X^0 \times \mathcal{N}(D, \sigma^2), \P_X^0 \times \mathcal{N}(-D, \sigma^2)) \\
    &= \E_{x,x^\prime \sim \P_X^0}[k_X(x,x^\prime) \MMD_{k_Y}^2(\mathcal{N}(D, \sigma^2), \mathcal{N}(-D, \sigma^2))] \\
    &= \E_{x,x^\prime \sim \P_X^0}[k_X(x,x^\prime)] \MMD_{k_Y}^2(\mathcal{N}(D, \sigma^2), \mathcal{N}(-D, \sigma^2)).
\end{align*}
For the second term, we can use the closed form expression of the MMD for Gaussians with a Gaussian kernel as in the proofs of Corollaries \ref{cor-rbf-gauss} and \ref{cor-rbf-gauss-class} to obtain:
\begin{align*}
    \MMD^2_{k_Y}(\mathcal{N}(D, \sigma^2), \mathcal{N}(-D, \sigma^2)) = 2 \frac{l_Y}{\sqrt{l_Y^2 + 2 \sigma^2}} \left(1 - \exp\left(- \frac{2 D^2}{l_Y^2 + 2 \sigma^2} \right) \right).
\end{align*}
Let $t := \frac{2 D^2}{l_Y^2 + 2 \sigma^2} = \frac{2 c^2 \sigma^2}{n(l_Y^2 + 2 \sigma^2)} $. Then if we choose $c > 0$ such that $t \in [0,1]$ we have that $1 - \exp(-t) \geq \frac{t}{2}$ and hence:
\begin{align*}
    \MMD_{k_Y}(\mathcal{N}(D, \sigma^2), \mathcal{N}(-D, \sigma^2)) \geq \sqrt{\frac{2 l_Y}{(l_Y^2 + 2 \sigma^2)^{3/2}}} \frac{c \sigma}{\sqrt{n}}
\end{align*}
Now define $m_X := \sqrt{\E_{x,x^\prime \sim \P_X^0}[k_X(x,x^\prime)]}$ which is independent of $n$. We then have:
\begin{align*}
    \MMD_k(\P_{\theta_0}, \P_{\theta_1}) \geq \frac{1}{\sqrt{n}} \left(m_X c \sigma \sqrt{\frac{2 l_Y}{(l_Y^2 + 2 \sigma^2)^{3/2}}} \right). 
\end{align*}
By Le Cam's method \citep[see][Lemma 15.9]{Wainwright_2019}, we have that:
\begin{align*}
    \inf_{\hat{\theta}_n} \sup_{\substack{{\P^0_{\theta^\star}}}} \E_{(x_i,y_i)_{1:n} \sim \P^0_{\theta^\star}}[\MMD_k(\P^0_{\theta^\star}, \P_{\hat{\theta}_n})] &\geq \frac{MMD(\P_{\theta_0}, \P_{\theta_1})}{4} (1 - TV(\P^{\times n}_{\theta_0}, \P^{\times n}_{\theta_1})) \\
    & \geq \frac{1}{\sqrt{n}} \left(m_X c \sigma \sqrt{\frac{2 l_Y}{(l_Y^2 + 2 \sigma^2)^{3/2}}} \right) (1 - c).
\end{align*}
\end{proof}

\section{Experimental setup \& additional results} \label{app-exp-dets}
In this section, we provide details of our experimental setup in Sections \ref{sec:sim} and \ref{sec:stu} as well as additional experimental results. 
\subsection{Experimental details}
In the Linear Regression example of Section \ref{sec-exp-lin-reg}, for the Robust-MEM (TLS) algorithm We set $T = 100$ and $c = 1$ and vary $\sigma^2_{\nu}$ while keeping $\sigma_{\epsilon}$ fixed at 0.5 throughout the experiment. A pseudo-code of the Robust-MEM algorithm can be found in Alg. \ref{alg:TLS_posterior_bootstrap}.
\begin{algorithm}[t] 
\SetAlgoLined
\SetKwInOut{Input}{input}
 \Input{$(w_i,y_i)_{1:n}$, $T$, $B$, $c$, $(\F_{w_i})_{1:n}$}
 \For{$j\gets1$ \KwTo $B$}{
 \For{$i\gets1$ \KwTo $n$}{
 Sample $\tilde{x}_{(1:T)}^{(i,j)} \simiid \F_{w_i}$ \\
 Sample $\xi_{1:T+1}^{(i,j)} \sim \Dir\left(\frac{c}{T}, \dots, \frac{c}{T},1\right)$. \\
 Set $\P^{(i,j)} := \sum_{k=1}^T \xi^{(i,j)}_k \delta_{\tilde{x}^{(i,j)}_k} + \xi^{(i,j)}_{T+1} \delta_{w_i}$ %\\
 }
 Calculate $\theta^{(j)} = \theta^{\star}_{TLS}(\P^{(1,j)}, \dots, \P^{(n,j)})$ as in (\ref{eq-tls-obj2})
 }
  \Return Posterior bootstrap sample $\theta^{(1:B)}$
    \caption{Robust-MEM with the TLS}
    \label{alg:TLS_posterior_bootstrap}
\end{algorithm}

For all applications of the Romust-MEM (MMD) method, we used automatic differentiation with the Adam optimizer 
\cite{kingma2014adam} 
provided by \texttt{JAX}, to minimize the MMD objective which is approximated via a U-statistic as in \cite{gretton2012kernel} using independent samples. In particular, for i.i.d samples $\{x_i\}_{i=1}^n \simiid \P$ and $\{y_{i}\}_{i=1}^m \simiid \Q$ the squared MMD between $\P, \Q$ can be approximated as:
\begin{align*}
    &\widehat{\MMD}^2_k(\P, \Q) = \frac{1}{n(n-1)} \sum_{i=1}^n \sum_{j \neq i}^n k(x_i, x_j) + \frac{1}{m(m-1)} \sum_{i=1}^m \sum_{j \neq i}^m k(y_i, y_j) \\
    &\quad \quad - \frac{2}{nm} \sum_{i=1}^n \sum_{j=1}^m k(x_i, y_j).
\end{align*}
In practice, we take $n=m$ to be equal to the number of observations. To initialize the optimisation several different approaches can be taken. In the nonlinear regression simulation, we sample uniformly at random from $[-1,4]$ at each bootstrap iteration for each regression coefficient. One could also initialize from the nonlinear Least Squares (NLS) estimate however adding randomness to the initialization procedure can aid numerical stability and increase posterior uncertainty leading to better coverage probabilities. In the real-world California Test Scores study we combine the above methods by first finding the NLS estimates and then sampling uniformly from intervals centered around these estimates. Initialization techniques for more complex models and a higher number of parameters, such as the Mental Health study, is discussed in Section \ref{app-bcr-model}.

Furthermore, the RBF kernel was used for both $k_X$ and $k_Y$ throughout our experiments. For the nonlinear regression simulation, the length scales for $k_X$ and $k_Y$ were set to 100 in the Classical error case, whereas the median heuristic \cite{gretton2012kernel} was used in the Berkson case as well as the California Test Scores study. For the Mental Health study, the length scales were 10 and 100 respectively. These values were chosen to ensure the loss and gradient do not vanish during optimisation. A pseudo-code of the Robust-MEM algorithm with the MMD-based loss function is outlined in \ref{alg:TLS_posterior_bootstrap_MMD}.

\begin{algorithm}[t] 
\SetAlgoLined
\SetKwInOut{Input}{input}
 \Input{$(w_i,y_i)_{1:n}$, $T$, $B$, $c$, $(\F_{w_i})_{1:n}$}
 \For{$j\gets1$ \KwTo $B$}{
 \For{$i\gets1$ \KwTo $n$}{
  Sample $\tilde{x}_{1:T}^{(i,j)} \simiid \F_{w_i}$ \\
 Sample $\xi_{1:(T+1)}^{(i,j)} \sim \Dir\left(\frac{c}{T}, \dots, \frac{c}{T}, 1\right)$. \\
 Set $\P^{(i,j)} := \sum_{k=1}^T \xi^{(i,j)}_k \delta_{\tilde{x}^{(i,j)}_k} + \xi^{(i,j)}_{T+1} \delta_{w_i}$ 
 }
 Obtain $\theta^{(j)} = \theta^\ast_{\text{MMD}}(\P^{(1,j)}, \dots, \P^{(n,j)})$ as in \ref{eq-mmd-obj} using numerical optimisation. 
 }
 \Return{Posterior bootstrap sample $\theta^{(1:B)}$}
 \caption{Robust-MEM with the MMD}
 \label{alg:TLS_posterior_bootstrap_MMD}
\end{algorithm}

\subsection{Additional results for the Linear Regression case} \label{app:sec-lin-wellsp}
We now consider the Linear regression simulation from Section \ref{sec-exp-lin-reg}. We first give additional results on the misspecified prior cases considered in the main text in Table \ref{tab:class-lin}. We observe that the Robust-MEM (TLS) method remains robust to ME even under a misspecified prior variance. In the absence of ME ($\sigma_{\nu}^2 = 1e-06$), OLS performs best whereas for small ME variance ($\sigma_{\nu}^2 = 0.5$) TLS achieves the smallest MSE with all methods performing well. The suggested method has the lowest MSE for higher values of the variance. 
\begin{table}
\small
\caption{MSE over 100 simulations with associated std. in the Linear Regression example with Classical ME (with a misspecified prior) for an increasing ME variance ($\sigma^2_{\nu}$).}
\label{tab:class-lin}
\begin{tabular}{llllll}
\hline
 Method & & \multicolumn{4}{c}{ME variance}  \\
   &  &  1e-06 &  0.5 & 1.0 & 2.0 \\
\hline
\multirow[t]{2}{*}{OLS} & $	\theta_1$ & 0.063 (0.083) & 1.632 (1.232) & 16.945 (7.153) & 103.85 (21.629) \\
 & $	\theta_2$ & 0.001 (0.001) & 0.016 (0.012) & 0.17 (0.071) & 1.035 (0.209) \\
\multirow[t]{2}{*}{R-MEM} & $	\theta_1$ & 1.618 (0.128) & 1.585 (0.161) & 1.464 (0.222) & 1.011 (0.335) \\
 (TLS) & $	\theta_2$ & 0.013 (0.002) & 0.013 (0.003) & 0.012 (0.004) & 0.011 (0.008) \\
\multirow[t]{2}{*}{SIMEX} & $	\theta_1$ & 0.07 (0.073) & 1.127 (0.681) & 15.404 (4.712) & 99.524 (16.695) \\
 & $	\theta_2$ & 0.001 (0.001) & 0.011 (0.007) & 0.154 (0.046) & 0.993 (0.159) \\
\multirow[t]{2}{*}{TLS} & $	\theta_1$ & 0.118 (0.156) & 0.313 (0.561) & 1.47 (1.948) & 14.871 (11.412) \\
 & $	\theta_2$ & 0.001 (0.002) & 0.003 (0.006) & 0.015 (0.019) & 0.15 (0.112) \\
 \hline
\end{tabular}
\end{table}

We now present the results of the Linear Regression simulation in the well-specified case, i.e. when $\F_N = \F_N^0 := N(0, \sigma_{\nu}^0)$ and $\F_X = \P_X^0 := N(0, 1)$, for $\sigma_{\nu}^0 \in \{10^{-6}, 0.5, 1, 2\}$. The linear regression fit is illustrated in Fig. \ref{fig:class-lin} for all methods. As expected, the SIMEX method performs better than in the misspecified case whereas all other methods perform similarly. 
\begin{figure}[t]
    \center
\includegraphics[width=\textwidth]{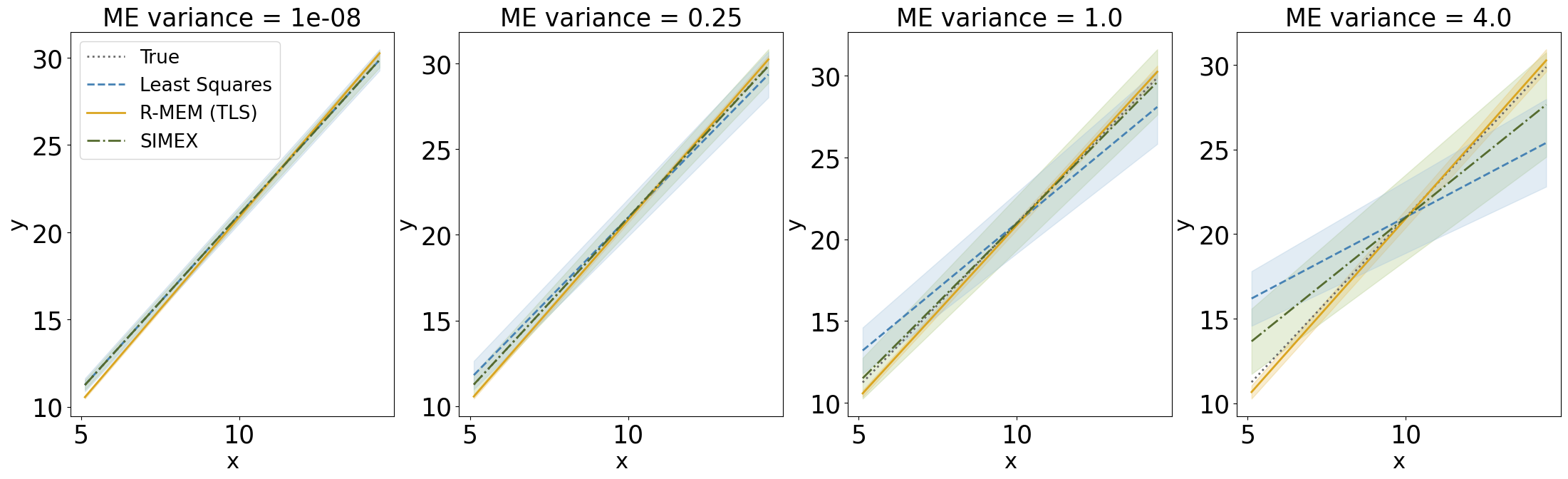}
 \caption{Model fit for linear regression in the well-specified case for an increasing ME variance (from left to right) over 100 simulations. The true model fit is shown along with our method (R-MEM (MMD)), nonlinear least squares (Least Squares) and SIMEX.}
  \label{fig:class-lin-wellsp}
\end{figure}
We also report the corresponding MSE and standard deviations of the parameter estimates over 100 runs are reported in Table \ref{tab:class-lin-wellsp}. We observe a decline in performance compared to the misspecified scenario for the first parameter in the Robust-MEM case with smaller ME variance values. Nonetheless, the method maintains its robustness, even when the ME variance equals 4.
\begin{table}
\caption{Mean Squared Error with associated standard deviation for $\theta = (\theta_1, \theta_2)$ in the Linear Regression example with Classical ME and well-specified prior. Results are reported for an increasing ME variance ($\sigma^2_{\nu}$) in the observed data, over 100 simulations.}
\label{tab:class-lin-wellsp}
\begin{tabular}{llllll}
\hline
 Method & & \multicolumn{4}{c}{ME variance}  \\
   &  &  1e-06 &  0.5 & 1.0 & 2.0 \\
\hline
\multirow[t]{2}{*}{LS} & $	\theta_1$ & 0.063 (0.083) & 1.632 (1.232) & 16.945 (7.153) & 103.85 (21.629) \\
 & $	\theta_2$ & 0.001 (0.001) & 0.016 (0.012) & 0.17 (0.071) & 1.035 (0.209) \\
\multirow[t]{2}{*}{Robust-MEM} & $	\theta_1$ & 1.615 (0.129) & 1.603 (0.166) & 1.548 (0.273) & 1.319 (0.573) \\
(TLS) & $	\theta_2$ & 0.013 (0.002) & 0.013 (0.003) & 0.013 (0.004) & 0.012 (0.006) \\
\multirow[t]{2}{*}{SIMEX} & $	\theta_1$ & 0.032 (0.046) & 0.165 (0.213) & 0.984 (1.004) & 26.326 (12.685) \\
 & $	\theta_2$ & 0.0 (0.0) & 0.002 (0.002) & 0.01 (0.01) & 0.264 (0.125) \\
\multirow[t]{2}{*}{TLS} & $	\theta_1$ & 0.118 (0.156) & 0.313 (0.561) & 1.47 (1.948) & 14.871 (11.412) \\
 & $	\theta_2$ & 0.001 (0.002) & 0.003 (0.006) & 0.015 (0.019) & 0.15 (0.112) \\
 \hline
\end{tabular}
\end{table}

\subsection{Additional results for the Nonlinear Regression case with Berkson error} \label{app:sec-nonlin-berk-wellsp}
Similarly, we consider the example of nonlinear regression with Berkson ME from Section \ref{sec-nonlin-reg} in the well-specified case where both the true and the prior ME distributions are Gaussian with increasing variance. Results are reported for 100 simulations in Fig. \ref{fig:berk-nonlin-wellsp}.
\begin{figure}[t]
    \center
\includegraphics[width=\textwidth]{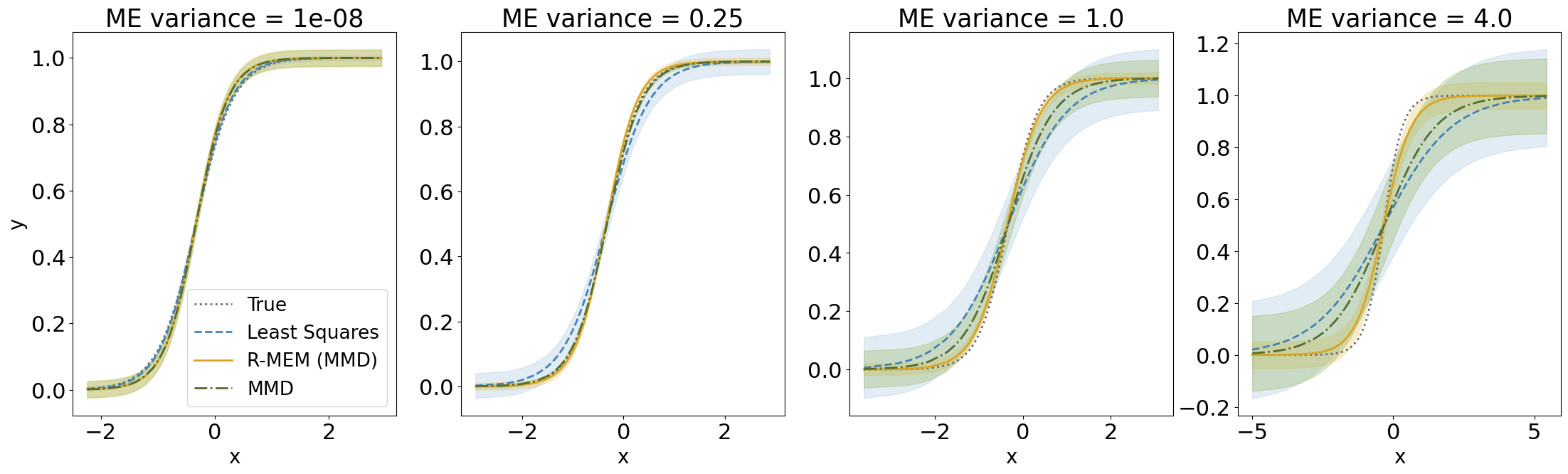}
 \caption{Model fit for the sigmoid curve regression function with Berkson ME for an increasing ME variance (from left to right) over 100 simulations with a well-specified prior centering measure. The true model fit is shown along with our method (R-MEM (MMD)), nonlinear least squares (Least Squares) and minimum MMD estimator (MMD).}
  \label{fig:berk-nonlin-wellsp}
\end{figure}
We observe very similar behavior to the misspecified case. For ME variances of $\leq 0.25$, all methods exhibit comparable performance. However, as the variance escalates, the Robust-MEM (MMD) approach stands out as the most resilient to ME effects. This is also seen by observing the associated MSE and their variances in Table \ref{tab:berk-mse-wellsp}. 
\begin{table}[t]
\caption{Mean Squared Error with associated standard deviation for $\theta = (\theta_1, \theta_2)$ in the Nonlinear Regression example with Berkson ME and a wellspecified prior centering measure. Results are reported for an increasing ME variance ($\sigma^2_{\nu}$) in the observed data, over 100 simulations.}
\label{tab:berk-mse-wellsp}
\begin{tabular}{llllll}
\hline
Method &  & \multicolumn{4}{c}{ME variance}  \\
&  & $1e-08$ & $0.25$ & $1.0$ & $4.0$ \\
\hline
\multirow[t]{2}{*}{MMD} & $	\theta_1$ & 0.245 (0.503) & 0.192 (0.503) & 0.213 (0.219) & 0.49 (0.335) \\
 & $	\theta_2$ & 1.335 (2.106) & 0.953 (1.386) & 1.526 (1.224) & 3.846 (1.524) \\
\multirow[t]{2}{*}{Robust-MEM} & $	\theta_1$ & 0.139 (0.24) & 0.127 (0.22) & 0.13 (0.179) & 0.244 (0.279) \\
(MMD) & $	\theta_2$ & 0.442 (0.346) & 0.238 (0.278) & 0.4 (0.613) & 1.264 (1.223) \\
\hline
\end{tabular}
\end{table}
The coverage probabilities in the well-specified case are presented in Table \ref{tab:berk-cov-wellsp}. 
\begin{table}
\centering
\caption{Coverage probability for the Robust-MEM (MMD) method in the Berkson error model with a well-specified prior distribution for the ME. The coverage probability is calculated as the frequency over 100 simulations that the 90\% posterior credible interval included the true parameter value.}
\label{tab:berk-cov-wellsp}
\begin{tabular}{lrrrr}
\hline
& \multicolumn{4}{c}{ME variance} \\ 
& 1e-08 & 0.25 & 1 & 4 \\
\hline
$	\theta_1$ & 0.73 & 0.84 & 0.92 & 0.890000 \\
$	\theta_2$ & 1 & 1 & 0.99 & 0.99 \\
\hline
\end{tabular}
\end{table}

\subsection{Results for Nonlinear Regression with Classical error} \label{app-nonlin-class}
We now evaluate the nonlinear regression example for the Classical ME case. 
We report results for both the well-specified and misspecified cases of prior centering measures. In the well-specified case in (\ref{eq-bayesian-post}), we set $\F_X := N(0,1)$ and $\F_{W \given X} := N(x, \sigma^2_{\nu})$ where $\sigma^2_\nu$ is the true ME variance. The prior centering measure $\F_{X \given w_i}$ then has density $f_{X \given w_i}(x \given w_i) \propto f_{W \given X}(w_i \given x)f_X(x)$ and is again a Normal distribution. 
In the misspecified case, we set $\F_X := N(0,1)$ and $\F_{W \given X} := N(x, 0.2)$. 
The obtained mean regression fit over 100 simulations is illustrated in Fig. \ref{fig:class-nonlin} along with the estimates obtained from standard nonlinear least squares and the SIMEX procedure. 
\begin{figure}[t]
     \centering
     \begin{subfigure}[b]{\textwidth}
         \centering
         \includegraphics[width=\textwidth]{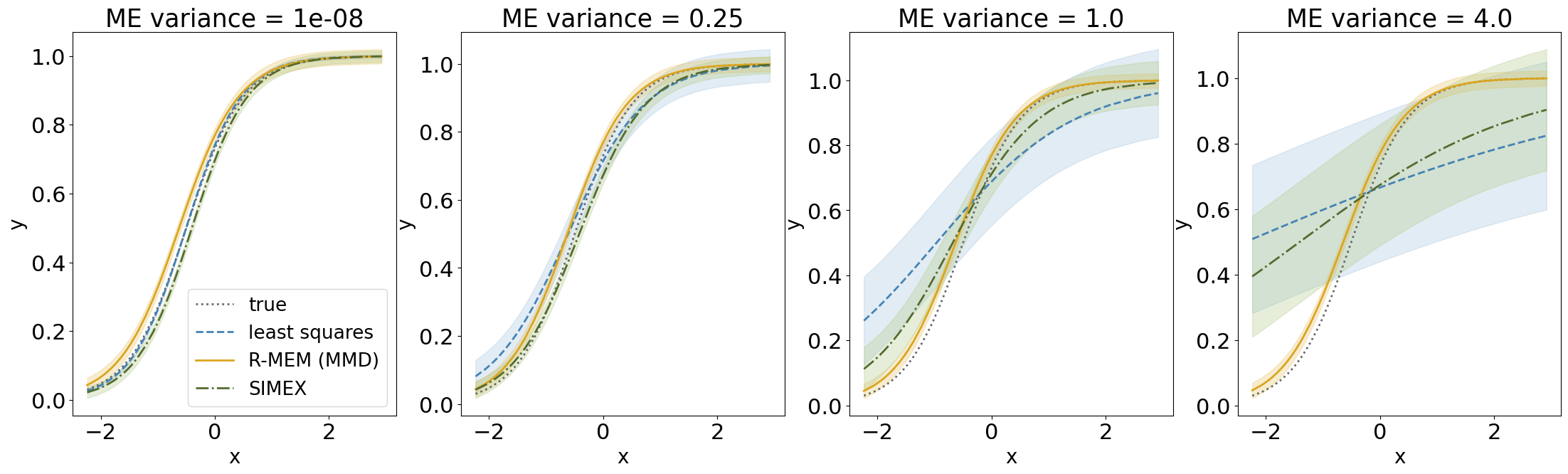}
         \caption{ME variance is correctly specified}
         \label{fig:class-nonlin-wellsp}
     \end{subfigure}
     \hfill
     \begin{subfigure}[b]{\textwidth}
         \centering
         \includegraphics[width=\textwidth]{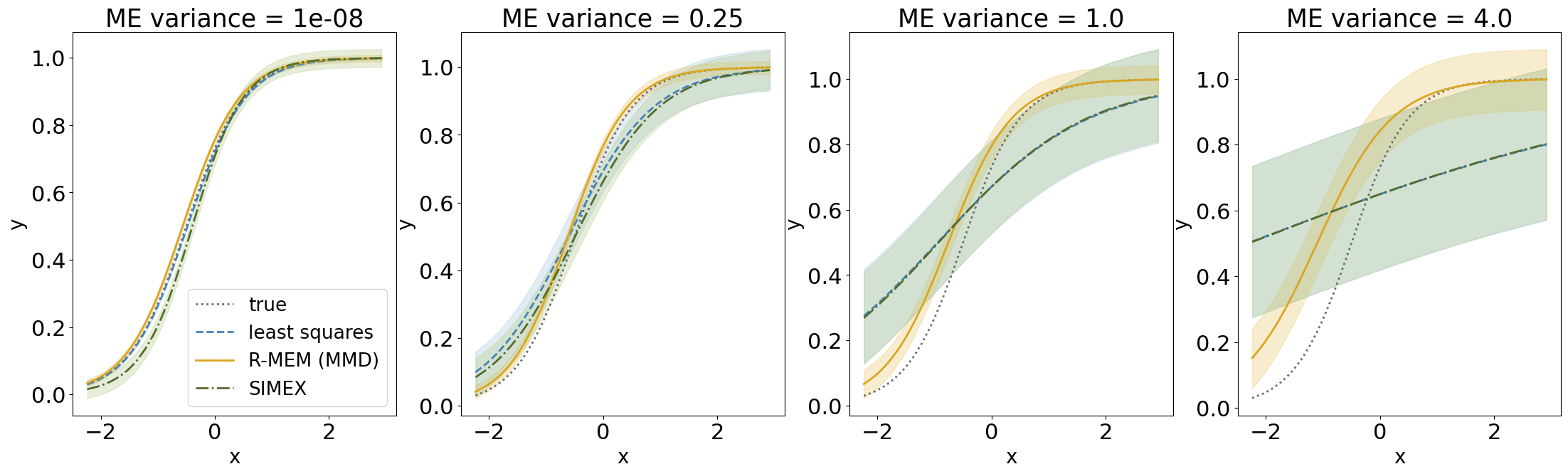}
         \caption{ME variance is misspecified}
         \label{fig:class-nonlin-misp}
     \end{subfigure}
        \caption{Model fit for the sigmoid curve regression function for Classical ME for an increasing ME variance (from left to right) over 100 simulations. The true model fit is shown along with our method (R-MEM (MMD)), nonlinear least squares (Least Squares) and SIMEX.}
        \label{fig:class-nonlin}
\end{figure}
In the SIMEX procedure we use the same a priori values for the ME variance, i.e. the true ME variance in the well-specified case and 0.2 in the misspecified. 
We observe that the Robust-MEM (MMD) remains more robust to ME in both the well-specified and misspecified cases for all degrees of ME. We further assess the average MSE and associated standard deviation over 100 simulations. In Table \ref{tab:class-nonlin-wellsp} we observe the case of well-specified ME variance. As expected, in the absence of ME on the covariates, Robust-MEM (MMD) performs slightly worse however remains the most robust over increasing values of ME. 
\begin{table}[ht!]
\centering
\caption{Mean Squared Error with associated standard deviation for $\theta = (\theta_1, \theta_2)$ in the nonlinear regression example with Classical ME of \textit{known} ME variance (well-specified case). Results are reported for an increasing ME variance ($\sigma^2_{\nu}$) in the observed data, over 100 simulations.}
\label{tab:class-nonlin-wellsp}
\begin{tabular}{llllll}
\hline
 Method &  & \multicolumn{4}{c}{ME variance}  \\
 &  & $1e-08$ & $0.25$ & $1.0$ & $4.0$ \\
\hline
\multirow[t]{2}{*}{LS} & $\theta_1$ & 0.096 (0.179) & 0.077 (0.093) & 0.118 (0.201) & 0.139 (0.127) \\
 & $	\theta_2$ & 0.264 (0.551) & 0.385 (0.314) & 1.455 (0.491) & 2.926 (0.36) \\
\multirow[t]{2}{*}{R-MEM (MMD)} & $\theta_1$ & 0.1 (0.117) & 0.101 (0.118) & 0.106 (0.122) & 0.113 (0.128) \\
 & $	\theta_2$ & 0.059 (0.078) & 0.056 (0.079) & 0.058 (0.087) & 0.06 (0.09) \\
\multirow[t]{2}{*}{SIMEX} & $\theta_1$ & 0.096 (0.179) & 0.111 (0.202) & 0.261 (1.562) & 0.13 (0.127) \\
 & $	\theta_2$ & 0.264 (0.552) & 0.323 (0.663) & 0.748 (1.019) & 2.244 (0.576) \\
 \hline
\end{tabular}
\end{table}
To further study the effect of a misspecified prior we look at the same metrics in the misspecified case in Table \ref{tab:class-nonlin-misp}. The impact of a misspecified ME variance appears minimal when the variance is small, indicating a lesser degree of misspecification. However, it becomes increasingly significant as the variance grows larger. Robust-MEM maintains a smaller overall MSE while also achieving satisfying coverage as reported in Table \ref{tab:class-nonlin-cover} for both well-specified and misspecified cases. 
\begin{table}
\centering
\caption{Mean Squared Error with associated standard deviation for $\theta = (\theta_1, \theta_2)$ in the nonlinear regression example with Classical ME of \textit{unknown} ME variance (misspecified case). Results are reported for an increasing ME variance ($\sigma^2_{\nu}$) in the observed data, over 100 simulations.}
\label{tab:class-nonlin-misp}
\begin{tabular}{llllll}
\hline
 Method &  & \multicolumn{4}{c}{ME variance}  \\
 &  & $1e-08$ & $0.25$ & $1.0$ & $4.0$ \\
\hline
\multirow[t]{2}{*}{LS} & $	\theta_1$ & 0.096 (0.179) & 0.077 (0.093) & 0.118 (0.202) & 0.139 (0.127) \\
 & $	\theta_2$ & 0.264 (0.551) & 0.385 (0.314) & 1.455 (0.491) & 2.926 (0.36) \\
\multirow[t]{2}{*}{Robust-MEM} & $	\theta_1$ & 0.1 (0.12) & 0.137 (0.154) & 0.252 (0.236) & 0.642 (0.408) \\
 (MMD) & $	\theta_2$ & 0.052 (0.063) & 0.066 (0.08) & 0.121 (0.134) & 0.358 (0.295) \\
\multirow[t]{2}{*}{SIMEX} & $	\theta_1$ & 0.121 (0.27) & 0.081 (0.093) & 0.131 (0.31) & 0.139 (0.128) \\
 & $	\theta_2$ & 0.46 (1.275) & 0.337 (0.305) & 1.414 (0.49) & 2.919 (0.362) \\
 \hline
\end{tabular}
\end{table}
\begin{table}
\centering
\caption{Coverage probability of the Robust-MEM (MMD) method for the nonlinear regression example with Classical ME. Coverage probability was calculated as the frequency over 100 simulations that the 90\% posterior credible interval included the true parameter value. \textit{Well-specified} refers to the case of known ME variance whereas \textit{misspecified} refers to the misspecification of the ME variance.}
\label{tab:class-nonlin-cover}
\begin{tabular}{lrrrrr}
\hline
& & \multicolumn{4}{c}{ME variance} \\
& & 1e-08 & 0.5 & 1.0 & 2.0 \\
\hline
\multirow[t]{2}{*}{Well-specified} &
$\theta_1$ & 0.96 & 0.97 & 0.98 & 0.98 \\
& $\theta_2$ & 1 & 1 & 1 & 1 \\
\multirow[t]{2}{*}{Misspecified} &
$\theta_1$ & 0.96 & 0.97 & 0.98 & 0.97 \\
& $\theta_2$ & 1 & 1 & 1 & 1 \\
\hline
\end{tabular}
\end{table}
\subsection{Mental Health Study} \label{app-bcr-model}
In this section we give details of the model used in \cite{harezlak2018semiparametric} for the treatment effect estimation in the Mental Health study which we also adopt in Section \ref{sec-exp-berry}. 

For each observation $i = 1, \dots, n_{placebo}$ in the placebo group we model
\begin{talign*}
    g^0(\theta^0,x) &= \beta_0 + \beta_1 x + \sum_{k=1}^K u_{0k} z_k(x), \\
    y_i \given x_i &\sim \P_{g^0(\theta^0,x_i)} = N(g^0(\theta^0,x_i), \sigma^2_{\epsilon})
\end{talign*}
for parameters $\theta^0 = (\beta_0, \beta_1, u_{01}, \dots, u_{0K})$ and spline basis $\{z_k\}_{k=1}^K$. 
Similarly, for each patient $i = 1, \dots, n_{treatment}$ in the treatment group we model 
\begin{talign*}
    g^1(\theta^1,x) &= \beta_0 + \beta_0^{\text{drug}} + (\beta_1 + \beta_1^{\text{drug}}) x + \sum_{k=1}^K u_{1k} z_k(x), \\
    y_i \given x_i &\sim \P_{g^1(\theta^1,x_i)} = N(g^1(\theta^1,x_i), \sigma^2_{\epsilon}).
\end{talign*}
We define the parameter of interest $\theta = (\beta, u_0, u_1)$ where
\begin{talign*}
    \beta = (\beta_0, \beta_1, \beta_0^{\text{drug}}, \beta_1^{\text{drug}}), u_0 = (u_{01}, \dots, u_{0K}) \quad \text{and} \quad  u_1 = (u_{11}, \dots, u_{1K}).
\end{talign*}
Let $n = n_{treatment} + n_{placebo}$ and define $Y = (y_1, \dots, y_n)$ and the design matrices 
\begin{talign*}
   x = &\begin{pmatrix} 1 & x_1 & 1-I_1 & (1-I_1)x_1 \\
   \vdots & \vdots & \vdots & \vdots \\
   1 & x_n & 1-I_n & (1-I_n)x_n
   \end{pmatrix}, \quad 
   Z_0 = \begin{pmatrix} I_1 z_1 (x_1) & \dots & I_1 z_K(x_1) \\
   \vdots & \ddots & \vdots \\
    I_n z_1 (x_n) & \dots & I_n z_K(x_n)
   \end{pmatrix}  \\
      &\quad \quad Z_1 = \begin{pmatrix} (1-I_1) z_1 (x_1) & \dots & (1-I_1) z_K(x_1) \\
   \vdots & \ddots & \vdots \\
    (1-I_n) z_1 (x_n) & \dots & (1-I_n) z_K(x_n)
   \end{pmatrix}.
\end{talign*}
Based on this notation, the full Bayesian model proposed in \cite{harezlak2018semiparametric} is 
\begin{talign*}
    Y \given x \sim \P_{g(\theta,x)} := N(x \beta + Z_0 u_0 + Z_1 u_1, \sigma_{\epsilon}^2 I).
\end{talign*}
Based on this model, the average treatment effect given a baseline score $x$ is:
\begin{talign*}
    ATE(x) \approx \beta_0^{\text{drug}} + \beta_1^{\text{drug}}x + \sum_{k=1}^K (u_{1k} - u_{0k}) z_k(x). 
\end{talign*}
We use the same spline basis as the one used in the \texttt{HRW} package by \cite{harezlak2018semiparametric}, namely the penalized spline defined by 
\begin{talign*}
    z_k(x) = (x - \kappa_k)_{+}
\end{talign*}
where $x_{+} = \max(x,0)$ and the knots $\{\kappa_k\}_{k=1}^K$ are approximately equally spaced points with respect to the quantiles of the $\{x_i\}_{i=1}^n$. 

In this application we set $c = 1$ and $T = 100$ in the Robust-MEM algorithm. For the optimisation step we used a learning rate of $\eta = 0.001$ and we employ a random restart method for the initialization of the optimisation: 100 points are sampled uniformly at random and the MMD loss is calculated using the observed data. The point with the smallest loss is used as the initial point and the resulting parameter estimate with the smallest loss during optimisation is picked. The uniform intervals are set to $[0,2]$ for the regression parameters $\beta$ and $[-0.1,0.1]$ for the spline parameters $u_0, u_1$. Posterior samples are obtained from 200 bootstrap iterations. Of course, when the optimisation objective is non-convex, there is usually a trade-off between coverage capability and MSE of the posterior estimates or intervals depending on the initialization of the optimisation. A wider uniform interval usually leads to higher uncertainty in our posterior beliefs which can exhibit better coverage probability of estimated intervals. On the contrary, if one is interested in the MSE of the posterior estimations then one can perform the above procedure more than 100 times to get a better initial point of the optimisation or use a narrower uniform distribution based on prior beliefs or historic data.

\end{document}